\documentclass{USC-Thesis}
\usepackage{amsmath}
\usepackage{amsfonts}
\usepackage{amssymb}
\usepackage{amsthm}
\usepackage{graphicx}
\usepackage{setspace}
\usepackage{epsfig}
\usepackage[titles]{tocloft}
\usepackage{lettrine}
\usepackage{multirow}
\input{Qcircuit}

\newtheorem{proposition}{Proposition}[section]

\theoremstyle{remark}

\renewcommand{\thechapter}{\arabic{chapter}}

\def\CA{{\cal A}}

\def\CD{{\cal D}}
\def\CE{{\cal E}}

\def\CH{{\cal H}}
\def\CI{{\cal J}}

\def\CK{{\cal K}}

\def\CN{{\cal N}}
\def\CO{{\cal O}}

\def\CQ{{\cal Q}}
\def\CR{{\cal R}}
\def\CS{{\cal S}}
\def\CT{{\cal T}}
\def\CU{{\cal U}}
\def\CV{{\cal V}}
\def\CW{{\cal W}}

\def\CY{{\cal Y}}

\def\bra#1{{\langle#1|}}
\def\ket#1{{|#1\rangle}}
\def\bracket#1#2{{\langle#1|#2\rangle}}

\begin{document}

\title{QUANTUM STEGANOGRAPHY AND QUANTUM ERROR-CORRECTION}

\author{Bilal A. Shaw}
\major{COMPUTER SCIENCE}
\month{August}
\year{2010}

\setcounter{secnumdepth}{4}
\setcounter{tocdepth}{4}

\maketitle

\dedication
    \begin{quote}
        \raggedleft {\em To my parents}\\
        \raggedleft {\em Afzal and Yaqoot Shaw,}\\
        \raggedleft {\em and to my sister}\\
        \raggedleft {\em Husna Shaw.}
    \end{quote}

\topmatter{Acknowledgments}
There are many people in my life without whom I would not have been able to find my way to the United States of America, and then to the University of Southern California (USC), and without whom I would not be writing my doctoral dissertation.  The path toward my doctoral thesis has been a long and arduous one.  My intellectual adventure would not have been possible without the undying support, will, generosity, passion, intellectual curiosity, and tough-love of my advisor.  Thank-you Todd for giving me a chance, and believing in me.

Had I not picked up \textit{Surely you are joking Mr.~Feynman} as a young boy growing up in Kashmir, I probably would never have enjoyed physics, let alone quantum physics.  Feynman's exploits, adventures, and passion for physics led me to pursue my doctorate.  His legacy continues to inspire me to strive for excellence in science, while at the same time reminding me to not take life too seriously.  Thank-you Dick!

I first got a taste of quantum information science (QIS) when I attended Computing Beyond Silicon Summer School (CBSSS) at Caltech in the summer of 2002, which I later helped to organize in the summer of 2004.  I had the pleasure of hearing John Preskill's talk on quantum fault-tolerance, and Isaac Chuang's talk on quantum computing.  At Caltech I met a brilliant young post-doctoral fellow by the name of Dave Bacon with whom I talked at length about quantum error-correction, and with whom I continued talking over the ensuing years.  I thank Dave for being my mentor.  I was fortunate that USC decided to delve into QIS by hiring Todd.  I was searching for Ph.D. advisors at the time, and having talked to Dave, I knew that Todd would be a perfect fit for me.  I am happy and fortunate in knowing that this came out to be true in the end.

I thank Leonard Adleman who was my advisor during my undergraduate and Masters years at USC.  He taught me the importance of mathematical rigor and maturity, doing good scientific research, and the art of giving an excellent talk.  With him I learned first-hand that we live in an increasingly multi-disciplinary world.  It is not enough to be just a mathematician, or just a biologist, or a physicist.  The new world demands young graduates to move seamlessly and effortlessly between various fields.  It demands us to be both a frog and an eagle; to be able to wallow in the mud at the edges of a pond, and dive into its depth without fear, and at the same time to be able to soar like an eagle above mountain-tops, with courage.  He taught me not to be afraid to work on ``cool'' ideas.  I also thank Nickolas Chelyapov who during this time taught me many molecular biology techniques and the discipline required of a good scientist.  I thank one of my best friends Paul Rothemund who gave me his full support while I was Len's student.

I thank Daniel Lidar who taught me quantum error-correction (QEC) and the theory of open quantum systems at USC.  The first paper that I wrote was a direct result of a project he gave his students to complete in his QEC course.  Later when I gave a talk based on this project at the American Physical Society's March meeting in 2008 in New Orleans, I won the best student presentation award.  I am deeply grateful to Daniel for being an excellent educator and scientist.

I thank Igor Devetak whom I met in 2005, and who later taught me quantum Shannon theory.  I thank him for his friendship and for the various discussions he and I had on information theory which helped shape my own intellectual landscape.  I also thank him for being part of Todd's research group meetings where I was among the first students to witness the birth of entanglement-assisted quantum codes.  This coding scheme ultimately allowed researchers to import all of classical coding theory into quantum information science.

I thank my friend, collaborator, and colleague Mark Wilde.  His passion for this field helped kindle my own when I felt as if the fire within me was starting to die.  His brilliance and ambition has been a source of inspiration to me.  I would also like to thank him while I was in New Orleans attending the American Physical Society's March meeting, for his parents opened the doors of their heart and their home to me.

I thank Min-Hsiu Hsieh, Hari Krovi, Jose Gonzalez, Stephen Jordan, Zhicheng Luo, Ognyan Oreshkov, Shesha Raghunathan, and Martin Varbanov with whom I have had various discussions on quantum information science in particular, and life in general.  Ognyan was a big help to me in completing my project on quantum error-correction.  I had an excellent time working on one of Igor Devetak's class projects with Hari.

I owe many thanks to Gerrielyn Ramos, Milly Montenegro, Anita Fung, Tim Boston, Lizsl DeLeon, Steve Schrader, and Jun Camacho who did an excellent job in taking care of my paperwork and making sure that I was on top of non-research related things.  I would not have been able to complete my research at USC without the support that I received from the computer science department at USC.  The teaching assistantships helped me to survive in an overly expensive city.  I thank Michael Crowley and David Wilczynski for taking me on as their teaching assistant.

I would like to thank the folks outside my USC circle who supported me with their kind words, thoughts, prayers and blessings.  At face value the Ph.D. is an original piece of research, but to me it has been more than mere tinkering with ideas, writing papers and giving talks.  To me it has signified my intellectual, spiritual, and emotional growth as a human being and as a scientist.  I am indebted to my yoga instructors Travis Eliot and Jennifer Pastiloff for introducing me to yoga.  This has been my salvation when I needed to unhook and unplug from the vibe of Los Angeles and USC.  I am deeply thankful to Los Angeles that nurtured my pursuits outside of academia.  I have basked in her beautiful sunshine, enjoyed walks on her beaches and hiked in her mountains.  I thank all the talented baristas at Peet's Coffee and Tea in Brentwood who provided me with free coffee while I was doing research.

Lastly I thank my friends Iftikhar Burhanuddin, Tony Barnstone, Matthew Harmon, and Jonathan Trachtman.  Matt and J. T. expanded my musical horizons and allowed me to explore musical genres that helped me during the many hours I spent writing my thesis and my papers.  My escape from science has been writing poetry and to that I owe much thanks to Tony.  I thank Ifti for all the interesting conversations we had on algorithms.  I also thank my uncle Bashir Shaw and my aunt Khursheed Shaw who provided a roof over my head for three years when I was an undergraduate at Whittier College.  I thank all of my teachers in Burn Hall High School, Kashmir, Saint Columba's High School, New Delhi, and Whittier College, California.

In the end none of this would have been possible without the unconditional and unwavering love of my father Afzal, my mother Yaqoot, my sister Husna, and the memory of my days as a child growing up in Kashmir.

\tableofcontents

\listoftables

\listoffigures

\topmatter{Abstract}
Quantum error-correcting codes have been the cornerstone of research in quantum information science (QIS) for more than a decade.  Without their conception, quantum computers would be a footnote in the history of science.  When researchers embraced the idea that we live in a world where the effects of a noisy environment cannot completely be stripped away from the operations of a quantum computer, the natural way forward was to think about importing classical coding theory into the quantum arena to give birth to quantum error-correcting codes which could help in mitigating the debilitating effects of decoherence on quantum data.  We first talk about the six-qubit quantum error-correcting code and show its connections to entanglement-assisted error-correcting coding theory and then to subsystem codes.  This code bridges the gap between the five-qubit (perfect) and Steane codes.  We discuss two methods to encode one qubit into six physical
qubits. Each of the two examples corrects an arbitrary single-qubit error. The first example is a degenerate six-qubit quantum error-correcting code. We explicitly provide the stabilizer generators, encoding circuits, codewords, logical Pauli
operators, and logical CNOT operator for this code. We also show how to convert this code into a non-trivial subsystem code
that saturates the subsystem Singleton bound.  We then prove that a six-qubit code without entanglement assistance cannot
simultaneously possess a Calderbank-Shor-Steane (CSS) stabilizer and correct an arbitrary single-qubit error.  A corollary of this
result is that the Steane seven-qubit code is the smallest single-error correcting CSS code.  Our second example is the construction of
a non-degenerate six-qubit CSS entanglement-assisted code. This code uses one bit of entanglement (an \textit{ebit}) shared between the sender (Alice) and the receiver (Bob) and corrects an arbitrary single-qubit error. The code we obtain is globally equivalent to the Steane seven-qubit code and thus corrects an arbitrary error on the receiver's half of the ebit as well.  We prove that this code is the smallest code with a CSS structure that uses only one ebit and corrects an arbitrary single-qubit error on the sender's side. We discuss the advantages and disadvantages for each of the two codes.

In the second half of this thesis we explore the yet uncharted and relatively undiscovered area of quantum steganography.  Steganography is the process of hiding secret information by embedding it in an ``innocent'' message.  We present protocols for hiding quantum information in a codeword of a quantum error-correcting code passing through a channel.  Using either a shared classical secret key or shared entanglement Alice disguises her information as errors in the channel.  Bob can retrieve the hidden information, but an eavesdropper (Eve) with the power to monitor the channel, but without the secret key, cannot distinguish the message from channel noise.  We analyze how difficult it is for Eve to detect the presence of secret messages, and estimate rates of steganographic communication and secret key consumption for certain protocols.  We also provide an example of how Alice hides quantum information in the perfect code when the underlying channel between Bob and her is the depolarizing channel.  Using this scheme Alice can hide up to four stego-qubits.


\mainmatter

\chapter{Introduction}

\begin{saying}
When I heard the learn'd astronomer;\\
When the proofs, the figures, were ranged in columns before me;\\
When I was shown the charts and the diagrams, to add, divide, and\\
measure them;\\
When I, sitting, heard the astronomer, where he lectured with much\\
applause in the lecture-room,\\
How soon, unaccountable, I became tired and sick;\\
Till rising and gliding out, I wander'd off by myself,\\
In the mystical moist night-air, and from time to time,\\
Look'd up in perfect silence at the stars.\\
---\textit{Walt Whitman}
\end{saying}

\section{A Brief History of Quantum Information Science}
\lettrine{A}lthough I am not sure what exactly was going through James Clerk Maxwell's (1831---1879) mind when he formulated his now famous laws of electromagnetism, I can certainly imagine that he had no idea that a few decades into the future Thomas Alva Edison residing in the New World would construct several devices operating on the physics that he [Maxwell] worked to unify, namely electricity and magnetism.  Moreover, he probably had no idea that in the future universities would award doctoral degrees in electrical engineering, and that there would be entire departments dedicated to understanding and engineering this force of Nature.  Maxwell's work on electromagnetism has often been dubbed as the ``second great unification in physics'' second only to Sir Isaac Newton's first unification.

Around the time that Edison was designing electrical machines in his copious notebooks and subsequently inventing them, the world of physics was undergoing a revolution of its own. This revolution started when physicists observed phenomena (paradoxes) in Nature that could not be explained through classical physics.  These included black-body radiation, the photoelectric effect, the stability of atoms (ultraviolet catastrophe), and matter waves (particle-wave duality).  Through the efforts of brilliant scientists such as Einstein, Bohr, Schr\"{o}dinger, Heisenberg, Pauli, Born and still others, these paradoxes were ultimately resolved.  Their efforts led to the creation of the new field of quantum mechanics, when finally researchers realized that the sub-atomic world is intrinsically probabilistic, and that there is an underlying uncertainty in simultaneously measuring certain real-world observables such as the momentum and position of a particle.  While the physicists were thinking about the atomic world, Claude Shannon ushered the field of information theory with his seminal paper of 1948 in which he established the source coding and channel coding theorems~\cite{shannon1948}.  The former gives us the ultimate limit on the compressibility of information while the latter tells us that given a noisy channel, one can encode information in such a way that the receiver of information can reliably decode it.  Of course this assumes that the designer of error-correcting codes has certain knowledge of the noise in the channel.  Without Shannon's seminal work we would not have the modern communication devices that have become so ubiquitous, and which we take for-granted every day.  While Claude Shannon was busy at Bell Labs solidifying his ideas on information theory and mathematical logic, Alan Turing was laying the foundation of modern computing theory.  Turing wanted to know if there was a universal machine that could compute solutions to computable functions.  All modern day computers from laptops to supercomputers are based on the simple notion that given a machine with enough memory, a finite set of symbols and a read/write head that can store its own state, one can perform complex computations.  Alan Turing's simple yet deep insight is a testament to the power of great ideas and their ability to transform society (for the better)~\cite{turing1937}.

In science, unification occurs when ideas in its various fields have had enough time to mature.  Several researchers before Maxwell had separately worked on electricity and magnetism before he unified these two pictures into the electromagnetic wave.  Among these researchers the notable ones were Andr\'{e}-Marie Ampere, Benjamin Franklin, J. J. Thompson, Oliver Heaviside, Joseph Henry, Nikola Tesla, and still others.  Similarly, it took a few decades for quantum physics, information theory, and computer science to mature before these seemingly separate fields could be united into quantum information science.

The story of this unification started when Richard Feynman and others asked the question of whether a computer operating on the laws of quantum physics might give an appreciable speed-up in simulating quantum physics problems that may be hard to solve on a classical computer~\cite{feynman1982,benioff1980,manin1980}.  This question was answered by Seth Lloyd in~\cite{lloyd1996}.  In the 1970s Charles Bennett (IBM) showed that any computation can in principle be done reversibly, based on earlier work by Rolf Landauer in the 1960s~\cite{landauer1982, bennett1973}.  Rolf Launder showed that if a computer erases a single bit of information, the amount of energy dissipated into the environment is at least $k_{B}T\ln{2}$, where $k_{B}$ is Boltzmann's constant and $T$ is the temperature of the environment of the computer~\cite{mikeandike2000}.  If a computation can be performed reversibly then no information is erased and hence there is no dissipation of energy.    To realize the theory of quantum computation it was crucial to make this connection between reversibility and energy dissipation.  This was because we evolve closed quantum system via unitary operators which are themselves reversible.

One of the hallmarks of quantum mechanics is that one can write superpositions of various quantum states.  One can then imagine exploiting these states in a parallel fashion via operators that transform these states.  Researchers quickly realized that naive applications of quantum parallelism added nothing to the power of the computer.  But in 1985, David Deutsch found a clever algorithm that exploited quantum parallelism indirectly to solve a problem more efficiently than any possible classical computer~\cite{deutsch1985}.  The problem was artificial, but it was the first example where a quantum computer was shown to be more powerful than a classical computer.  Richard Jozsa and David Deustch later generalized this problem and they showed that there was an exponential gap in query complexity between a quantum and a classical computer.  We were starting to see a grand unification of quantum mechanics with computer science.

Meanwhile in 1984 Charles Bennett and Gilles Brassard found a clever way to exploit the properties of quantum mechanics for information processing~\cite{bb84}.  They were able to exploit the uncertainty principle as a way to distribute a cryptographic key to two parties, Alice and Bob, with perfect security.  In their set-up single quanta were used to send bits of the cryptographic key in one of two possible complementary variables.  If Eve tried to intercept the quantum bits and measure them, this would automatically disturb them in such a way that it could always be detected by Alice and Bob.  Upon learning this, they would abort the protocol as it had been compromised.  This scheme was aptly termed as the BB84 (Bennett-Brassard) protocol for quantum key distribution.  We often hear people refer to the latter as quantum cryptography.  We must make it clear that this is a misnomer.  All that this protocol does is distribute keys securely to two parties.  Alice and Bob can subsequently use their shared keys to send cryptographic messages to each other.  We should also mention that this is the only quantum information protocol which has been realized commercially.  Charles Bennett and his collaborators found even more interesting quantum communication protocols.  In particular they found one where a sender could send an unknown quantum state to a receiver using a classical channel and consuming shared entanglement between the sender and the receiver.  We call this quantum communication protocol as quantum teleportation~\cite{prl.70.1895}.  Bennett along with Stephen Weisner engineered a protocol by which a sender could transmit classical information by consuming shared entanglement and utilizing a quantum channel between the sender and the receiver~\cite{prl.69.2881}.  We call this quantum communication protocol as super-dense coding.

Even before Charles Bennett and company started exploiting the power of quantum physics to perform information theoretical tasks, Stephen Wiesner had indicated in a paper \textit{Conjugate Coding} that it was impossible to perfectly duplicate or clone two or more unknown nonorthogonal quantum states.  This paper which was submitted to \textit{IEEE Transactions on Information Theory} in 1970 was summarily rejected as it was written in a language incomprehensible to computer scientists, information theorists, and electrical engineers. His work was finally published in its original form in~\cite{wiesner1983} in 1983 in \textit{Association of Computing Machinery, Special Interest Group in Algorithms and Computation Theory}.  In the same paper Wiesner also suggested that his no-cloning theorem could also enable the possibility of money whose authenticity would be guaranteed by the laws of quantum physics.  The no-cloning theorem places a severe constraint on what can be achieved by uniting ideas from quantum mechanics and information theory.  This was one of the earliest concerns raised by researchers on the viability of an operational quantum computer.

There was an uneasiness in the research community that perhaps a quantum computer would only be good for solving toy problems and, moreover, clever randomized algorithms could probably beat these toy quantum algorithms.  In addition because of the sensitive nature of quantum systems and their susceptibility to noisy environments researchers believed that any quantum state that one would prepare as input to a quantum computer would rapidly decohere because of unwanted interactions with its environment.  Another major objection to a viable quantum computer was the quantum mechanical principle that measurement destroys a quantum state.  As we shall see in subsequent chapters that despite these objections researchers were able to forge ahead and develop quantum error-correcting codes that reduced the negative effects of decoherence and got around the measurement problem.

In 1994 Daniel Simon discovered a quantum algorithm for another toy problem that provided an exponential separation in query complexity over a classical randomized algorithm~\cite{simon1994}.  The stage was being set for a major breakthrough.  In the same year a mathematician by the name of Peter Shor at AT\&T Bell Labs discovered two polynomial time quantum algorithms~\cite{shor1994,shor1997}.  The first algorithm could factor integers in polynomial time, whereas the second algorithm gave an exponential speed-up over the discrete logarithm problem.  Both factoring and the discrete logarithm problems are extremely hard problems.  The reason why Shor's algorithm turned several heads was because firstly it became apparent that quantum algorithms were good at something other than solving mere toy problems, and secondly the security of well known cryptosystems such as RSA~\cite{rsa1978} relies on the hardness of factoring while the security of ElGamal cryptosystem~\cite{elgamal1985} relies on the hardness of the discrete logarithm problem.  Researchers in cryptography have been working on post-quantum cryptography.  In the event when large-scale quantum computers become operational, which public-key schemes would be resistant to quantum computer attacks?

The explosion in the number of papers being published in quantum information science is extraordinary.  In just a few years the field has matured to the point where it is no longer possible for a single researcher to keep track of all the papers being published.  In the theoretical arena alone researchers are working on quantum error-correcting codes, fault-tolerant quantum computing, quantum algorithms, quantum Shannon theory, quantum entanglement, quantum complexity theory, and open quantum systems.  There is of course a parallel effort to realize various quantum computing and quantum Shannon schemes experimentally through solid-state physics, nanophotonics, and Bose-Einstein condensates.

So perhaps in the near future when quantum information science comes of age, the term ``quantum computer'' will be a household name, just as PCs or MACs are today.  Perhaps in the future QIS will not be a loose conglomerate of such academic departments as electrical engineering, computer science, physics, and mathematics, but will be a department in its own right.  Maybe we will routinely hear quantum states being teleported over quantum networks, or high school students running quantum simulations on their personal quantum computers, or big business deals being carried out using quantum cryptography and quantum steganography.  Perhaps then the efforts and passion of all the hardworking researchers in this field will be truly realized on a societal level; a triumph.

\section{Thesis Organization}
This thesis is a major contribution to quantum error-correction theory and quantum Shannon theory.  This thesis is essentially divided into two parts.  In the first half we present the theory of quantum error-correction starting with quantum stabilizer codes in the next chapter.  The six-qubit code bridges the gap between the five-qubit (perfect code) and the Steane codes.  The five-qubit has been extensively studied in~\cite{mikeandike2000} because it was the first example of the smallest non-degenerate code that could correct an arbitrary single-qubit error~\cite{perfectcode1996}.  The seven-qubit code, popularly known as the Steane code after its discoverer Andrew Steane, has also been studied extensively because it lends itself nicely to a fault-tolerant design of universal quantum gates~\cite{steane1996}.  In fact if we open~\cite{mikeandike2000} we find that the authors provide a detailed study of these two codes, but none for the six-qubit code.  We hope that when quantum computers become operational, the six-qubit code will provide a middle ground for encoding quantum data.

In the second half of this thesis we present the theory of quantum steganography.  It came as a surprise to us that this sub-area of quantum Shannon theory had barely been explored given that classical steganography has been around since Simmons first posed the problem in 1983 using information-theoretic terminology~\cite{simmons1983}.  In order to develop quantum steganography we needed to develop intuition about classical steganography using classical error-correcting codes.  In Chapter~\ref{chap:classicalsteganography} we present classical ideas on information hiding.  After we develop our intuition for classical steganography we present our full formalism for quantum steganography in Chapter~\ref{chap:quantumsteg}.  We provide the rate of steganographic information and give specific examples of how to hide qubits using quantum error-correcting codes transmitted over the binary-symmetric and depolarizing channels.  In Chapter~\ref{chap:hidingqinfo} we show how one can use the second protocol detailed in Chapter~\ref{chap:quantumsteg} to hide up to four stego-qubits in the perfect code.  For a specific encoding we give the optimal number of stego-qubits that Alice can send to Bob over the depolarizing channel as a function of the error-rate.  In the same chapter we also show the amount of key consumed by Alice and Bob to realize this protocol.  We hope that this work will blossom into a sub-field of quantum information science just as quantum cryptography has in recent years.

The mathematics behind these ideas is simple and straightforward, and all that one requires to fully comprehend this thesis is mathematical maturity, knowledge of college level linear algebra, and familiarity with Dirac notation used in quantum mechanics.  We end the thesis with closing remarks in which we outline some further work that one can pursue in quantum steganography.


\chapter{Quantum Stabilizer Codes}
\label{chap:stabilizercodes}
\begin{saying}
It is not knowledge, but the act of learning, \\
not possession but the act of getting there,\\
which grants the greatest enjoyment.\\
---\textit{Carl Friedrich Gauss}
\end{saying}
\lettrine{T}he stabilizer formalism is to quantum error-correction what additive codes are to classical linear error-correction.  It is a powerful mathematical framework based on group theory.  Using this formalism a quantum-code engineer can import classical codes using the CSS construction (dual-containing codes)~\cite{steane1996,PRA.54.1098.1996,mikeandike2000}, and design quantum codes with good error-correcting properties and good rates.  More recently the stabilizer formalism has been used to design graph states~\cite{graph2005}.  It is used in cluster-state quantum computation~\cite{cluster2001}, and the theory of entanglement~\cite{arx2004fattal}.  It has also been extended to entanglement-assisted quantum codes~\cite{science2006brun}.  Much of the stabilizer formalism was worked out be Daniel Gottesman in his brilliant thesis~\cite{thesis97gottesman}.  For more details we direct the reader to his thesis.

In this chapter we review the stabilizer formalism for block-codes, the steps involved in error-correction, the Clifford operations that one uses to construct the encoding unitary circuits for these codes.  We assume the reader is familiar with elementary group theory at the undergraduate college level~\cite{booklang}.  In this chapter we have used Mark Wilde's thesis for some of the exposition~\cite{thesismarkwilde}.

\section{The Stabilizer Formalism}
One of the hallmarks of the stabilizer formalism is that it gives us a compact representation of a quantum error-correcting code.  Instead of keeping track of the amplitudes of complicated superpositions of quantum codewords, we can just keep track of the \textit{stabilizer generators} of the code (which are linear in the number of qubits).  As an example consider the following GHZ state:
\begin{equation}
\label{sec:stabformalism_ghz}
\ket{\psi}_{GHZ} \equiv \frac{\ket{000}+\ket{111}}{\sqrt{2}}.
\end{equation}
The following operators leave the GHZ state unchanged:
\begin{eqnarray}
\label{sec:ghzstabilizers}
I_{1}I_{2}I_{3} & \equiv & I \otimes I \otimes I, \nonumber \\
Z_{1}Z_{2} & \equiv & Z \otimes Z \otimes I, \nonumber \\
Z_{1}Z_{3} & \equiv & Z \otimes I \otimes Z, \nonumber \\
Z_{2}Z_{3} & \equiv & I \otimes Z \otimes Z, \nonumber \\
X_{1}X_{2}X_{3} & \equiv & X \otimes X \otimes X. \nonumber
\end{eqnarray}
We say that the GHZ state is fixed or stabilized by the operators above.  The main idea in the stabilizer formalism is that one can just think of these operators and their transformation as they are acted upon by error operators, encoding unitary operators, and measurement operators.  The stabilizer is a subgroup of the general Pauli group $\mathbb{G}_{n}$ on $n$ qubits.  More specifically $\mathbb{G}_{n}$ is an $n$-fold tensor product of the matrices from the Pauli set $\mathbb{P} = \left\{I, X, Y, Z\right\}$:
\[
I\equiv%
\begin{bmatrix}
1 & 0\\
0 & 1
\end{bmatrix}
,\ X\equiv%
\begin{bmatrix}
0 & 1\\
1 & 0
\end{bmatrix}
,\ Y\equiv%
\begin{bmatrix}
0 & -i\\
i & 0
\end{bmatrix}
,\ Z\equiv%
\begin{bmatrix}
1 & 0\\
0 & -1
\end{bmatrix}
.
\]
We write the set $\mathbb{G}_{n}$ as:
\begin{equation}
\mathbb{G}_{n}=\left\{  e^{i\phi}A_{1}\otimes\cdots\otimes A_{n}:\forall j\in\left\{
1,\ldots,n\right\}  ,\ \ A_{j}\in\mathbb{P},\ \ \phi\in\left\{  0,\pi/2,\pi
,3\pi/2\right\}  \right\}  .
\end{equation}
An example of an element from $\mathbb{G}_{n}$ is $X_{1} \equiv X \otimes I \otimes \ldots \otimes I$, where $X$ acts on the first qubit, while the $n-1$ identity operators act on rest of the $n-1$ qubits.  While writing these $n$-fold tensor operators we will often omit the tensor product sign, and just indicate the operator as $X_{1}$ or as $XI \ldots I$.  We need the factors $\pm 1, \pm i$ in front of the group elements so that the group is closed under multiplication.  Three interesting properties of the elements of $\mathbb{G}_{n}$ are that they either commute or anticommute with each other, each has eigenvalue +1 or -1, and the product of each element with itself is the $n$-fold identity operator.  The stabilizer is a subgroup $\mathbb{S}_{n}$ of $\mathbb{G}_{n}$ such that all the elements in it commute with each other.  Each element of $\mathbb{S}_{n}$ stabilizes some number of quantum states.  The intersection of all these subspaces fixed by each element of $\mathbb{S}_{n}$ constitutes the codespace.  In group theory one can reconstitute the entire group by the underlying group operation between some specific elements of the group called generators.  So this way we get a very compact representation of the group.  Instead of listing all the elements of the group, we just specify its generators.  As an example, consider the the group $\mathbb{G}_{3}$ whose members 3-fold tensor product of Pauli operators and whose members operate on three qubits.  This group has $4^{3}*4 = 256$ members.  We enumerate some of the members of this group below:
\[
\mathbb{G}_{3} \equiv \left\{\pm III, \pm iIII, \pm XII, \pm iXII, \ldots , \pm ZZZ, \pm iZZZ \right\}.
\]
Let $\mathbb{S}_{3} = \left\{III, ZZI, IZZ, ZIZ \right\}$ be a subgroup of $\mathbb{G}_{3}$.  Each member of $\mathbb{S}_{3}$ fixes a subspace. $ZZI$ fixes the states $\ket{000}, \ket{001}, \ket{110}$, and $\ket{111}$.  Similarly, $IZZ$ stabilizes the states $\ket{000}, \ket{011}, \ket{100}$, and $\ket{111}$.  $III$ stabilizes the states $\ket{000}$ and $\ket{111}$, while $ZIZ$ stabilizes $\ket{000}, \ket{010}, \ket{101}$, and $\ket{111}$.  The codespace is the intersection of the subspaces spanned by all the members of $\mathbb{S}_{3}$.  So the codespace $\mathbb{C}_{3} = \left\{\ket{000},\ket{111}\right\}$.  Notice that we did not really need to write all the four members of $\mathbb{S}_{3}$.  We can represent $\mathbb{S}_{3}$ compactly through its generators $ZZI$ and $IZZ$.  We write the generator set as $G_{\mathbb{S}_{3}} = \left\langle ZZI, IZZ\right\rangle$.  With these two generators we can generate $\mathbb{S}_{3}$ as follows: $III = (ZZI)(ZZI)$ and $ZIZ = (ZZI)(IZZ)$.  So when we give a description of a code in terms of its stabilizers, we mean the stabilizer generators.  The stabilizer generators form an independent set, in the sense that no generator is a product of two or more generators.  In the section on the conversion of stabilizer generators to binary matrices, this translates into the rows of the matrix being linearly independent---no row in the matrix can be obtained by the product of two or more rows.  The stabilizer for a quantum error-correcting code cannot contain the element $-I$.  To see this assume that the codespace is not the trivial subspace.  Let $\ket{\psi}$ be a state that is stabilized $-I$.  This implies that $-I\ket{\psi} = \ket{\psi}$, which would mean that $\ket{\psi}$ is the trivial state, and hence a trivial subspace which we assumed was not the case.  The stabilizer generators must all commute to constitute a valid codespace.  These stabilizer generators for quantum codes function as parity check matrices for classical linear codes.

An element $g_{1}\neq \pm I$ from the set $\mathbb{G}_{n}$ has two eigenspaces labeled by eigenvalue +1 and -1 of equal size $2^{n}/2 = 2^{n-1}$.  Now pick another element $g_{2}$ from $\mathbb{G}_{n}$ different from both $\pm I$ and $g_{1}$, but which commutes with $g_{1}$.  $g_{2}$ also has two eigenspaces of equal size $2^{n-1}$ labeled by eigenvalue +1 and -1.  But since $g_{1}$ and $g_{2}$ commute they have simultaneous eigenspaces.  So both operators can now stabilize quantum states in four eigenspaces which we label as $\left\{++,+-,-+,--\right\}$.  Each eigenspace has size $2^{n}/2^{2} = 2^{n-2}$.  In this way we can keep picking up elements from $\mathbb{G}_{n}$ and subdividing the $2^{n}$ dimensional Hilbert space into subspaces of equal sizes.  These subspaces are all orthogonal to each other.  So if we choose $n-k$ elements $g_{1}, g_{2}, \ldots , g_{n-k}$ from $\mathbb{G}_{n}$ such that they all commute with each other and none is equal to $\pm I$, then we end up dividing the Hilbert space into $2^{n-k}$ subspaces each of size $2^{k}$.  Call this set $\mathbb{S}_{n}$.  Clearly $\mathbb{S}_{n} \subset \mathbb{G}_{n}$.  We can label these subspaces as $\left\{++ \ldots +, -+ \ldots +, \ldots , -- \ldots -\right\}$.  So each subspace can contain a Hilbert space of $k$ qubits.  We choose a special subspace represented by the label $++ \ldots +$ which is the simultaneous +1 eigenspace of the elements of $\mathbb{S}_{n}$.  We call this subspace the codespace and it represents a quantum error-correcting code that encodes $k$ logical qubits into $n$ physical qubits, and we indicate this as $[[n,k]]$.  There is nothing special about the simultaneous +1 eigenspace.  We could have easily chosen any of the other subspaces as well.  So in addition to the commutation property, the stabilizer generators must all have eigenvalues equal to +1.  When an error occurs it moves the codeword out of the simultaneous +1 eigenspace into one of the orthogonal subspaces.  It is this orthogonality property of these subspaces that allows us to detect the error and safely recover from it.  We can multiply the stabilizer generators and obtain an equivalent representation.  This operation would be equivalent to adding bases together in vector notation.  In the next section we list the steps of how a quantum error-correcting code operates.

\begin{figure}
[ptb]
\begin{center}
\includegraphics[
natheight=3.386600in,
natwidth=8.973300in,
height=1.9614in,
width=5.1742in
]%
{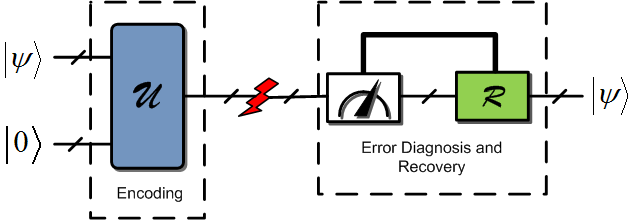}%
\caption{Steps in quantum error-correction. Thin lines denote quantum
information and thick lines denote classical information. Slanted bars denote
multiple qubits. Alice the sender of quantum information encodes a multi-qubit state $\left\vert
\psi\right\rangle $ using ancilla qubits $\left\vert
0\right\rangle $. She sends the encoded state over a noisy quantum channel.
Bob on receiving the quantum information performs multi-qubit measurements to extract classical information about
the errors. He finally performs a recovery operation $R$\ to reverse the
channel error.  This recovery operation is composed of firstly correcting the errors and secondly
decoding the corrected state.}%
\label{fig:stabilizer-code}%
\end{center}
\end{figure}

\section{Quantum Error-Correction}
\label{sec:quantum error-correction}
An $[[n,k]]$ quantum error-correcting code operates when:
\begin{enumerate}
	\item Alice encodes her $k$ qubits of quantum information along with $n-k$ ancilla qubits $\ket{0}^{\otimes n-k}$ into the general multi-qubit state $\ket{\psi}$.  She encodes this into the simultaneous +1-eigenspace of the stabilizer generators $g_{1}, g_{2}, \ldots , g_{n-k}$ as shown in Figure~\ref{fig:stabilizer-code}.  In this way, Alice has encoded $k$ logical qubits into $n$ physical qubits.
	\item Alice sends her $n$-qubit state (quantum codeword) over a noisy quantum channel to Bob, using the channel $n$ times.
	\item Bob on receiving the codeword performs a measurement of the $n-k$ stabilizer generators to determine if an error has occurred.  Recall that the stabilizer generators are composed from Pauli matrices which are Hermitian and hence valid observables on which one can perform measurements and obtain classical information.  In the ideal case when no errors occur, the measurement will give Bob a bit-string of $n-k$ +1 values.  When an error occurs the codeword moves out of the +1-eigenspace to one of the the other $2^{n-k}-1$ subspaces.  When a measurement of a stabilizer generators gives Bob a value of -1, then he knows exactly which error has occurred as this new bit-string differs from the all +1 bit-string, and he can undo that error by reapplying the error on that particular qubit.  These measurements performed by Bob do not disturb the original quantum state.  Moreover it suffices to correct only a discrete set of errors $\CE$ even though errors in quantum information can be continuous~\cite{PRA.52.R2493.1995, mikeandike2000}.  $\CE$ is a subset of the general Pauli group $\mathbb{G}_{n}$.  Bob can uniquely identify which error has occurred if the following condition holds:
\begin{equation}
\forall E_{a},E_{b}\in\mathcal{E\ \ \ \ \ }\exists\ g_{i}\in\mathbb{S}_{n}%
:\left\{  g_{i},E_{a}^{\dag}E_{b}\right\}  =0\text{ or }E_{a}^{\dag}E_{b}%
\in\mathbb{S}_{n}.
\end{equation}
What the above condition is stating is that as long as an error anticommutes with one of the stabilizer generators, Bob ought to be able to detect this.  This is the active correction part.  The second condition says that there can be ``error'' operators that commute with the stabilizer generators.  These are undetectable by Bob, however, Bob does not care about them.  This is the passive correction part.
\item If Bob can diagnose which error has occurred, he can correct for it during the recovery procedure as shown in Figure~\ref{fig:stabilizer-code}.  Once he obtains the corrected codeword, he can decode it to obtain the original $k$ information qubits.
\end{enumerate}
\begin{figure}
[ptb]
\begin{center}
\includegraphics[
natheight=3.386600in,
natwidth=8.973300in,
height=1.9614in,
width=5.1742in
]%
{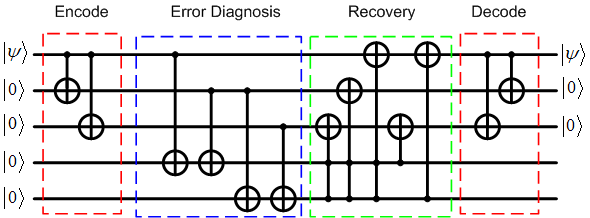}
\caption{Alice encodes her one qubit of quantum information with the help of two ancilla qubits
using the encoding unitary shown in the dashed red box in the left side of the figure.  She sends the codeword to Bob who
uses the error-diagnosis circuit shown in the dashed blue box.  The result of this measurement is projected on to the two ancilla qubits at the bottom.  The information from these ancillas (syndrome result) is then used to apply the recovery procedure.  Bob can now finally decode the quantum state to obtain the original one qubit of quantum information.}
\label{fig:threebitcode}
\end{center}
\end{figure}
We now give a specific example of the $[[3,1]]$ stabilizer code that can correct a single bit-flip error.  With this code Alice can encode one logical qubit into three physical qubits.  So $n = 3$ and $k = 1$.  This code has $n-k = 3-1 = 2$ stabilizer generators. These are $g_{1} = ZZI$ and $g_{2} = IZZ$.  The codewords for this code are $\ket{000}$ and $\ket{111}$.  Figure~\ref{fig:threebitcode} shows the
\begin{table}\centering
\begin{tabular}
[c]{|c||c||c|}
\hline
Target & Conjugation & Target \\ \hline \hline
I & $HIH^{\dagger}$ & $I$ \\ \hline
X & $HXH^{\dagger}$ & $Z$ \\ \hline
Y & $HYH^{\dagger}$ & $-Y$ \\ \hline
Z & $HZH^{\dagger}$ & $X$ \\
\hline
\end{tabular}
\caption{Hadamard transformation of the Pauli matrices.}
\label{tbl:hadamardtable}
\end{table}
\begin{table}\centering
\begin{tabular}
[c]{|c||c||c|}
\hline
Target & Conjugation & Target \\ \hline \hline
I & $PIP^{\dagger}$ & $I$ \\ \hline
X & $PXP^{\dagger}$ & $Y$ \\ \hline
Y & $PYP^{\dagger}$ & $X$ \\ \hline
Z & $PZP^{\dagger}$ & $Z$ \\
\hline
\end{tabular}
\caption{Phase transformation of the Pauli matrices.}
\label{tbl:phasetable}
\end{table}
encoding unitary, the error diagnosis and recovery circuits as well as the decoding unitary.  The bit patterns corresponding to bit-flip errors are: $X_{1} = XII = \left\{-1,+1\right\}, X_{2} = IXI = \left\{-1, -1\right\}, X_{3} = IIX = \left\{+1, -1\right\}, I = III = \left\{+1, +1\right\}$.  Since the $[[3,1]]$ code is a very simple code, the encoding unitary circuits are sparse.  In a later chapter when we exploit quantum error-correcting codes to hide quantum information we will deal with the $[[7,1]]$ Steane code whose encoding unitary circuit consists of thirty-six quantum gates!
\section{Encoding Unitary}
The stabilizer generators also give us a way to generate the encoding circuit for any stabilizer code.  These unitaries belong to a special class called Clifford unitaries.  What makes them special is that when we conjugate an element from $\mathbb{G}_{n}$ with a Clifford unitary, it produces another element of $\mathbb{G}_{n}$.  Specifically, given a Clifford unitary $U$ and an element $G$ from $\mathbb{G}_{n}$, $UGU^{\dagger} = G'$, where $G'$ is an element of $\mathbb{G}_{n}$.  Gottesman showed in his thesis~\cite{thesis97gottesman} that CNOT, Hadamard, and Phase gates are sufficient to construct any Clifford unitary.  The CNOT acts on two qubits, whereas the Hadamard and Phase gates act on a single qubit.  In Section~\ref{sec:eaqeccircuit} of Chapter~\ref{chap:sixqubitcode}, we detail an algorithm that generates the encoding circuit for the six-qubit code.  The matrix corresponding to the Hadamard gate is:
\begin{equation}
H=\frac{1}{\sqrt{2}}%
\begin{bmatrix}
1 & 1\\
1 & -1
\end{bmatrix}.
\end{equation}
The matrix for the Phase gate is:
\begin{equation}
P=%
\begin{bmatrix}
1 & 0\\
0 & i
\end{bmatrix}.
\end{equation}
\begin{table}\centering
\begin{tabular}
[c]{|c|c||c||c|c|}
\hline
Target & Control & Conjugation & Target & Control\\ \hline \hline
$I$ & $I$ & $CNOT (I \otimes I) CNOT^{\dagger}$ & $I$ & $I$\\ \hline
$I$ & $X$ & $CNOT (I \otimes X) CNOT^{\dagger}$ & $X$ & $I$\\ \hline
$I$ & $Y$ & $CNOT (I \otimes Y) CNOT^{\dagger}$ & $X$ & $Y$\\ \hline
$I$ & $Z$ & $CNOT (I \otimes Z) CNOT^{\dagger}$ & $I$ & $Z$\\ \hline
$X$ & $I$ & $CNOT (X \otimes I) CNOT^{\dagger}$ & $X$ & $I$\\ \hline
$X$ & $X$ & $CNOT (X \otimes X) CNOT^{\dagger}$ & $I$ & $X$\\ \hline
$X$ & $Y$ & $CNOT (X \otimes Y) CNOT^{\dagger}$ & $I$ & $Y$\\ \hline
$X$ & $Z$ & $CNOT (X \otimes Z) CNOT^{\dagger}$ & $X$ & $Z$\\ \hline
$Y$ & $I$ & $CNOT (Y \otimes I) CNOT^{\dagger}$ & $Y$ & $Z$\\ \hline
$Y$ & $X$ & $CNOT (Y \otimes X) CNOT^{\dagger}$ & $Z$ & $Y$\\ \hline
$Y$ & $Y$ & $CNOT (Y \otimes Y) CNOT^{\dagger}$ & $-Z$ & $X$\\ \hline
$Y$ & $Z$ & $CNOT (Y \otimes Z) CNOT^{\dagger}$ & $Y$ & $I$\\ \hline
$Z$ & $I$ & $CNOT (Z \otimes I) CNOT^{\dagger}$ & $Z$ & $Z$\\ \hline
$Z$ & $X$ & $CNOT (Z \otimes X) CNOT^{\dagger}$ & $-Y$ & $Y$\\ \hline
$Z$ & $Y$ & $CNOT (Z \otimes Y) CNOT^{\dagger}$ & $Y$ & $X$\\ \hline
$Z$ & $Z$ & $CNOT (Z \otimes Z) CNOT^{\dagger}$ & $Z$ & $I$\\ \hline
\hline
\end{tabular}
\caption{CNOT transformation of two qubits.}
\label{tbl:cnottable}
\end{table}
The matrix for the CNOT is slightly more complicated:
\begin{equation}
\text{CNOT}=%
\begin{bmatrix}
1 & 0 & 0 & 0\\
0 & 1 & 0 & 0\\
0 & 0 & 0 & 1\\
0 & 0 & 1 & 0
\end{bmatrix}.
\end{equation}
The standard basis for the one qubit Pauli group $\mathbb{G}_{1}$ is $X$ and $Z$.  We list the transformation of this basis under conjugation with the Hadamard and Phase gates in Table~\ref{tbl:hadamardtable} and in Table~\ref{tbl:phasetable} respectively.  We list the conjugate transformation of two-qubit Pauli matrices with the $CNOT$ gate in Table~\ref{tbl:cnottable}.  These tables will be useful in the thesis later when we calculate the encoding of various error operators.
\section{Concluding Remarks}
The theory of error-correcting codes is vast and the number of papers published in this sub-field of quantum information science is unprecedented.  We only provided the mathematical framework necessary to understand the results of Chapter~\ref{chap:sixqubitcode}.  There is a further mathematical tool called the Pauli-to-binary isomorphism that is important in understanding and appreciating the results of the next chapter.  Instead of reproducing the tool here, we refer the reader to~\cite{arx2006brun}.

\chapter{The Six-Qubit Code}
\label{chap:sixqubitcode}
\begin{saying}
It matters if you just don't give up.\\
---\textit{Stephen Hawking}
\end{saying}

\section{Introduction}
\lettrine{I}t has been more than a decade since Peter Shor's seminal paper on quantum
error correction \cite{PRA.52.R2493.1995}.  He showed how to protect one qubit
against decoherence by encoding it into a subspace of a Hilbert space larger
than its own. For the first time, it was possible to think about quantum
computation from a practical standpoint.

Calderbank and Shor then provided asymptotic rates for the
existence of quantum error-correcting codes and gave upper bounds
for such rates \cite{PRA.54.1098.1996}.  They defined a quantum
error-correcting code as an isometric map that encodes $k$ qubits
into a subspace of the Hilbert space of $n$ qubits.  As long as
only $t$ or fewer qubits in the encoded state undergo errors, we
can decode the state correctly.  The notation for describing such
codes is $[[n,k,d]]$, where $d$ represents the distance of
the code, and the code encodes $k$ logical qubits into $n$
physical qubits.

These earlier codes are examples of additive or stabilizer codes. Additive codes
encode quantum information into the +1
eigenstates of $n$-fold tensor products of Pauli operators
\cite{PRA.54.1862.1996,thesis97gottesman}.  Gottesman developed an elegant
theory, the stabilizer formalism, that describes error correction, detection,
and recovery in terms of algebraic group theory~\cite{thesis97gottesman}.

Steane constructed a seven-qubit code that encodes one qubit, corrects an arbitrary
single-qubit error, and is an example of a Calderbank-Shor-Steane (CSS)\ code
\cite{steane1996}.  The five-qubit quantum error-correcting code is a \textquotedblleft perfect
code\textquotedblright\ in the sense that it encodes one qubit with the
smallest number of physical qubits while still correcting an arbitrary single-qubit
error \cite{perfectcode1996,PRA.54.3824.1996}.

Even though every stabilizer code is useful for fault-tolerant
computation \cite{PRA.54.1862.1996,thesis97gottesman}, CSS codes
allow for simpler fault-tolerant procedures. For example, an
encoded CNOT gate admits a transversal implementation without the
use of ancillas if and only if the code is of the CSS type
\cite{thesis97gottesman}.  The five-qubit code is not a CSS code
and does not possess the simple fault-tolerant properties of CSS
codes \cite{mikeandike2000}. The Steane code is a CSS code and
is well-suited for fault-tolerant computation because it has
bitwise implementations of the Hadamard and the phase gates as
well (the logical $X$ and $Z$ operators have bitwise
implementations for any stabilizer code \cite{PRA.54.1862.1996}).
However, an experimental realization of the seven-qubit code may
be more difficult to achieve than one for the five-qubit code
because it uses two additional physical qubits for encoding.

Calderbank {\it et al.}~discovered two distinct six-qubit quantum
codes \cite{ieee1998calderbank} which encode one qubit and correct
an arbitrary single-qubit error. They discovered the first of
these codes by extending the five-qubit code and the other one
through an exhaustive search of the encoding space. Neither of
these codes is a CSS code.

The five-qubit code and the Steane code have been studied
extensively \cite{mikeandike2000}, but the possibility for encoding one
qubit into six has not received much attention except for the
brief mention in Ref.~\cite{ieee1998calderbank}. In this thesis, we
bridge the gap between the five-qubit code and the Steane code by
discussing two examples of a six-qubit code. We also present
several proofs concerning the existence of single-error-correcting
CSS codes of a certain size. One of our proofs gives insight into
why Calderbank {\it et al.}~were unable to find a six-qubit CSS
code. The other proofs use a technique similar to the first proof
to show the non-existence of a CSS entanglement-assisted code that
uses fewer than six local physical qubits where one of the local qubits is half of one ebit, and corrects
an arbitrary single-qubit error.

We structure our work according to our four main results. We first
present a degenerate six-qubit quantum code and show how to
convert this code to a subsystem code. Our second result is a
proof for the non-existence of a single-error-correcting CSS
six-qubit code. Our third result is the construction of a
six-qubit CSS entanglement-assisted quantum code. This code is
globally equivalent to the Steane code. We finally
show that the latter is the smallest example of an
entanglement-assisted CSS code that corrects an arbitrary
single-qubit error.

In Section~\ref{sec:isaac}, we present a degenerate
six-qubit quantum error-correcting code that corrects an
arbitrary single-qubit error.  We present the logical Pauli operators
, CNOT and encoding circuit for this code.  We also prove
that a variation of this code gives us a non-trivial example of
a subsystem code that saturates the
subsystem Singleton bound~\cite{klap0703213}.

In Section~\ref{sec:ogy}, we present a proof that a single-error-correcting CSS six-qubit
code does not exist.
Our proof enumerates all possible CSS forms for the five
stabilizer generators of the six-qubit code and shows that none of
these forms corrects the set of all single-qubit errors.

Section~\ref{sec:bilal-mark} describes the construction of a
six-qubit non-degenerate entanglement-assisted CSS code and
presents its stabilizer generators, encoding circuit, and logical Pauli operators. This code
encodes one logical qubit into six local physical qubits. One of
the physical qubits used for encoding is half of an ebit that the
sender shares with the receiver. The six-qubit
entanglement-assisted code is globally equivalent to
the seven-qubit Steane code~\cite{steane1996} and thus corrects an
arbitrary single-qubit error on all of the qubits (including the
receiver's half of the ebit). This ability to correct errors on
the receiver's qubits in addition to the sender's qubits is not
the usual case with codes in the entanglement-assisted paradigm, a
model that assumes the receiver's halves of the ebits are noise
free because they are already on the receiving end of the channel.
We show that our example is a trivial case of a more general
rule---every $[[n,1,3]]$ code is equivalent to a $[[n-1,1,3;1]]$
entanglement-assisted code by using any qubit as Bob's half of the
ebit.

Finally, in section~\ref{sec:ogycss}, we present a proof that the
Steane code is an example of the smallest entanglement-assisted
code that corrects an arbitrary single-qubit error on the sender's
qubits, uses only one ebit, and possesses the CSS form.

Section~\ref{sec:eaqeccircuit} gives a procedure to obtain the encoding circuit for
the six-qubit CSS entanglement-assisted code. It also lists a
table detailing the error-correcting properties for the
degenerate six-qubit code.

\section{Degenerate Six-Qubit Quantum Code}

\label{sec:isaac}This section details an example of a six-qubit
code that corrects an arbitrary single-qubit error.  We explicitly present the
stabilizer generators, encoding circuit, logical codewords,
logical Pauli operators and CNOT\ operator for this
code. We also show how to convert this code into a
subsystem code where one of the qubits is a gauge qubit.  We
finish this section by discussing the advantages and disadvantages
of this code.

Calderbank {\it et al.} mention the existence of two non-equivalent six-qubit codes~\cite{ieee1998calderbank}.
Their first example is a trivial extension of the five-qubit code.
They append an ancilla qubit to the five-qubit code to obtain this code.
Their second example is a non-trivial six-qubit code. They argue that there are no other codes ``up to equivalence.''
Our example is not reducible to the trivial six-qubit code because every one of its qubits is entangled with the others.
It therefore is equivalent to the second non-trivial six-qubit code in Ref.~\cite{ieee1998calderbank}
according to the arguments of Calderbank  {\it et al.}

\begin{table}[tbp] \centering
\begin{tabular}
[c]{c|cccccc}\hline\hline
$h_{1}$ & $Y$ & $I$ & $Z$ & $X$ & $X$ & $Y$\\
$h_{2}$ & $Z$ & $X$ & $I$ & $I$ & $X$ & $Z$\\
$h_{3}$ & $I$ & $Z$ & $X$ & $X$ & $X$ & $X$\\
$h_{4}$ & $I$ & $I$ & $I$ & $Z$ & $I$ & $Z$\\
$h_{5}$ & $Z$ & $Z$ & $Z$ & $I$ & $Z$ & $I$\\ \hline
$\overline{X}$ & $Z$ & $I$ & $X$ & $I$ & $X$ & $I$\\
$\overline{Z}$ & $I$ & $Z$ & $I$ & $I$ & $Z$ & $Z$
\\\hline\hline
\end{tabular}
\caption{Stabilizer generators $h_1$, \ldots, $h_5$, and logical operators $\bar{X}$ and $\bar{Z}$ for the six-qubit code.  The convention in the above generators is that $Y=ZX$.}\label{tbl:isaac-613}
\end{table}%

Five generators specify the degenerate six-qubit code.
Table~\ref{tbl:isaac-613} lists the generators $h_{1}$, \ldots, $h_{5}$ in the stabilizer $\mathcal{S}$, and the logical operators $\overline{X}$ and
$\overline{Z}$ for the six-qubit code. Figure~\ref{fig:six-qubit-1}\ illustrates an
encoding circuit for the six-qubit code.

The quantum error-correcting conditions guarantee that the six-qubit
code corrects an arbitrary single-qubit error
\cite{mikeandike2000}. Specifically, the error-correcting
conditions are as follows: a stabilizer $\mathcal{S}$ with
generators $s_i$ where $i=1,\ldots,n-k$ (in our case $n=6$ and
$k=1$), corrects an error set $\mathcal{E}$ if every error pair
$E^{\dag}_{a}E_{b}\in \mathcal{E}$ either anticommutes with at
least one stabilizer generator
\begin{equation}
\exists\ s_{i}\in\mathcal{S}:\left\{  s_{i},E_{a}^{\dag}E_{b}\right\}  =0,
\end{equation}
or is in the stabilizer,
\begin{equation}
E_{a}^{\dag}E_{b}\in\mathcal{S}.
\end{equation}
These conditions imply the ability to correct any linear
combination of errors in the set $\mathcal{E}$
\cite{mikeandike2000,book2007mermin}.  At least one
generator from the six-qubit stabilizer anticommutes with each of the
single-qubit Pauli errors, $X_{i},Y_{i},Z_{i}$ where
$i=1,\ldots,6$, because the generators have at least one
$Z$ and one $X$ operator in all six positions. Additionally, at
least one generator from the stabilizer anticommutes with each
pair of two distinct Pauli errors (except $Z_{4}Z_{6}$, which is
in the stabilizer $\mathcal{S}$).
Table~\ref{table:sixonethree1} lists such a
generator for every pair of distinct Pauli errors for the
six-qubit code. These arguments and the table listings prove that
the code can correct an arbitrary single-qubit error.
\begin{figure}
[ptb]
\begin{center}
\[
\Qcircuit@C= 0.5em @R = 0.5em @!R{
\lstick{\ket{0}} & \gate{H} & \qw & \qw & \qw & \qw & \qw & \qw & \qw & \qw & \ctrl{5} & \qw & \ctrl{5} & \qw & \qw\\
\lstick{\ket{0}} & \gate{H} & \qw & \qw & \qw & \qw & \qw & \qw & \qw & \ctrl{4} & \qw & \qw & \qw & \qw & \qw\\
\lstick{\ket{0}} & \qw & \qw & \qw & \qswap & \ctrl{3} & \gate{H} & \ctrl{2} & \gate{H} & \qw & \qw & \qw & \targ & \gate{H} & \qw\\
\lstick{\ket{0}} & \gate{H} & \ctrl{2} & \gate{H} & \qw & \qw & \qw & \targ & \qw & \qw & \qw & \qw & \targ & \qw & \qw\\
\lstick{\ket{0}} & \qw & \qw & \qw & \qw & \qw & \qw & \targ & \qw & \targ & \qw & \qw & \targ & \qw & \qw\\
\lstick{\ket{\psi}} & \qw & \targ & \qw & \qswap\qwx[-3] & \targ & \qw & \qw & \qw & \targ & \targ & \gate{H} & \targ & \qw & \qw\\
}
\]
\end{center}
\caption
{Encoding circuit for the first six-qubit code.  The \textit{H} gate is a Hadamard gate.  For example, we apply
a Hadamard on qubit four followed by a CNOT with qubit four as the control qubit and qubit six as the target qubit. }
\label{fig:six-qubit-1}
\end{figure}
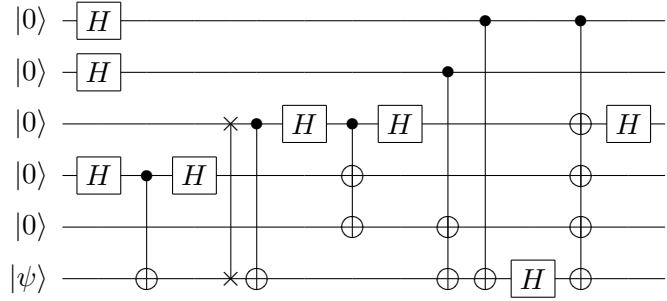
The logical basis states for the six-qubit code are as follows:
\begin{align*}
|\overline{0}\rangle &  =%
\begin{array}
[c]{c}%
|000000\rangle-|100111\rangle+|001111\rangle-|101000\rangle\ -\\
|010010\rangle+|110101\rangle+|011101\rangle-|111010\rangle
\end{array},\\
|\overline{1}\rangle &  =%
\begin{array}
[c]{c}%
|001010\rangle+|101101\rangle+|000101\rangle+|100010\rangle\ -\\
|011000\rangle-|111111\rangle+|010111\rangle+|110000\rangle
\end{array},
\end{align*}
where we suppress the normalization factors of the above codewords.

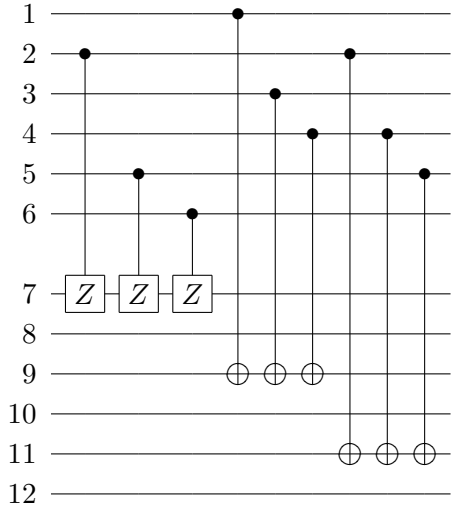
\begin{figure}
[ptb]
\begin{center}
\[
\Qcircuit @C = .5em @R = 0.1em @!R{
\lstick{1} & \qw & \qw & \qw & \ctrl{9} & \qw & \qw & \qw & \qw & \qw & \qw\\
\lstick{2} & \ctrl{6} & \qw & \qw & \qw & \qw & \qw & \ctrl{10} & \qw & \qw & \qw\\
\lstick{3} & \qw & \qw & \qw & \qw & \ctrl{7} & \qw & \qw & \qw & \qw & \qw\\
\lstick{4} & \qw & \qw & \qw & \qw & \qw & \ctrl{6} & \qw & \ctrl{8} & \qw & \qw\\
\lstick{5} & \qw & \ctrl{3} & \qw & \qw & \qw & \qw & \qw & \qw & \ctrl{7} & \qw\\
\lstick{6} & \qw & \qw & \ctrl{2} & \qw & \qw & \qw & \qw & \qw & \qw & \qw\\
& & & & & & & & & &\\
\lstick{7} & \gate{Z} & \gate{Z} & \gate{Z} & \qw & \qw & \qw & \qw  & \qw & \qw & \qw\\
\lstick{8} & \qw & \qw & \qw & \qw & \qw & \qw & \qw & \qw & \qw & \qw\\
\lstick{9} & \qw & \qw & \qw & \targ & \targ & \targ & \qw & \qw & \qw & \qw\\
\lstick{10} & \qw & \qw & \qw & \qw & \qw & \qw & \qw & \qw & \qw & \qw\\
\lstick{11} & \qw & \qw & \qw & \qw & \qw & \qw & \targ & \targ & \targ & \qw\\
\lstick{12} & \qw & \qw & \qw & \qw & \qw & \qw & \qw & \qw & \qw & \qw
}
\]
\end{center}
\caption{Logical CNOT for the six-qubit quantum code.  The first six qubits represent a logical source qubit and the last six represent a logical target qubit.  For example we begin the circuit by applying a CZ (controlled-Z) gate from source qubit two to target qubit seven.}\label{fig-logicalcnot-1}
\end{figure}

A series of CNOT and controlled-$Z$ operations implement the logical CNOT
operation for the six-qubit code.  Let CN$\left(  i,j\right)  $ denote a CNOT
acting on physical qubits $i$ and $j$ with qubit $i$ as
the control and qubit $j$ as the target. Let CZ$\left(  i,j\right)  $
denote controlled-$Z$ operations.\ The logical CNOT for the six-qubit code is as follows:%
\begin{multline*}
\overline{\text{CNOT}}=\text{CZ}\left(  2,7\right)  \ \text{CZ}\left(
5,7\right)  \ \text{CZ}\left(  6,7\right)  \ \text{CN}\left(  1,9\right)  \\
\text{CN}\left(  3,9\right)  \ \text{CN}\left(  4,9\right)  \ \text{CN}%
\left(  2,11\right) \ \text{CN}\left(  4,11\right)  \ \text{CN}\left(  5,11\right)
\end{multline*}
Figure~\ref{fig-logicalcnot-1} depicts the logical CNOT\ acting on two logical qubits
encoded with the six-qubit code.

Both the six-qubit code and the five-qubit code correct
an arbitrary single-qubit error. But the six-qubit code has the
advantage that it corrects a larger set of errors than the
five-qubit code. This error-correcting capability comes at the expense of a larger
number of qubits---it corrects a larger set of errors because the Hilbert
space for encoding is larger than that for the five-qubit code. In
comparison to the Steane code, the six-qubit code uses a smaller
number of qubits, but the disadvantage is that it does not admit
a simple transversal implementation of the logical CNOT. In
addition, the Steane code admits a bitwise implementation
of all logical single-qubit Clifford gates whereas the six-qubit
code does not.

\subsection{Subsystem Code Construction}

\label{subsystem}We convert the degenerate six-qubit code from the previous
section into a subsystem code. The degeneracy inherent in the code allows us
to perform this conversion.  The code
still corrects an arbitrary single-qubit error after we replace one of the
unencoded ancilla qubits with a gauge qubit.

We briefly review the history of subsystem codes. The essential
insight of Knill {\it et al.} was that the most general way to
encode quantum information is into a subsystem rather than
a subspace \cite{knillPRL00}. In the case when the information is
encoded in a single subsystem, the Hilbert space decomposes as
$\CH = (\CH_{A} \otimes \CH_{B})\oplus \CH_{C}$ where the
subsystem $\CH_{A}$ stores the protected information. Errors that
act on subsystem $\CH_{B}$, also known as the gauge subsystem, do
not require active correction because $\CH_{B}$ does not store any
valuable information. This passive error-correction ability of a
subsystem code may lead to a smaller number of stabilizer
measurements during the recovery process and may lead to an
improvement of the accuracy threshold for quantum computation \cite{aliferisPRL07}.
Kribs {\it et al.} recognized in Ref.~\cite{kribsPRL05} that
this subsystem structure of a Hilbert space is useful
for active quantum error-correction as well (Knill {\it et al.} did not explicitly recognize this ability in Ref.~\cite{knillPRL00}.)

We now detail how to convert the six-qubit code from the
previous section into a subsystem code.  The sixth unencoded qubit is the
information qubit and the encoding operation transforms it into
subsystem $\CH_A$. We convert the fourth unencoded ancilla qubit to a
gauge qubit. We simply consider it as a noisy qubit so that the
operators $X_4$ and $Z_4$ have no effect on the quantum
information stored in subsystem $\CH_{A}$. The operators $X_4$ and
$Z_4$ generate the unencoded gauge group. The encoding circuit in
Figure~\ref{fig:six-qubit-1} transforms these unencoded operators
into $X_4$ and $Z_4Z_6$ respectively. These operators together
generate the encoded gauge subgroup $H = \left\langle
X_4,Z_4Z_6\right\rangle$.  Errors in this subgroup do not affect the
encoded quantum information.  The code is still able to correct an
arbitrary single-qubit error because each one of the single-qubit
Pauli error pairs anticommutes with at least one of the generators
from the new stabilizer $\widetilde{\mathcal{S}}=\left\langle
h_1,h_2,h_3,h_5\right\rangle$, or belong to $H$
\cite{poulinPRL05}. Table~\ref{table:sixonethree1} shows this
property for all error pairs. The code passively corrects the error pairs $X_4$, $Z_4Z_6$,
$Y_4Z_6$ because they belong to the gauge subgroup.

\begin{table}[tbp] \centering
\begin{tabular}
[c]{c|cccccc}\hline\hline
$h_{1}$ & $Y$ & $I$ & $Z$ & $X$ & $X$ & $Y$\\
$h_{2}$ & $Z$ & $X$ & $I$ & $I$ & $X$ & $Z$\\
$h_{3}$ & $I$ & $Z$ & $X$ & $X$ & $X$ & $X$\\
$h_{5}$ & $Z$ & $Z$ & $Z$ & $I$ & $Z$ & $I$\\ \hline
$H_{X}$ & $I$ & $I$ & $I$ & $X$ & $I$ & $I$\\
$H_{Z}$ & $I$ & $I$ & $I$ & $Z$ & $I$ & $Z$\\ \hline
$\overline{X}$ & $Z$ & $I$ & $X$ & $I$ & $X$ & $I$\\
$\overline{Z}$ & $I$ & $Z$ & $I$ & $I$ & $Z$ & $Z$
\\\hline\hline
\end{tabular}
\caption{Stabilizer generators $h_1, h_2, h_3$ and $h_5$, gauge subgroup generators $H_{X}$ and $H_{Z}$, and logical operators $\bar{X}$ and $\bar{Z}$ for the six-qubit code.  The convention in the above generators is that $Y=ZX$.}\label{tbl:subsystem}%
\end{table}%

The six-qubit single-error-correcting subsystem code
discussed above saturates the Singleton bound for subsystem codes
\cite{klap0703213},
\begin{equation}
n -k - r \geq 2(d - 1),
\end{equation}
where for our code, $n=6$, $k=1$, $r=1$, and $d=3$.
This code is the smallest non-trivial subsystem
code that corrects an arbitrary single-qubit
error \footnote{A trivial way to saturate this bound is to add a noisy qubit to the five-qubit code!}.  One of
the advantages of using the subsystem construction is that we only
need to perform four stabilizer measurements instead of five
during the recovery process.

\section{Non-existence of a $[[6,1,3]]$ CSS Code}

\label{sec:ogy}Our proposition below proves that it is impossible
for a six-qubit code to possess the CSS\ structure while
correcting an arbitrary single-qubit error.  An immediate
corollary of this proposition is that the seven-qubit code is the smallest
single-error-correcting CSS code.

\begin{proposition}
There is no six-qubit code that encodes one qubit, possesses the CSS\ structure, and corrects an
arbitrary single-qubit error.
\end{proposition}

\begin{proof}
We first suppose that a code with the above properties exists. If
a $[[6,1,3]]$ CSS code exists, its stabilizer $\mathcal{S}$ must
have five generators:
\begin{equation}
\mathcal{S}=\langle g_{1},...,g_{5}\rangle.
\end{equation}
The CSS\ structure implies that each of these generators includes
$X$ operators only or $Z$ operators only (except for the
identity). The set of correctable Pauli errors $\{E_{j}\}$ in the
Pauli group acting on six qubits satisfies
$\{E_{i}E_{j},\mathcal{S}\}=0$ unless $E_{i}E_{j}\in \mathcal{S}$,
for all $i,j$. We show below that no set of five CSS\ stabilizer
generators acting on six qubits can correct an arbitrary
single-qubit error and possess the CSS\ structure.

First assume that such generators exist. It is not possible that
all generators consist of the same type of operators (all $X$ or
all $Z$) because single-qubit errors of the same type ($X$ or $Z$)
are then not correctable. Consider the possibility that there is
one generator of one type, say $X$, and four generators of the
other type, say $Z$. If the generator of type $X$ has an identity
acting on any qubit, say the first one, then the error $Z_{1}$
commutes with all generators. This error is not correctable unless
it belongs to the stabilizer. But if it belongs to the stabilizer,
the first qubit of the code must be fixed in the state
$|0\rangle$, which makes for a trivial code. The other possibility
is that the $X$-type generator has the form $g_{1}=XXXXXX$. But
then any combination of two $Z$-errors ($Z_{i}Z_{j}$) commutes
with it, and so they have to belong to the stabilizer. But there
are five independent such combinations of errors ($Z_{1}Z_{2}$,
$Z_{1}Z_{3}$, $Z_{1}Z_{4}$, $Z_{1}Z_{5}$, $Z_{1}Z_{6}$) and only
four generators of the $Z$ type. Therefore, it is impossible for
the code to have four generators of one type and one generator of
the other type.

The only possibility left is that there are two generators of one type, say
$X$, and three generators of the other type, say $Z$. The two $X$-type generators
should not both have identity acting on any given qubit because a $Z$ error on
that qubit commutes with all generators. Such an error cannot belong to the
stabilizer because it would again make for a trivial code. Specifically, we write the two $X$-type generators
($g_{1}$ and $g_{2}$) one above the other
\begin{equation}
\begin{matrix}
g_{1}\\
g_{2}
\end{matrix}
=
\begin{matrix}
- & - & - & - & - & -\\
- & - & - & - & - & -
\end{matrix},
\label{eq:matrix}
\end{equation}
where we leave the entries unspecified in the above equation,
but they are either $X$ or $I$. Both generators cannot have the column
\begin{equation}
\begin{matrix}
I\\
I
\end{matrix}
\nonumber
\end{equation}
in (\ref{eq:matrix}) because both generators
cannot have identities acting on the same qubit. Thus, only three different
columns can build up the generators in (\ref{eq:matrix}):
\begin{equation}
\begin{matrix}
I\\
X
\end{matrix}
\hspace{0.5cm},\hspace{0.5cm}
\begin{matrix}
X\\
I
\end{matrix}
\hspace{0.5cm},\hspace{0.5cm}
\begin{matrix}
X\\
X
\end{matrix}
\hspace{0.5cm}\nonumber
\end{equation}
We distinguish the following cases:
\begin{enumerate}
\item Each column appears twice.
\item One column appears three times, another column appears twice, and the
third column appears once.
\item One column appears three times and another column appears three times.
\item At least one column appears more than three times.
\end{enumerate}
If one and the same column appears on two different places, say qubit one and
qubit two as in the following example,
\begin{equation}
\begin{matrix}
g_{1}\\
g_{2}
\end{matrix}
=
\begin{matrix}
X & X & - & - & - & -\\
I & I & - & - & - & -
\end{matrix}
\hspace{0.5cm},
\end{equation}
then a pair of $Z$ errors on these qubits ($Z_{1}Z_{2}$) commutes with all
generators, and therefore belongs to the stabilizer.

In the first case considered above, there are three such pairs of errors,
which up to a relabeling of the qubits can be taken to be $Z_{1}Z_{2}$,
$Z_{3}Z_{4}$, $Z_{5}Z_{6}$. They can be used as stabilizer generators because
these operators are independent. But then the following pairs of single-qubit
$X$ errors commute with all generators: $X_{1}X_{2}$, $X_{3}X_{4}$,
$X_{5}X_{6}$. This possibility is ruled out because the latter cannot be part
of the stabilizer generators.

In the second case, up to a relabeling of the qubits, we have the following pairs of
$Z$ errors that commute with the stabilizer: $Z_{1}Z_{2}$, $Z_{1}Z_{3}$,
$Z_{2}Z_{3}$, $Z_{4}Z_{5}$. Only three of all four are independent, and they
can be taken to be stabilizer generators. But then all three generators of
$Z$-type have the identity acting on the sixth qubit, and therefore the error
$X_{6}$ is not correctable (and it cannot be a stabilizer generator because
it would make for a trivial code).

In the third case, the pairs $Z_{1}Z_{2}$, $Z_{1}Z_{3}$, $Z_{2}Z_{3}$, $Z_{4}Z_{5}$,
$Z_{4}Z_{6}$, $Z_{5}Z_{6}$ (up to a relabeling), four of which are
independent, commute with the stabilizer. But they cannot all belong to the
stabilizer because there are only three possible generators of the $Z$-type.

Finally, in the fourth case, we have three or more independent pairs of $Z$ errors
commuting with the stabilizer (for example $Z_{1}Z_{2}$, $Z_{1}Z_{3}$,
$Z_{1}Z_{4}$, which corresponds to the first four columns being identical). If
the independent pairs are more than three, then their number is more
than the possible number of generators. If they are exactly three, we can take
them as generators. But then $Z$-type generators act trivially upon two qubits,
and therefore $X$ errors on these qubits are not correctable. This last step
completes the proof.
\end{proof}

\section{Non-degenerate Six-Qubit CSS Entanglement-Assisted Quantum Code}

We detail the construction of a six-qubit
CSS entanglement-assisted quantum code in this section. We first
review the history of entanglement-assisted quantum coding and
discuss the operation of an entanglement-assisted code. We then
describe our construction. It turns out that the code we obtain is
equivalent to the Steane code \cite{steane1996} when including
Bob's qubit, and therefore is not a new code. It suggests,
however, a general rule for which we present a proof---every
$[[n,1,3]]$ code is equivalent to a $[[n-1,1,3;1]]$
entanglement-assisted code with any qubit serving as Bob's half of
the ebit. Even though our code is a trivial example of this rule,
it is instructive to present its derivation from the perspective
of the theory of entanglement-assisted codes.

\label{sec:bilal-mark}Bowen constructed an example of a quantum
error-correcting code that exploits shared entanglement between sender and
receiver \cite{PhysRevA.66.052313}. Brun, Devetak, and Hsieh later generalized
Bowen's example and developed the entanglement-assisted stabilizer formalism
\cite{arx2006brun,science2006brun}. This theory is an extension of the
standard stabilizer formalism and uses shared entanglement to formulate
stabilizer codes. Several references provide a review
\cite{arx2006brun,science2006brun,arxiv2007brun} and generalizations of the
theory to entanglement-assisted operator codes
\cite{isit2007brun,arxiv2007brun}, convolutional entanglement distillation
protocols \cite{arx2007wilde}, continuous-variable codes
\cite{arx2007wildeEA}, and entanglement-assisted quantum convolutional codes
\cite{arx2007wildeEAQCC}. Gilbert {\it et al.}~also generalized their \textquotedblleft
quantum computer condition\textquotedblright\ for fault tolerance to the
entanglement-assisted case \cite{arx2007gilbert}.  Entanglement-assisted
codes are a special case of ``correlation-assisted codes'', where Bob's qubit
is also allowed to be noisy. Such codes are in turn instances
of general linear quantum error-correcting codes \cite{SL:07}.

An entanglement-assisted quantum error-correcting code operates as
follows. A sender and receiver share $c$ ebits before
communication takes place. The sender possesses her half of the
$c$ ebits, $n-k-c$ ancilla qubits, and $k$ information qubits. She
performs an encoding unitary on her $n$\ qubits and sends them
over a noisy quantum communication channel. The receiver combines
his half of the $c$\ ebits with the $n$ encoded qubits and
performs measurements on all of the qubits to diagnose the errors
from the noisy channel. The generators corresponding to the
measurements on all of the qubits form a commuting set. The
generators thus form a valid stabilizer, they do not disturb the
encoded quantum information, and they learn only about the errors
from the noisy channel. The notation for such a code is
$[[n,k,d;c]]$, where $d$ is the distance of the code.

The typical assumption for an entanglement-assisted quantum code
is that noise does not affect Bob's half of the ebits because they
reside on the other side of a noisy quantum communication channel
between Alice and Bob. Our $[[6,1,3;1]]$ entanglement-assisted code
is globally equivalent to the $[[7,1,3]]$ Steane code
and thus corrects errors on Bob's side as well. From a
computational perspective, a code that corrects errors on all
qubits is more powerful than a code that does not. From the
perspective of the entanglement-assisted paradigm, however, this feature
is unnecessary and may result in decreased
error-correcting capabilities of the code with respect to errors
on Alice's side.

We construct our code using the parity check matrix of a classical
code. Consider the parity check matrix for the $\left[
7,4,3\right]$ Hamming
code:%
\begin{equation}
\left[
\begin{array}
[c]{ccccccc}%
1 & 0 & 0 & 1 & 0 & 1 & 1\\
0 & 1 & 0 & 1 & 1 & 0 & 1\\
0 & 0 & 1 & 0 & 1 & 1 & 1
\end{array}
\right]
\end{equation}
The Hamming code encodes four classical bits and corrects a single-bit error.
We remove one column of the above parity check matrix to form a new parity check matrix
$H$\ as follows:%
\begin{equation}
H=\left[
\begin{array}
[c]{cccccc}%
1 & 0 & 0 & 1 & 0 & 1\\
0 & 1 & 0 & 1 & 1 & 0\\
0 & 0 & 1 & 0 & 1 & 1
\end{array}
\right]  .
\end{equation}
The code corresponding to $H$\ encodes three bits and still corrects a
single-bit error. We begin constructing the stabilizer for an
entanglement-assisted quantum code by using the CSS construction
\cite{isit2007brun,arxiv2007brun}:%
\begin{equation}
\left[  \left.
\begin{array}
[c]{c}%
H\\
0
\end{array}
\right\vert
\begin{array}
[c]{c}%
0\\
H
\end{array}
\right]  .\label{eq:EA-six-qubit-binary}%
\end{equation}
The left side of the above matrix is the \textquotedblleft Z\textquotedblright%
\ side and the right side of the above matrix is the \textquotedblleft
X\textquotedblright\ side. The isomorphism between $n$-fold tensor products of
Pauli matrices and $n$-dimensional binary vectors gives a correspondence
between the matrix in (\ref{eq:EA-six-qubit-binary})\ and the set of
Pauli generators below \cite{thesis97gottesman,mikeandike2000,arx2006brun}:%
\begin{equation}%
\begin{array}
[c]{cccccc}%
Z & I & I & Z & I & Z\\
I & Z & I & Z & Z & I\\
I & I & Z & I & Z & Z\\
X & I & I & X & I & X\\
I & X & I & X & X & I\\
I & I & X & I & X & X
\end{array}
\end{equation}
The above set of generators have good quantum error-correcting properties
because they correct an arbitrary single-qubit error. These properties follow
directly from the properties of the classical code. The problem with the above
generators is that they do not form a commuting set and thus do not correspond
to a valid quantum code. We use entanglement to resolve this problem by
employing the method outlined in
Ref.~\cite{arx2006brun,science2006brun,arxiv2007brun}.
\begin{table}[tbp] \centering
\begin{tabular}
[c]{ccc}%
\begin{tabular}
[c]{c|c|cccccc}
& Bob & \multicolumn{6}{|c}{Alice}\\\hline\hline
$g_{1}^{\prime}$ & $I$ & $I$ & $Z$ & $I$ & $I$ & $I$ & $I$\\
$g_{2}^{\prime}$ & $I$ & $I$ & $I$ & $Z$ & $I$ & $I$ & $I$\\
$g_{3}^{\prime}$ & $Z$ & $Z$ & $I$ & $I$ & $I$ & $I$ & $I$\\
$g_{4}^{\prime}$ & $I$ & $I$ & $I$ & $I$ & $Z$ & $I$ & $I$\\
$g_{5}^{\prime}$ & $I$ & $I$ & $I$ & $I$ & $I$ & $Z$ & $I$\\
$g_{6}^{\prime}$ & $X$ & $X$ & $I$ & $I$ & $I$ & $I$ & $I$\\\hline
$\overline{X}^{\prime}$ & $I$ & $I$ & $I$ & $I$ & $I$ & $I$ & $X$\\
$\overline{Z}^{\prime}$ & $I$ & $I$ & $I$ & $I$ & $I$ & $I$ & $Z$%
\\\hline\hline
\end{tabular}
& \quad\quad &
\begin{tabular}
[c]{c|c|cccccc}
& Bob & \multicolumn{6}{|c}{Alice}\\\hline\hline
$g_{1}$ & $I$ & $Z$ & $I$ & $Z$ & $Z$ & $Z$ & $I$\\
$g_{2}$ & $I$ & $Z$ & $Z$ & $I$ & $I$ & $Z$ & $Z$\\
$g_{3}$ & $Z$ & $Z$ & $I$ & $I$ & $Z$ & $I$ & $Z$\\
$g_{4}$ & $I$ & $X$ & $X$ & $I$ & $I$ & $X$ & $X$\\
$g_{5}$ & $I$ & $I$ & $X$ & $X$ & $X$ & $I$ & $X$\\
$g_{6}$ & $X$ & $X$ & $I$ & $I$ & $X$ & $I$ & $X$\\\hline
$\overline{X}$ & $I$ & $I$ & $I$ & $I$ & $X$ & $X$ & $X$\\
$\overline{Z}$ & $I$ & $I$ & $Z$ & $Z$ & $I$ & $Z$ & $I$\\\hline\hline
\end{tabular}
\\
(a) &  & (b)
\end{tabular}
\caption{(a) The generators and logical operators for the
unencoded state. Generators $g'_3$ and $g'_6$ indicate that Alice
and Bob share an ebit. Alice's half of the ebit is her first qubit
and Bob's qubit is the other half of the ebit. Generators $g'_1$,
$g'_2$, $g'_4$, and $g'_5$ indicate that Alice's second, third,
fourth, and fifth respective qubits are ancilla qubits in the
state $\left\vert 0 \right\rangle$. The unencoded logical
operators $\bar{X}'$ and $\bar{Z}'$ act on the sixth qubit and
indicate that the sixth qubit is the information qubit. (b) The
encoded generators and logical operators for the $[[6,1,3;1]]$
entanglement-assisted quantum error-correcting code.}\label{tbl:ea-613}%
\end{table}%

Three different but related methods determine the minimum number of ebits that
the entanglement-assisted quantum code requires:

\begin{enumerate}
\item Multiplication of the above generators with one another according to the
``symplectic
Gram-Schmidt orthogonalization algorithm''
forms a new set of generators \cite{arx2006brun,science2006brun}.
The error-correcting properties of the code are
invariant under these multiplications because the code is an additive code.
The resulting code has equivalent error-correcting properties and uses the
minimum number of ebits. We employ this technique in this work.

\item A slightly different algorithm in the appendix of Ref.~\cite{arx2007wilde}
determines the minimum number of ebits
required, the stabilizer measurements to perform, and the local encoding
unitary that Alice performs to rotate the unencoded state to the encoded state.
This algorithm is the most useful because it
\textquotedblleft kills three birds with one stone.\textquotedblright

\item The minimum number of ebits for a CSS\ entanglement-assisted code is equal to
the rank of $HH^{T}$ \cite{isit2007brun,arxiv2007brun}. This simple formula is
useful if we are only concerned with computing the minimum number of ebits. It
does not determine the stabilizer generators or the encoding circuit.
Our code requires one ebit to form a valid
stabilizer code because the rank
of $HH^{T}$ for our code is equal to one.
\end{enumerate}

Table~\ref{tbl:ea-613}(b) gives the final form of the stabilizer
for our entanglement-assisted six-qubit code. We list both the
unencoded and the encoded generators for this code in
Table~\ref{tbl:ea-613}.

Our code uses one ebit shared between sender and receiver in the encoding
process. The sender performs a local encoding unitary that encodes one qubit
with the help of four ancilla qubits and one ebit.

The symplectic Gram-Schmidt algorithm yields a symplectic matrix that rotates
the unencoded symplectic vectors to the encoded symplectic vectors. The
symplectic matrix corresponds to an encoding unitary acting on the unencoded
quantum state \cite{arx2006brun,science2006brun}. This correspondence results
from the Stone-von Neumann Theorem and unifies the Schr\"{o}dinger and
Heisenberg pictures for quantum error correction \cite{eisert-2003-1}.

The symplectic Gram-Schmidt algorithm also determines the logical operators
for the code. Some of the vectors in the symplectic matrix that do not
correspond to a stabilizer generator are equivalent to the logical operators
for the code. It is straightforward to determine which symplectic vector
corresponds to which logical operator ($X$ or $Z$) by observing the action of
the symplectic matrix on vectors that correspond to the unencoded $X$ or $Z$
logical operators.

For our code, the symplectic matrix is as follows:%
\begin{equation}
\left[  \left.
\begin{array}
[c]{cccccc}%
1 & 0 & 0 & 1 & 0 & 1\\
1 & 0 & 1 & 1 & 1 & 0\\
1 & 1 & 0 & 0 & 1 & 1\\
0 & 0 & 0 & 0 & 0 & 0\\
0 & 0 & 0 & 0 & 0 & 0\\
0 & 1 & 1 & 0 & 1 & 0\\
0 & 0 & 0 & 0 & 0 & 0\\
0 & 0 & 0 & 0 & 0 & 0\\
0 & 0 & 0 & 0 & 0 & 0\\
0 & 0 & 0 & 1 & 0 & 1\\
0 & 1 & 0 & 1 & 0 & 1\\
0 & 0 & 0 & 0 & 0 & 0
\end{array}
\right\vert
\begin{array}
[c]{cccccc}%
0 & 0 & 0 & 0 & 0 & 0\\
0 & 0 & 0 & 0 & 0 & 0\\
0 & 0 & 0 & 0 & 0 & 0\\
1 & 1 & 0 & 0 & 1 & 1\\
0 & 1 & 1 & 1 & 0 & 1\\
0 & 0 & 0 & 0 & 0 & 0\\
1 & 0 & 0 & 1 & 0 & 1\\
0 & 0 & 1 & 1 & 1 & 1\\
0 & 0 & 1 & 0 & 1 & 0\\
0 & 0 & 0 & 0 & 0 & 0\\
0 & 0 & 0 & 0 & 0 & 0\\
0 & 0 & 0 & 1 & 1 & 1
\end{array}
\right]
\end{equation}
The index of the rows of the above matrix corresponds to the operators in the
unencoded stabilizer in Table~\ref{tbl:ea-613}(a). Therefore, the first five
rows correspond to the encoded $Z$ operators in the stabilizer and the sixth
row corresponds to the logical $\overline{Z}$ operator. As an example, we can
represent the unencoded logical $\overline{Z}$ operator in
Table~\ref{tbl:ea-613}(a) as the following binary vector:%
\begin{equation}
\left[  \left.
\begin{array}
[c]{cccccc}%
0 & 0 & 0 & 0 & 0 & 1
\end{array}
\right\vert
\begin{array}
[c]{cccccc}%
0 & 0 & 0 & 0 & 0 & 0
\end{array}
\right]  .
\end{equation}
Premultiplying the above matrix by the above row vector gives the binary form
for the encoded logical $\overline{Z}$ operator. We can then translate this
binary vector to a six-fold tensor product of Paulis equivalent to the
logical $\overline{Z}$ operator in Table~\ref{tbl:ea-613}(b). Using this same
idea, the first row of the above matrix corresponds to Alice's Paulis in
$g_{3}$, the second row to $g_{1}$, the third row to $g_{2}$, the fourth row
to $g_{4}$, the fifth row to $g_{5}$, and the seventh row to $g_{6}$. The last
six rows in the above matrix correspond to encoded $X$ operators and it is
only the last row that is interesting because it acts as a logical $X$ operator.

\begin{figure}
[ptb]
\begin{center}
\[
\Qcircuit@C= 0.5em @R = 0.5em @!R{
& & & \qw& \qw& \qw& \qw& \qw& \qw& \qw& \qw& \qw& \qw& \qw & \qw & \qw & \qw\gategroup{1}{3}{2}{3}{0.5em}{\{}\\
& & \lstick{\raisebox{2em}{$\ket{\Phi^{+}}^{BA}$}} & \qw& \qw& \qw& \qw& \qw& \qw& \qw& \qw& \qw& \ctrl{5} & \gate{H} & \ctrl{5} & \gate{H} & \qw\\
& \lstick{\ket{0}^A} & \gate{H} & \qw & \qw & \qw & \qw & \qw & \qw & \qw & \ctrl{3} & \gate{H} & \qw & \qw & \qw & \qw & \qw \\
& \lstick{\ket{0}^A} & \gate{H} & \qw & \qw & \qw & \qw& \ctrl{3} & \gate{H} & \qswap & \qw & \qw & \qw & \qw & \qw & \qw & \qw\\
& \lstick{\ket{0}^A} & \gate{H} & \qw & \qw & \ctrl{2} & \gate{H} & \qw & \qw & \qw &\targ & \gate{H} & \targ & \gate{H} & \targ & \gate{H} & \qw\\
& \lstick{\ket{0}^A} & \gate{H} & \ctrl{1} & \gate{H} & \targ& \qw & \qw & \qw & \qswap\qwx[-2] & \targ & \gate{H} & \qw & \qw & \qw & \qw & \qw\\
& \lstick{\ket{\psi}^A} & \qw &\targ & \gate{H} & \targ & \qw & \targ \qw & \qw & \qw & \qw & \qw & \targ & \gate{H} & \targ & \gate{H} & \qw\\
}
\]
\end{center}
\caption
{Encoding circuit for the [[6,1,3;1]] code.  The ``H'' gate is a Hadamard gate.}\label{fig-eacodecircuit}
\end{figure}
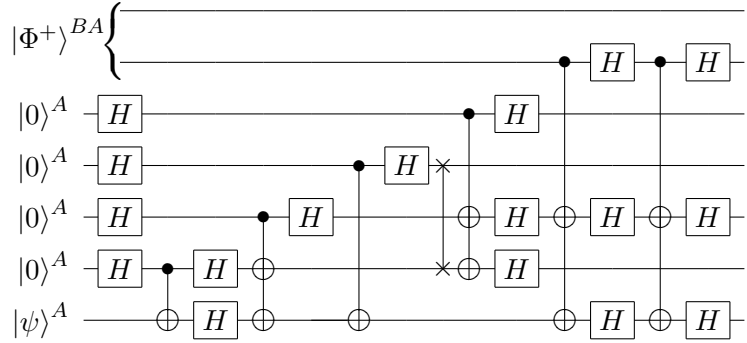%

Figure~\ref{fig-eacodecircuit} gives the encoding circuit for the
code.

We now detail the operations that give the equivalence of this
code to the seven-qubit Steane code. Consider the generators in
Table~\ref{tbl:ea-613}(b). Label the columns from left to right as
$1$, $2$, \ldots, $7$ where ``1'' corresponds to Bob's column.
Replace the generator $g_1$ by $g_1 g_2 g_3$, and the generator
$g_5$ by $g_5 g_6$. Switch the new generators $g_4$ and $g_5$.
Switch columns 2 and 3. Switch columns 1 and 5. Cyclically permute
the columns once so that 1 becomes 7, 2 becomes 1, 3 becomes 2,
..., 7 becomes 6. The resulting code is exactly the Steane code if
one reads it from right to left (i.e., up to the permutation $1
\leftrightarrow 7$, $2 \leftrightarrow 6$, $3 \leftrightarrow 5$).

Inspection of the encoded logical operators in
Table~\ref{tbl:ea-613}(b) reveals that Alice can perform logical
$\overline{X}$ and $\overline{Z}$ operations locally. Since the
CNOT\ has a transversal implementation for the Steane code, if
Alice and Bob possess two logical qubits each encoded with this
entanglement-assisted code, they can apply an encoded CNOT
transversally by the use of classical communication to coordinate
their actions. We point out, however, that the idea of computation
in the entanglement-assisted paradigm is not well motivated, since
if classical communication is allowed, Alice could send the
initial state to Bob and inform him of the operations that
need to be applied. An interesting open question is if there exist
codes that allow fault-tolerant computation on Alice's
side only.

From this example, we observe that some subset of the
entanglement-assisted codes correct errors on Bob's side.  This
phenomenon can be understood as an instance of the more general
correlation-assisted codes and linear quantum error-correction
theory detailed in Ref.~\cite{SL:07}. It may be useful from a
practical standpoint to determine which entanglement-assisted
codes satisfy this property. Here we provide an answer for the
case of single-error-correcting codes that use one bit of entanglement.

\begin{proposition}
Every $[[n,1,3]]$ code is equivalent
to a $[[n-1,1,3;1]]$ code with any qubit serving as Bob's half of
the ebit.
\end{proposition}

\begin{proof}
We prove this proposition by showing that any column in the table of
stabilizer generators for such a code can be reduced to the
standard form of Bob's column in an entanglement-assisted code
(containing exactly one $X$ and one $Z$ operator). Without loss of
generality, consider the column corresponding to the first qubit.
This column generally may contain $X$, $Y$, $Z$, or $I$ operators,
but if the code corrects any error on the first qubit, there must
be at least two different Pauli operators in this column. We can
reduce this column to the desired form as follows. Choose one of
the generators that contains $X$ on the first qubit, and replace
each of the other generators that contain an $X$ there by its
product with the chosen generator. Do the same for $Y$ and $Z$.
Thus we are left with at most one generator with $X$, one with $Y$
and one with $Z$. To eliminate $Y$, we replace it by its product
with the $X$ and $Z$ generators. If either $X$ or $Z$ is missing,
we replace the $Y$ generator with its product with the other
non-trivial generator.
\end{proof}

This result can be understood as a reflection of the fact that in
a code that corrects arbitrary single-qubit errors, every qubit is
maximally entangled with the rest and therefore can be thought of
as part of an ebit. The latter can also be seen to follow from the
property that every single-qubit error must send the code space to
an orthogonal subspace.

Note that for the case of $[[n,1,3;c]]$ codes with $c>1$, the
relation could be more complicated. If such a code corrects
an arbitrary single-qubit error, it is equivalent to an
$[[n+c,1,3]]$ code, but it is not obvious whether a $[[n+c,1,3]]$ code
can be interpreted as a $[[n,1,3;c]]$ code because the type of
entanglement that exists between $c$ qubits and the rest $n$
qubits may not be the same as that of $c$ e-bits.

\section{Non-existence of $[[n,1,3;1]]$ CSS codes for $n\leq 5$}
\label{sec:ogycss}

We now show that there does not exist a smaller entanglement-assisted
CSS code that uses only one ebit and corrects an arbitrary
single-qubit error on Alice's side. The proof is similar to that
for the non-existence of a $[[6,1,3]]$ CSS code.

\begin{proposition}
There does not exist an $[[n,1,3;1]]$ entanglement-assisted CSS
code for $n\leq 5$.
\end{proposition}

\begin{proof}
We being this proof by giving a dimensionality argument for the non-existence
of quantum codes (CSS or non-CSS) with $n < 4$.  This can be easily seen as follows.
Assume that the code is non-degenerate. There are $3n$ different single-qubit errors on
Alice's side, which means that there must exist $3n+1$ orthogonal
subspaces of dimension $2$ inside the entire $2^{n+1}$-dimensional
Hilbert space, i.e., $(3n+1)2\leq 2^{n+1}$. This is impossible for
$n<4$. Since for $n\leq3$ the number of generators is at most 3,
and two of the generators have to act non-trivially on Bob's side,
we can have degeneracy with respect to errors on Alice's side only
for $n=3$ with exactly one of the generators being equal to a pair
of errors on Alice's side. These two errors would be the only
indistinguishable single-qubit errors on Alice's side (no other
pair of errors on Alice's side can belong to the stabilizer),
which reduces the number of required orthogonal subspaces from
$3\times 3+1=10$ to 9. The required dimensions are $2\times 9=18$
and they cannot fit in the $2^4=16$-dimensional Hilbert space.

Suppose that there exists a $[[5,1,3;1]]$ CSS code. Its stabilizer
must have 5 generators ($S=\langle g_1,..., g_5\rangle$), each
consisting of only $X$ and $I$ operators or $Z$ and $I$ operators.
For an entanglement-assisted code, the generators must be of
the form
\begin{equation}
\begin{tabular}
[c]{cccccc|c}
$g_1=$ & $-$ & $-$ &$-$ & $-$ & $-$ & $X$\\
$g_2=$ & $-$ & $-$ &$-$ & $-$ & $-$ & $Z$\\
$g_3=$ & $-$ & $-$ &$-$ & $-$ & $-$ & $I$\\
$g_4=$ & $-$ & $-$ &$-$ & $-$ & $-$ & $I$ \\
$g_5=$ & $-$ & $-$ &$-$ & $-$ & $-$ & $I$
\end{tabular}
\end{equation}
where we have left the entries on Alice's side unspecified. The
set of correctable Pauli errors on Alice's side $\{E_j\in
\mathcal{P}_5\}$ (where $\mathcal{P}_5$ is the five-qubit Pauli
group) must satisfy $\{E_iE_j,S\}=0$ unless $E_i E_j \in S$, for
all $i,j=1,2,3,4,5$. All generators cannot be of the same type
($X$ or $Z$). The possibility that there is one generator of one
type, say $X$, and four generators of the other ($Z$) type, is
also ruled out because the $X$-type generator would have to be of
the form $g_1=XXXXX|X$ in order that every qubit is acted upon
non-trivially by at least one $X$ operator from the stabilizer.
This would mean, however, that any combination of two $Z$-errors
($Z_iZ_j$, $i,j=1,2,3,4,5$) would commute with the stabilizer, and
so it would have to belong to the stabilizer. There are four
independent such combinations of errors
($Z_1Z_2$,$Z_1Z_3$,$Z_1Z_4$,$Z_1Z_5$) which would have to be the
other four generators. But then there would be no possibility for
a $Z$ operator on Bob's side (as in $g_2$). Therefore, this is
impossible.

The only possibility is that there are 2 generators of one type,
say $X$, and 3 generators of the other type ($Z$). The two
$X$-type generators should not both have identity acting on any
given qubit on Alice's side because a $Z$ error on that qubit would
commute with all generators. Consider the following form for the two $X$-type generators:
\begin{equation}
\begin{tabular}
[c]{cccccc|c}
$g_1=$ & $-$ &$-$ & $-$ & $-$ & $-$ & $X$\\
$g_3=$ & $-$ &$-$ & $-$ & $-$ & $-$ & $I$
\end{tabular}
\end{equation}
There are three different columns that can fill the unspecified
entries in the above table:
\begin{equation}
\begin{matrix}
I\\
X
\end{matrix}\hspace{0.5cm}, \hspace{0.5cm}
\begin{matrix}
X\\
I
\end{matrix}\hspace{0.5cm}, \hspace{0.5cm}
\begin{matrix}
X\\
X
\end{matrix}\hspace{0.5cm}.
\notag
\end{equation}
We distinguish the following cases: two columns appear twice
and one column appears once, one column appears three times and
another column appears twice, one column appears three times and
each of the other columns appears once, at least one column
appears more than three times.

In the first case, up to relabeling of the qubits, we
distinguish the following possibilities:
\begin{equation}
\begin{tabular}
[c]{cccccc|c}
$g_1'=$ & $I$ & $I$ & $X$ & $X$ & $X$ & $X$\\
$g_3'=$ & $X$ & $X$ & $I$ & $I$ & $X$ & $I$
\end{tabular}\label{possibility1}
\end{equation}
\begin{equation}
\begin{tabular}
[c]{cccccc|c}
$g_1''=$ & $X$ &$X$ & $I$ & $I$ & $X$ & $X$\\
$g_3''=$ & $X$ &$X$ & $X$ & $X$ & $I$ & $I$
\end{tabular}\label{possibility2}
\end{equation}
\begin{equation}
\begin{tabular}
[c]{cccccc|c}
$g_1'''=$ & $X$ &$X$ & $X$ & $X$ & $I$ & $X$\\
$g_3'''=$ & $X$ &$X$ & $I$ & $I$ & $X$ & $I$
\end{tabular}\label{possibility3}
\end{equation}
For each possibility, the pairs of errors $Z_1Z_2$ and $Z_3Z_4$
commute with the stabilizer and therefore they would have to be
equal to the stabilizer generators $g_4$ and $g_5$. But the pairs
of errors $X_1X_2$ and $X_3X_4$ would commute with $g_1$, $g_3$,
$g_4$ and $g_5$. Since these errors do not belong to the
stabilizer, they would have to anti-commute with $g_3$. Therefore,
up to interchanging the first and second, or the third and fourth
qubits, the generator $g_2$ must have the form
\begin{equation}
\begin{tabular}
[c]{cccccc|c} $g_3=$ & $Z$ &$I$ & $Z$ & $I$ & $Z$ & $Z$
\end{tabular}.
\end{equation}
(Note that the fifth entry must be $Z$ because there must be at
least one generator that has a $Z$ acting on that qubit.) But it
can be verified that for each of the possibilities
\eqref{possibility1}, \eqref{possibility2} and
\eqref{possibility3}, $g_3$ anti-commutes with one of the $X$-type
generators. Therefore, the first case is impossible.

In the second case, one of the possible columns appears three
times and another column appears twice, e.g.,
\begin{equation}
\begin{tabular}
[c]{cccccc|c}
$g_1=$ & $X$ &$X$ & $X$ & $X$ & $X$ & $X$\\
$g_3=$ & $X$ &$X$ & $X$ & $I$ & $I$ & $I$
\end{tabular}
\end{equation}
In such a case we would have three independent pairs of $Z$ errors
($Z_1Z_2$, $Z_1Z_3$ and $Z_4Z_5$) which commute with the
stabilizer and therefore have to belong to it. But then there
would be no possibility for a $Z$ operator on Bob's side (the
generator $g_2$). Therefore this case is impossible.

In the third case, one column appears three times and each other
column appears once, as in

\begin{equation}
\begin{tabular}
[c]{cccccc|c}
$g_1=$ & $X$ &$X$ & $X$ & $X$ & $I$ & $X$\\
$g_3=$ & $X$ &$X$ & $X$ & $I$ & $X$ & $I$
\end{tabular}
\end{equation}
In this case, the pairs of errors $Z_1Z_2$ and $Z_1Z_3$ commute
with the stabilizer and must be equal to $g_4$ and $g_5$. But in
order for the fourth and fifth qubits to be each acted upon by at
least one $Z$ operator from the stabilizer, the generator $g_2$
would have to be of the form
\begin{equation}
\begin{tabular}
[c]{cccccc|c} $g_2=$ & $-$ &$-$ & $-$ & $Z$ & $Z$ & $Z$
\end{tabular}
\end{equation}
This means that the pair of errors $X_4X_5$ commutes with the
stabilizer, and since it is not part of the stabilizer, this case
is also impossible.

Finally, if one column appears more than three times, there would
be at least three independent pairs of $Z$ errors on Alice's side
which have to belong to the stabilizer. This leaves no possibility
for a $Z$ operator on Bob's side, i.e., this case is also ruled
out. Therefore, a $[[5,1,3;1]]$ CSS code does not exist.

In a similar way we can show that a $[[4,1,3;1]]$ CSS code does not
exist. Such a code would have 4 generators of the form
\begin{equation}
\begin{tabular}
[c]{ccccc|c}
$g_1=$ &  $-$ &$-$ & $-$ & $-$ & $X$\\
$g_2=$ &  $-$ &$-$ & $-$ & $-$ & $Z$\\
$g_3=$ &  $-$ &$-$ & $-$ & $-$ & $I$\\
$g_4=$ &  $-$ &$-$ & $-$ & $-$ & $I$
\end{tabular}
\end{equation}
The possibilities that all of the generators are of the same type,
or that one generator is of one type and the other three are of
the other type, are readily ruled out by arguments similar to
those for the $[[5,1,3;1]]$ code. The only possibility is two
$X$-type generators and two $Z$-type generators. The table of the
$X$-type generators
\begin{equation}
\begin{tabular}
[c]{ccccc|c}
$g_1=$  &$-$ & $-$ & $-$ & $-$ & $X$\\
$g_3=$  &$-$ & $-$ & $-$ & $-$ & $I$
\end{tabular}
\end{equation}
has to be filled by the same three columns we discussed before. As
we saw in our previous arguments, in the case when one column
appears three or more times there are at least two independent
pairs of errors on Alice's side which commute with the stabilizer.
These errors would have to belong to the stabilizer, but this
leaves no possibility for a $Z$ operator on Bob's side. In the
case when one column appears twice and another column appears
twice, the situation is analogous. The only other case is when one
column appears twice and each of the other two columns appears
once, as in
\begin{equation}
\begin{tabular}
[c]{ccccc|c}
$g_1=$  &$X$ & $X$ & $I$ & $X$ & $X$\\
$g_3=$  &$X$ & $X$ & $X$ & $I$ & $I$
\end{tabular}
\end{equation}
Since in this case the pair of errors $Z_1Z_2$ would commute with
the stabilizer, this pair would have to be equal to the generator
$g_4$. The third and fourth qubits each have to be acted upon by
at least one $Z$ operators from the stabilizer. Thus the generator
$g_2$ would have to have the form
\begin{equation}
\begin{tabular}
[c]{ccccc|c} $g_2=$  &$-$ & $-$ & $Z$ & $Z$ & $Z$
\end{tabular}.
\end{equation}
But then the pair $X_3X_4$ which does not belong to the stabilizer
would commute with all stabilizer generators. Therefore a
$[[4,1,3;1]]$ CSS code does not exist.\end{proof}

We point out that a $[[4,1,3;1]]$ non-CSS code was found in
Ref.~\cite{science2006brun}. This is the smallest
possible code that can encode one qubit with the use of only one
ebit, and at the same time correct an arbitrary single-qubit error
on Alice's side. Here we have identified an example of the smallest
possible CSS code with these characteristics.

\section{Entanglement-Assisted Encoding Circuit}
\label{sec:eaqeccircuit}
Here we detail an algorithm that generates the encoding circuit for the $[[6,1,3;1]]$ code.
We follow the recipe outlined in the appendix of Ref.~\cite{arx2007wilde}.
We begin by first converting the stabilizer generators in Table~\ref{tbl:ea-613}(b) into a
binary form which we refer to as a $Z|X$ matrix.  We obtain the the left $Z$ submatrix
by inserting a ``1'' wherever we see a $Z$ in the stabilizer generators.  We obtain the
$X$ submatrix by inserting a ``1'' wherever we see a corresponding $X$
in the stabilizer generator.  If there is a $Y$ in the generator, we insert a ``1''
in the corresponding row and column of both the $Z$ and $X$ submatrices.

The idea is to convert matrix~\ref{appendixeqn1} to matrix~\ref{appendixeqn17} through a series of row and column operations.
The binary form of the matrix in~\ref{appendixeqn1} corresponds to the stabilizer generators in Table~\ref{tbl:ea-613}(a).
We can use CNOT, Hadamard, Phase, and SWAP gates.
\begin{enumerate}
\item When we apply a CNOT gate from qubit $i$ to qubit $j$, it adds column $i$ to
column $j$ in the $X$ submatrix, and in the $Z$ submatrix it adds column $j$ to column $i$.
\item A Hadamard on qubit $i$ swaps
column $i$ in the $Z$ submatrix with column $i$ in the $X$ submatrix.
\item A Phase gate on qubit
$i$ adds column $i$ in the $X$ submatrix to column $i$ in the $Z$ submatrix.
\item When we apply a
SWAP gate from qubit $i$ to qubit $j$, we exchange column $i$ with column $j$ in $Z$ submatrix
and column $i$ and column $j$ in the $X$ submatrix.
\end{enumerate}
Row operations do not change the error-correcting properties of the code.  They do not
cost us in terms of gates.  They are also crucial in determining the minimum number of ebits for the code.

\begin{equation}
\left[  \left.
\begin{array}
[c]{cccccc}
1 & 0 & 0 & 1 & 0 & 1 \\
1 & 0 & 1 & 1 & 1 & 0 \\
1 & 1 & 0 & 0 & 1 & 1 \\
0 & 0 & 0 & 0 & 0 & 0 \\
0 & 0 & 0 & 0 & 0 & 0 \\
0 & 0 & 0 & 0 & 0 & 0
\end{array}
\right\vert
\begin{array}
[c]{cccccc}
0 & 0 & 0 & 0 & 0 & 0 \\
0 & 0 & 0 & 0 & 0 & 0 \\
0 & 0 & 0 & 0 & 0 & 0 \\
1 & 1 & 0 & 0 & 1 & 1 \\
1 & 0 & 1 & 1 & 1 & 0 \\
1 & 0 & 0 & 1 & 0 & 1
\end{array}
\right]
\label{appendixeqn1}
\end{equation}
We begin the algorithm by computing the symplectic product \cite{arx2006brun} between the various rows of the matrix.
The first row is symplectically orthogonal to the second row.  Moreover, it is symplectically orthogonal
to all the rows except row six.  So we swap the second row with the sixth row.
\begin{equation}
\left[  \left.
\begin{array}
[c]{cccccc}
1 & 0 & 0 & 1 & 0 & 1 \\
0 & 0 & 0 & 0 & 0 & 0 \\
1 & 1 & 0 & 0 & 1 & 1 \\
0 & 0 & 0 & 0 & 0 & 0 \\
0 & 0 & 0 & 0 & 0 & 0 \\
1 & 0 & 1 & 1 & 1 & 0
\end{array}
\right\vert
\begin{array}
[c]{cccccc}
0 & 0 & 0 & 0 & 0 & 0 \\
1 & 0 & 0 & 1 & 0 & 1 \\
0 & 0 & 0 & 0 & 0 & 0 \\
1 & 1 & 0 & 0 & 1 & 1 \\
1 & 0 & 1 & 1 & 1 & 0 \\
0 & 0 & 0 & 0 & 0 & 0
\end{array}
\right]
\label{appendixeqn2}
\end{equation}
Now apply Hadamard gates to qubits, one, four and six.
This operation swaps the columns one, four and six on the $Z$ side with columns
one, four and six on the $X$ side.
\begin{equation}
\left[  \left.
\begin{array}
[c]{cccccc}
0 & 0 & 0 & 0 & 0 & 0 \\
1 & 0 & 0 & 1 & 0 & 1 \\
0 & 1 & 0 & 0 & 1 & 0 \\
1 & 0 & 0 & 0 & 0 & 1 \\
1 & 0 & 0 & 1 & 0 & 0 \\
0 & 0 & 1 & 0 & 1 & 0
\end{array}
\right\vert
\begin{array}
[c]{cccccc}
1 & 0 & 0 & 1 & 0 & 1 \\
0 & 0 & 0 & 0 & 0 & 0 \\
1 & 0 & 0 & 0 & 0 & 1 \\
0 & 1 & 0 & 0 & 1 & 0 \\
0 & 0 & 1 & 0 & 1 & 0 \\
1 & 0 & 0 & 1 & 0 & 0
\end{array}
\right]
\label{appendixeqn3}
\end{equation}
Apply a CNOT from qubit one to qubit four and a CNOT from qubit one to qubit six.
This operation adds column one to four and column one to column six on the $X$ side.  On the $Z$
side of the matrix, the CNOT operation adds column four to column one and column six to column one.
\begin{equation}
\left[  \left.
\begin{array}
[c]{cccccc}
0 & 0 & 0 & 0 & 0 & 0 \\
1 & 0 & 0 & 1 & 0 & 1 \\
0 & 1 & 0 & 0 & 1 & 0 \\
0 & 0 & 0 & 0 & 0 & 1 \\
0 & 0 & 0 & 1 & 0 & 0 \\
0 & 0 & 1 & 0 & 1 & 0
\end{array}
\right\vert
\begin{array}
[c]{cccccc}
1 & 0 & 0 & 0 & 0 & 0 \\
0 & 0 & 0 & 0 & 0 & 0 \\
1 & 0 & 0 & 1 & 0 & 0 \\
0 & 1 & 0 & 0 & 1 & 0 \\
0 & 0 & 1 & 0 & 1 & 0 \\
1 & 0 & 0 & 0 & 0 & 1
\end{array}
\right]
\label{appendixeqn4}
\end{equation}
Now apply a Hadamard gate on qubit one.
\begin{equation}
\left[  \left.
\begin{array}
[c]{cccccc}
1 & 0 & 0 & 0 & 0 & 0 \\
0 & 0 & 0 & 1 & 0 & 1 \\
0 & 1 & 0 & 0 & 1 & 0 \\
0 & 0 & 0 & 0 & 0 & 1 \\
0 & 0 & 0 & 1 & 0 & 0 \\
0 & 0 & 1 & 0 & 1 & 0
\end{array}
\right\vert
\begin{array}
[c]{cccccc}
0 & 0 & 0 & 0 & 0 & 0 \\
1 & 0 & 0 & 0 & 0 & 0 \\
0 & 0 & 0 & 1 & 0 & 0 \\
0 & 1 & 0 & 0 & 1 & 0 \\
0 & 0 & 1 & 0 & 1 & 0 \\
0 & 0 & 0 & 0 & 0 & 1
\end{array}
\right]
\label{appendixeqn5}
\end{equation}
Apply a Hadamard gate on qubit four and qubit six.  This operation swaps columns four and six on
$Z$ side with respective columns on the $X$ side.
\begin{equation}
\left[  \left.
\begin{array}
[c]{cccccc}
1 & 0 & 0 & 0 & 0 & 0 \\
0 & 0 & 0 & 0 & 0 & 0 \\
0 & 1 & 0 & 1 & 1 & 0 \\
0 & 0 & 0 & 0 & 0 & 0 \\
0 & 0 & 0 & 0 & 0 & 0 \\
0 & 0 & 1 & 0 & 1 & 1
\end{array}
\right\vert
\begin{array}
[c]{cccccc}
0 & 0 & 0 & 0 & 0 & 0 \\
1 & 0 & 0 & 1 & 0 & 1 \\
0 & 0 & 0 & 0 & 0 & 0 \\
0 & 1 & 0 & 0 & 1 & 1 \\
0 & 0 & 1 & 1 & 1 & 0 \\
0 & 0 & 0 & 0 & 0 & 0
\end{array}
\right]
\label{appendixeqn6}
\end{equation}
Finally, we apply a CNOT gate from qubit one to qubit four and another CNOT gate
from qubit one to qubit six.
\begin{equation}
\left[  \left.
\begin{array}
[c]{cccccc}
1 & 0 & 0 & 0 & 0 & 0 \\
0 & 0 & 0 & 0 & 0 & 0 \\
0 & 1 & 0 & 1 & 1 & 0 \\
0 & 0 & 0 & 0 & 0 & 0 \\
0 & 0 & 0 & 0 & 0 & 0 \\
0 & 0 & 1 & 0 & 1 & 1
\end{array}
\right\vert
\begin{array}
[c]{cccccc}
0 & 0 & 0 & 0 & 0 & 0 \\
1 & 0 & 0 & 0 & 0 & 0 \\
0 & 0 & 0 & 0 & 0 & 0 \\
0 & 1 & 0 & 0 & 1 & 1 \\
0 & 0 & 1 & 1 & 1 & 0 \\
0 & 0 & 0 & 0 & 0 & 0
\end{array}
\right]
\label{appendixeqn7}
\end{equation}
At this point we are done processing qubit one and qubit two.  We now proceed to manipulate columns
two through six on the $Z$ and $X$ side.  We apply a Hadamard gate on qubit two, four and five.
\begin{equation}
\left[  \left.
\begin{array}
[c]{cccccc}
1 & 0 & 0 & 0 & 0 & 0 \\
0 & 0 & 0 & 0 & 0 & 0 \\
0 & 0 & 0 & 0 & 0 & 0 \\
0 & 1 & 0 & 0 & 1 & 0 \\
0 & 0 & 0 & 1 & 1 & 0 \\
0 & 0 & 1 & 0 & 0 & 1
\end{array}
\right\vert
\begin{array}
[c]{cccccc}
0 & 0 & 0 & 0 & 0 & 0 \\
1 & 0 & 0 & 0 & 0 & 0 \\
0 & 1 & 0 & 1 & 1 & 0 \\
0 & 0 & 0 & 0 & 0 & 1 \\
0 & 0 & 1 & 0 & 0 & 0 \\
0 & 0 & 0 & 0 & 1 & 0
\end{array}
\right]
\label{appendixeqn8}
\end{equation}
Perform a CNOT gate from qubit two to qubit four and from qubit two to qubit five.
\begin{equation}
\left[  \left.
\begin{array}
[c]{cccccc}
1 & 0 & 0 & 0 & 0 & 0 \\
0 & 0 & 0 & 0 & 0 & 0 \\
0 & 0 & 0 & 0 & 0 & 0 \\
0 & 0 & 0 & 0 & 1 & 0 \\
0 & 0 & 0 & 1 & 1 & 0 \\
0 & 0 & 1 & 0 & 0 & 1
\end{array}
\right\vert
\begin{array}
[c]{cccccc}
0 & 0 & 0 & 0 & 0 & 0 \\
1 & 0 & 0 & 0 & 0 & 0 \\
0 & 1 & 0 & 0 & 0 & 0 \\
0 & 0 & 0 & 0 & 0 & 1 \\
0 & 0 & 1 & 0 & 0 & 0 \\
0 & 0 & 0 & 0 & 1 & 0
\end{array}
\right]
\label{appendixeqn9}
\end{equation}
Perform a Hadamard on qubit two.
\begin{equation}
\left[  \left.
\begin{array}
[c]{cccccc}
1 & 0 & 0 & 0 & 0 & 0 \\
0 & 0 & 0 & 0 & 0 & 0 \\
0 & 1 & 0 & 0 & 0 & 0 \\
0 & 0 & 0 & 0 & 1 & 0 \\
0 & 0 & 0 & 1 & 1 & 0 \\
0 & 0 & 1 & 0 & 0 & 1
\end{array}
\right\vert
\begin{array}
[c]{cccccc}
0 & 0 & 0 & 0 & 0 & 0 \\
1 & 0 & 0 & 0 & 0 & 0 \\
0 & 0 & 0 & 0 & 0 & 0 \\
0 & 0 & 0 & 0 & 0 & 1 \\
0 & 0 & 1 & 0 & 0 & 0 \\
0 & 0 & 0 & 0 & 1 & 0
\end{array}
\right]
\label{appendixeqn10}
\end{equation}
We have processed qubit three.  Now look at the submatrix from columns three to
six on the $Z$ and $X$ side.  Perform a SWAP gate between qubit three and qubit five.
This operation swaps column three with five in the $Z$ submatrix and column three and
five in the $X$ submatrix.
\begin{equation}
\left[  \left.
\begin{array}
[c]{cccccc}
1 & 0 & 0 & 0 & 0 & 0 \\
0 & 0 & 0 & 0 & 0 & 0 \\
0 & 1 & 0 & 0 & 0 & 0 \\
0 & 0 & 1 & 0 & 0 & 0 \\
0 & 0 & 0 & 1 & 1 & 0 \\
0 & 0 & 1 & 0 & 0 & 1
\end{array}
\right\vert
\begin{array}
[c]{cccccc}
0 & 0 & 0 & 0 & 0 & 0 \\
1 & 0 & 0 & 0 & 0 & 0 \\
0 & 0 & 0 & 0 & 0 & 0 \\
0 & 0 & 0 & 0 & 0 & 1 \\
0 & 0 & 0 & 0 & 1 & 0 \\
0 & 0 & 1 & 0 & 0 & 0
\end{array}
\right]
\label{appendixeqn11}
\end{equation}
Perform a Hadamard gate on qubit three, followed by a CNOT gate from qubit three to qubit six,
and another Hadamard on qubit three.
\begin{equation}
\left[  \left.
\begin{array}
[c]{cccccc}
1 & 0 & 0 & 0 & 0 & 0 \\
0 & 0 & 0 & 0 & 0 & 0 \\
0 & 1 & 0 & 0 & 0 & 0 \\
0 & 0 & 1 & 0 & 0 & 0 \\
0 & 0 & 1 & 1 & 0 & 0 \\
0 & 0 & 0 & 0 & 1 & 1
\end{array}
\right\vert
\begin{array}
[c]{cccccc}
0 & 0 & 0 & 0 & 0 & 0 \\
1 & 0 & 0 & 0 & 0 & 0 \\
0 & 0 & 0 & 0 & 0 & 0 \\
0 & 0 & 0 & 0 & 0 & 0 \\
0 & 0 & 0 & 0 & 1 & 1 \\
0 & 0 & 0 & 0 & 0 & 0
\end{array}
\right]
\label{appendixeqn12}
\end{equation}
Add row four to five.
\begin{equation}
\left[  \left.
\begin{array}
[c]{cccccc}
1 & 0 & 0 & 0 & 0 & 0 \\
0 & 0 & 0 & 0 & 0 & 0 \\
0 & 1 & 0 & 0 & 0 & 0 \\
0 & 0 & 1 & 0 & 0 & 0 \\
0 & 0 & 0 & 1 & 0 & 0 \\
0 & 0 & 0 & 0 & 1 & 1
\end{array}
\right\vert
\begin{array}
[c]{cccccc}
0 & 0 & 0 & 0 & 0 & 0 \\
1 & 0 & 0 & 0 & 0 & 0 \\
0 & 0 & 0 & 0 & 0 & 0 \\
0 & 0 & 0 & 0 & 0 & 0 \\
0 & 0 & 0 & 0 & 1 & 1 \\
0 & 0 & 0 & 0 & 0 & 0
\end{array}
\right]
\label{appendixeqn13}
\end{equation}
We have completed processing qubit four. Now focus on columns four to six. Apply a Hadamard gate on qubit four,
followed by CNOT gate from qubit four to qubit five, and again from qubit four to qubit six.
\begin{equation}
\left[  \left.
\begin{array}
[c]{cccccc}
1 & 0 & 0 & 0 & 0 & 0 \\
0 & 0 & 0 & 0 & 0 & 0 \\
0 & 1 & 0 & 0 & 0 & 0 \\
0 & 0 & 1 & 0 & 0 & 0 \\
0 & 0 & 0 & 0 & 0 & 0 \\
0 & 0 & 0 & 0 & 1 & 1
\end{array}
\right\vert
\begin{array}
[c]{cccccc}
0 & 0 & 0 & 0 & 0 & 0 \\
1 & 0 & 0 & 0 & 0 & 0 \\
0 & 0 & 0 & 0 & 0 & 0 \\
0 & 0 & 0 & 0 & 0 & 0 \\
0 & 0 & 0 & 1 & 0 & 0 \\
0 & 0 & 0 & 0 & 0 & 0
\end{array}
\right]
\label{appendixeqn14}
\end{equation}
Perform a Hadamard gate on qubit four.
\begin{equation}
\left[  \left.
\begin{array}
[c]{cccccc}
1 & 0 & 0 & 0 & 0 & 0 \\
0 & 0 & 0 & 0 & 0 & 0 \\
0 & 1 & 0 & 0 & 0 & 0 \\
0 & 0 & 1 & 0 & 0 & 0 \\
0 & 0 & 0 & 1 & 0 & 0 \\
0 & 0 & 0 & 0 & 1 & 1
\end{array}
\right\vert
\begin{array}
[c]{cccccc}
0 & 0 & 0 & 0 & 0 & 0 \\
1 & 0 & 0 & 0 & 0 & 0 \\
0 & 0 & 0 & 0 & 0 & 0 \\
0 & 0 & 0 & 0 & 0 & 0 \\
0 & 0 & 0 & 0 & 0 & 0 \\
0 & 0 & 0 & 0 & 0 & 0
\end{array}
\right]
\label{appendixeqn15}
\end{equation}
Now look at columns five and six.  Apply a Hadamard gate on qubit five and qubit six, followed by a
CNOT gate from qubit five to qubit six.
\begin{equation}
\left[  \left.
\begin{array}
[c]{cccccc}
1 & 0 & 0 & 0 & 0 & 0 \\
0 & 0 & 0 & 0 & 0 & 0 \\
0 & 1 & 0 & 0 & 0 & 0 \\
0 & 0 & 1 & 0 & 0 & 0 \\
0 & 0 & 0 & 1 & 0 & 0 \\
0 & 0 & 0 & 0 & 0 & 0
\end{array}
\right\vert
\begin{array}
[c]{cccccc}
0 & 0 & 0 & 0 & 0 & 0 \\
1 & 0 & 0 & 0 & 0 & 0 \\
0 & 0 & 0 & 0 & 0 & 0 \\
0 & 0 & 0 & 0 & 0 & 0 \\
0 & 0 & 0 & 0 & 0 & 0 \\
0 & 0 & 0 & 0 & 1 & 0
\end{array}
\right]
\label{appendixeqn16}
\end{equation}
Perform a Hadamard on qubit five.
\begin{equation}
\left[  \left.
\begin{array}
[c]{cccccc}
1 & 0 & 0 & 0 & 0 & 0 \\
0 & 0 & 0 & 0 & 0 & 0 \\
0 & 1 & 0 & 0 & 0 & 0 \\
0 & 0 & 1 & 0 & 0 & 0 \\
0 & 0 & 0 & 1 & 0 & 0 \\
0 & 0 & 0 & 0 & 1 & 0
\end{array}
\right\vert
\begin{array}
[c]{cccccc}
0 & 0 & 0 & 0 & 0 & 0 \\
1 & 0 & 0 & 0 & 0 & 0 \\
0 & 0 & 0 & 0 & 0 & 0 \\
0 & 0 & 0 & 0 & 0 & 0 \\
0 & 0 & 0 & 0 & 0 & 0 \\
0 & 0 & 0 & 0 & 0 & 0
\end{array}
\right]
\label{appendixeqn17}
\end{equation}
We have finally obtained a binary matrix that corresponds to the canonical stabilizer
generators in Table~\ref{tbl:ea-613}(a).  Figure~\ref{fig-eacodecircuit} gives the encoding
circuit for the all the quantum operations that we performed above.  Multiplying the above
operations in reverse takes us from the unencoded canonical stabilizers to the encoded ones.


\begin{table*}[tbp] \centering
\begin{tabular}
[c]{|c|c||c|c||c|c||c|c||c|c||c|c|}\hline\hline
Error & AG & Error & AG & Error & AG & Error & AG & Error & AG & Error & AG\\\hline\hline
$X_1X_2$ & $h_1$ & $X_1X_3$ & $h_2$ & $X_1X_4$ & $h_1$ &
$X_1X_5$ & $h_1$ & $X_1X_6$ & $h_5$ & $X_1Y_2$ & $h_1$
\\\hline
$X_1Y_3$ & $h_2$ & $X_1Y_4$ & $h_2$ & $X_1Y_5$ & $h_3$ &
$X_1Y_6$ & $h_1$ & $X_1Z_2$ & $h_1$ & $X_1Z_3$ & $h_1$
\\\hline
$X_1Z_4$ & $h_2$ & $X_1Z_5$ & $h_3$ & $X_1Z_6$ & $h_2$ &
$X_2X_3$ & $h_1$ & $X_2X_4$ & $h_3$ & $X_2X_5$ & $h_3$
\\\hline
$X_2X_6$ & $h_1$ & $X_2Y_3$ & $h_1$ & $X_2Y_4$ & $h_1$ &
$X_2Y_5$ & $h_1$ & $X_2Y_6$ & $h_2$ & $X_2Z_3$ & $h_5$
\\\hline
$X_2Z_4$ & $h_1$ & $X_2Z_5$ & $h_1$ & $X_2Z_6$ & $h_1$ &
$X_3X_4$ & $h_1$ & $X_3X_5$ & $h_1$ & $X_3X_6$ & $h_2$
\\\hline
$X_3Y_4$ & $h_3$ & $X_3Y_5$ & $h_2$ & $X_3Y_6$ & $h_1$ &
$X_3Z_4$ & $h_3$ & $X_3Z_5$ & $h_2$ & $X_3Z_6$ & $h_3$
\\\hline
$X_4X_5$ & $h_5$ & $X_4X_6$ & $h_1$ & $X_4Y_5$ & $h_1$ &
$X_4Y_6$ & $h_2$ & $X_4Z_5$ & $h_1$ & $X_4Z_6$ & $h_1$
\\\hline
$X_5X_6$ & $h_1$ & $X_5Y_6$ & $h_2$ & $X_5Z_6$ & $h_1$ &
$Y_1X_2$ & $h_2$ & $Y_1X_3$ & $h_1$ & $Y_1X_4$ & $h_2$
\\\hline
$Y_1X_5$ & $h_2$ & $Y_1X_6$ & $h_1$ & $Y_1Y_2$ & $h_3$ &
$Y_1Y_3$ & $h_1$ & $Y_1Y_4$ & $h_1$ & $Y_1Y_5$ & $h_1$
\\\hline
$Y_1Y_6$ & $h_3$ & $Y_1Z_2$ & $h_5$ & $Y_1Z_3$ & $h_2$ &
$Y_1Z_4$ & $h_1$ & $Y_1Z_5$ & $h_1$ & $Y_1Z_6$ & $h_1$
\\\hline
$Y_2X_3$ & $h_1$ & $y_2X_4$ & $h_2$ & $Y_2X_5$ & $h_2$ &
$Y_2X_6$ & $h_1$ & $Y_2Y_3$ & $h_1$ & $Y_2Y_4$ & $h_1$
\\\hline
$Y_2Y_5$ & $h_1$ & $Y_2Y_6$ & $h_5$ & $Y_2Z_3$ & $h_5$ &
$Y_2Z_4$ & $h_1$ & $Y_2Z_5$ & $h_1$ & $Y_2Z_6$ & $h_1$
\\\hline
$Y_3X_4$ & $h_1$ & $Y_3X_5$ & $h_1$ & $Y_3X_6$ & $h_2$ &
$Y_3Y_4$ & $h_5$ & $Y_3Y_5$ & $h_2$ & $Y_3Y_6$ & $h_1$
\\\hline
$Y_3Z_4$ & $h_5$ & $Y_3Z_5$ & $h_2$ & $Y_3Z_6$ & $h_5$ &
$Y_4X_5$ & $h_1$ & $Y_4X_6$ & $h_2$ & $Y_4Y_5$ & $h_2$
\\\hline
$Y_4Y_6$ & $h_1$ & $Y_4Z_5$ & $h_2$ & $Y_4Z_6$ & $h_4$ &
$Y_5X_6$ & $h_3$ & $Y_5Y_6$ & $h_1$ & $Y_5Z_6$ & $h_2$
\\\hline
$Z_1X_2$ & $h_1$ & $Z_1X_3$ & $h_5$ & $Z_1X_4$ & $h_1$ &
$Z_1X_5$ & $h_1$ & $Z_1X_6$ & $h_2$ & $Z_1Y_2$ & $h_1$
\\\hline
$Z_1Y_3$ & $h_3$ & $Z_1Y_4$ & $h_3$ & $Z_1Y_5$ & $h_2$ &
$Z_1Y_6$ & $h_1$ & $Z_1Z_2$ & $h_1$ & $Z_1Z_3$ & $h_1$
\\\hline
$Z_1Z_4$ & $h_3$ & $Z_1Z_5$ & $h_2$ & $Z_1Z_6$ & $h_3$ &
$Z_2X_3$ & $h_1$ & $Z_2X_4$ & $h_2$ & $Z_2X_5$ & $h_2$
\\\hline
$Z_2X_6$ & $h_1$ & $Z_2Y_3$ & $h_1$ & $Z_2Y_4$ & $h_1$ &
$Z_2Y_5$ & $h_1$ & $Z_2Y_6$ & $h_3$ & $Z_2Z_3$ & $h_2$
\\\hline
$Z_2Z_4$ & $h_1$ & $Z_2Z_5$ & $h_1$ & $Z_2Z_6$ & $h_1$ &
$Z_3X_4$ & $h_3$ & $Z_3X_5$ & $h_3$ & $Z_3X_6$ & $h_1$
\\\hline
$Z_3Y_4$ & $h_1$ & $Z_3Y_5$ & $h_1$ & $Z_3Y_6$ & $h_2$ &
$Z_3Z_4$ & $h_1$ & $Z_3Z_5$ & $h_1$ & $Z_3Z_6$ & $h_1$
\\\hline
$Z_4X_5$ & $h_1$ & $Z_4X_6$ & $h_2$ & $Z_4Y_5$ & $h_2$ &
$Z_4Y_6$ & $h_1$ & $Z_4Z_5$ & $h_2$ & $Z_4Z_6$ & $h_4$
\\\hline
$Z_5X_6$ & $h_3$ & $Z_5Y_6$ & $h_1$ & $Z_5Z_6$ & $h_2$ &
$X_1$ & $h_1$ & $X_2$ & $h_3$ & $X_3$ & $h_1$
\\\hline
$X_4$ & $h_4$ & $X_5$ & $h_5$ & $X_6$ & $h_1$ &
$Y_1$ & $h_2$ & $Y_2$ & $h_2$ & $Y_3$ & $h_3$
\\\hline
$Y_4$ & $h_3$ & $Y_5$ & $h_3$ & $Y_6$ & $h_3$ &
$Z_1$ & $h_1$ & $Z_2$ & $h_2$ & $Z_3$ & $h_3$
\\\hline
$Z_4$ & $h_3$ & $Z_5$ & $h_3$ & $Z_6$ & $h_1$ \\ \hline\hline
\end{tabular}
\caption{Distinct pairs of single-qubit Pauli errors for the $[[6,1,3]]$
quantum code. Each double-lined column lists a pair of single-qubit
errors and a corresponding anticommuting generator (AG) for the code.  $X_4$ and $Z_4Z_6$ lie in the gauge subgroup $H$.}\label{table:sixonethree1}%
\end{table*}%

\section{Concluding Remarks}
We have discussed two different examples of a six-qubit code
and have included a subsystem construction for the degenerate six-qubit code.  Our proof
explains why a six-qubit CSS code does not exist and clarifies
earlier results in Ref.~\cite{ieee1998calderbank} based on a
search algorithm. An immediate corollary of our result is that the
seven-qubit Steane code is the smallest CSS code capable of
correcting an arbitrary single-qubit error. An interesting open
problem is to generalize this tight lower bound to the setting of
CSS codes with a higher distance. We expect that our proof
technique may be useful for this purpose.

Our first example is a degenerate six-qubit code that corrects an arbitrary
single-qubit error. The presentation
of the encoding circuit and the operations required for a logical
$X$, $Z$, and CNOT\ should aid in the implementation and
operation of this code.  We have converted this code into
a subsystem code that is non-trivial and saturates the subsystem Singleton bound. Our six-qubit subsystem
code requires only four stabilizer measurements during the
recovery process.  This reduction in measurements may have implications for improving
fault-tolerance thresholds.

Our second example is an entanglement-assisted $[[6,1,3;1]]$ CSS
code that is globally equivalent to the Steane seven-qubit code.
We have presented the construction of this code from a set of six
non-commuting generators on six qubits. We have further shown that
every $[[n,1,3]]$ code can be used as a $[[n-1,1,3;1]]$
entanglement-assisted code.

Based on the proof technique that we used for the earlier
six-qubit code, we have shown that the Steane code is an example
of the smallest entanglement-assisted code that possesses the
CSS structure and uses exactly one ebit.  Here too, an interesting
open problem is the generalization of this tight lower
bound to higher distance entanglement-assisted codes or to codes
that use more than one ebit.

From the next chapter we begin the theory of quantum steganography.  It was important for us to
introduce the reader to stabilizer codes as we use them extensively in quantum steganography
to hide quantum information.

\chapter{Classical Steganography}
\label{chap:classicalsteganography}
\begin{saying}
What is that confers the noblest delight? What is that which swells a man's breast with pride above that which any other experience can bring to him? Discovery! To know that you are walking where none others have walked; that you are beholding what human eye has not seen before; that you are breathing a virgin atmosphere. To give birth to an idea -- to discover a great thought -- an intellectual nugget, right under the dust of a field that many a brain -- plow had gone over before. To find a new planet, to invent a new hinge, to find the way to make the lightnings carry your messages. To be the first -- that is the idea. To do something, say something, see something, before any body else -- these are the things that confer a pleasure compared with which other pleasures are tame and commonplace, other ecstasies cheap and trivial.\\
---\textit{Mark Twain}
\end{saying}

\section{Introduction}
\label{sec:qsteg_intro}
\lettrine{W}e begin the second half of the thesis with a classical model of steganography which we later extend to a quantum one in Chapter~\ref{chap:quantumsteg}.  Steganography is the science of hiding a message within a larger innocent-looking plain-text message, and communicating the resulting data over a communications channel or through a courier so that the steganographic message is readable only by the intended receiver.  The word comes from the Greek words \textit{steganos} which means ``covered,'' and \textit{graphia} which means ``writing.''  The art of information hiding dates back to 440 B.C. to the Greeks.  In \textit{The Histories}, Herodotus records two incidents of the use of steganography.  In the first incident Demaratus a Greek king uses a wax tablet to warn the Spartans of an impending attack by the Persian king Xerxes~\cite{Herodotus}.  Wax tablets were used as reusable surfaces to write on, which were constructed on wooden bases.  Demaratus scratched the steganographic message on the wood, and then covered it with beeswax.  Once the Spartans received the wax tablet from the courier, all that they needed to do was to melt the beeswax and read the hidden warning.  In another story Herodotus records how Histiaeus tattoos a secret message on the shaved scalp of his slave, and then waits for the hair to grow back before dispatching him to the Ionian city of Miletus.  After arriving in Miletus, the slave shaves his head and reveals the secret message to the city's regent, Aristagorus, who upon seeing the message is encouraged to rise up against the Persian king.  The ancient Chinese used a wooden mask with holes cut out at random places to use as a steganographic device.  They would place the wooden block on a blank sheet of paper and after writing the secret message in the holes they would fill in the blanks in the paper with regular text.  The mask acted as a secret key to unlock the hidden message.  In the seventeenth and eighteenth centuries people used logarithmic tables to hide secret messages by introducing errors in the least significant digits.

The term steganography was first used in 1499 by Johannes Trithemius in his \textit{Trithemius} which was one of the first treatises on the use of cryptographic and steganographic techniques.  The study of modern steganography was initiated by Simmons and can be stated as follows~\cite{simmons1983}.  Alice and Bob are imprisoned in two different cells that are far apart.  They would like to devise an escape plan but the only way they can communicate with each other is through a courier who is loyal to and under the command of the warden (Eve - adversary) of the penitentiary.  The courier leaks all information to the warden.  If the warden suspects that either Alice or Bob are conspiring to escape from the penitentiary, she will cut off all communication between them, and move both of them to a maximum security cell.  It is assumed that prior to their incarceration Alice and Bob had access to a shared secret key which they will later exploit to send secret messages hidden in a cover text.  Can Alice and Bob devise an escape plan without arousing the suspicion of the warden?

We should point out that steganography is inherently different from cryptography.  In the latter, using an encryption algorithm the sender encrypts the secret message utilizing a key (private or public).  The resulting message from such an encrypting procedure is called a ``cipher-text.''  To an eavesdropper such as Eve, who does not have access to the secret key, the cipher-text looks like gibberish.  If Eve observes that the transmitted cipher-text message is gibberish, she might realize that the message contains private information.  In stark contrast to the latter is steganography where we do not necessarily need to encrypt the intended secret message.  We merely hide it within a larger plain-text message, often referred to as the ``cover-text'' or ``cover-work.''  The resulting message, called the ``stego-text'' must appear as an oblivious or benign message to Eve.  In the usual steganographic protocol we assume that Alice and Bob have access to a secure, shared secret key.  Alice uses this key and an embedding function to hide a secret message, while Bob uses the same key to extract the secret message.  Alice can take this a step further by first encrypting the message and then hiding the resulting cipher-text inside the cover-text.  She may use the same key to encrypt the secret message producing a cipher-text and then hide it within a cover-text to produce a stego-text message.  Bob, on receiving the stego-text message would first have to extract the secret cipher-text and then run a decryption algorithm to obtain the final secret message.

\begin{figure}
[ptb]
\begin{center}
\includegraphics[
natheight=3.386600in,
natwidth=8.973300in,
height=1.9614in,
width=5.1742in
]%
{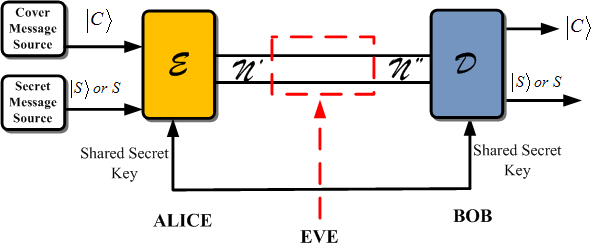}
\caption{There are three different inputs to the steganographic encoder $\CE$ : a cover-message $\ket{C}$; the secret message that we would like to hide, which can be quantum $\ket{S}$ or classical $S$; a shared secret key which may be quantum (ebit) $\ket{\CK}$ or classical $\CK$.  Eve can monitor some part of the noisy quantum channel $\CN$ shown in the red box.  Bob can decode the steganographic message using the decoder $\CD$ and the shared secret key $\ket{\CK}$ or $\CK$ and recover $\ket{C}$, and $\ket{S}$ or $S$ with very high probability.}
\label{fig:generalstegoprotocol}
\end{center}
\end{figure}
In this work we assume that Eve is a passive observer and all she can do is measure codewords.  The whole point of a steganographic protocol is not to apprise Eve that Alice and Bob are transmitting steganographic messages to each other.  The way to achieve this is for the message to look ``innocent'' to Eve.  Alice must be smart in applying the errors to the codewords and fool Eve into believing that the errors on the codewords are the result of an actual noise model on the quantum channel.  One of our major contributions to this theory is the notion of the security of a steganographic protocol which we elaborate in Chapter~\ref{chap:quantumsteg} by exploiting the properties of the diamond norm.  We emphasize that a steganographic protocol is considered broken if Eve detects the presence of any sort of covert activity on the channel.

Figure~\ref{fig:generalstegoprotocol} gives a schematic description of a general steganographic protocol between Alice and Bob, who share a secret key $\CK$ which may be classical or quantum (an ebit).  Table~\ref{tbl:protocol} shows the various scenarios of how one can hide a secret message in quantum codewords by using appropriate quantum error-correcting codes.
\begin{table}[tbp] \centering
\begin{tabular}{|c|c|c|}
\hline
\label{tbl:protocol}
Hidden Message & Key & Code \\ \hline \hline
Classical & Bit-string & CQECC \\ \hline
Classical & E-bit & EAQECC \\ \hline
Quantum & Bit-string & QECC \\ \hline
Quantum & E-bit & EAQECC \\ \hline
\hline
\end{tabular}
\caption{The first column represents message classical/quantum that we intend to hide.  The second column represents the shared secret key between Alice and Bob.  The third column represents the quantum error-correcting code.  Here CQECC stands for Classically-Enhanced Quantum Error-Correcting Code; QECC stands for Quantum Error-Correcting Code and EAQECC stands for Entanglement-Assisted Quantum Error-Correcting Code}
\end{table}

We organize the current chapter by first presenting a classical model of steganography by detailing two protocols that Alice uses to transmit steganographic information to Bob.  In Section~\ref{sec:exampleI} we give two examples of how Alice uses the syndromes of an error-correcting code to optimize the transmission of steganographic bits over a binary symmetric channel with bit-flip error-rate $p$.  In Section~\ref{sec:inner-outercode} we take a coding-theoretic approach to the problem.  Instead of encoding her steganographic data into the syndromes of an error-correcting code, Alice encodes her information into the codewords themselves using two different error-correcting codes which we refer to as the ``outer'' and ``inner'' codes.  The former is used as a cover-data, while we use the latter to encode steganographic bits.  We give numerical results on the optimal number of bits that Alice can send to Bob with the syndrome encoding using the three-bit and five-bit repetition codes over a noiseless binary-symmetric channel.  We also give numerical results when the BSC has noise for the three-bit repetition code.  We were unable to extend this to the five-bit code because the analysis became complicated and there were too many variables to optimize over.  We give numerical evidence of the optimal number of steganographic bits that Alice sends to Bob using the inner-outer coding technique.  We end the chapter with concluding remarks before proceeding to detail quantum steganography.

\section{Syndrome Encoding}
\label{sec:exampleI}
\begin{saying}
How do I work? I grope.\\
---\textit{Albert Einstein}
\end{saying}
We begin with a very simple example of a classical steganographic protocol where Alice would like to hide a single bit of information, called a steganographic bit, in the syndromes of a classical error-correcting code.  Alice and Bob a priori decide on a particular code.  After encoding her information by disguising them as errors, she uses a noiseless binary symmetric-channel (BSC) to send the information to Bob.  For this example let us assume that they are using the [3,1,3] repetition code that encodes three physical bits into one logical bit, and can correct up to a single bit-flip error.  Here $n = 3$, $k = 1$, and the distance $d = 3$.  This code can correct a single bit-flip error.  The codewords for this code are straightforward.
\begin{eqnarray}
\label{eqn1}
\overline{0} = 000~, \\
\overline{1} = 111~.
\end{eqnarray}
The number of syndromes for this code are $2^{n-k} = 2^2 = 4$.  We label the syndromes $s_0,\ldots,s_3$, where the first syndrome $s_0$ always corresponds to no error on the codeword. The rest of the syndromes $s_1,s_2$, and $s_3$ correspond to single bit-flips on each of the three physical bits of the codeword.  We need a natural error model for the bit-flip channel.  We assume that single bit-flips occur with probability $p$ independently of each other.  The probability of no bit-flip is $(1-p)$, while as the probability of a bit-flip error is $p$.  There is no intrinsic noise in the channel.  The point of this protocol is for Alice to emulate a realistic noise model for the bit-flip channel and fool Eve into thinking the noise is actually coming from the environment or is due to the physics of the channel.  Later we will analyze the case where there is also intrinsic noise in the channel.  We give the channel probability distribution in Table~\ref{table1}.  Notice that for each of the transformed codewords we give a sum of the probabilities of correctable and uncorrectable errors.  For example if we observe the first row of Table~\ref{table1} we see that we can obtain the codeword in two different ways - when there is no error on 000, which gives the probability $(1-p)^3$, and when the codeword 111 is flipped to 000 by three independent errors on each of the three physical bits.  The probability of the latter occurring is $p^3$.  Similar reasoning gives us the total probabilities for the rest of the rows in Table~\ref{table1}.
\begin{table}
\begin{tabular}
[c]{|c|c|c|c|c|}
\hline
Error & Corrupted Codeword & Syndrome & Probability of Error & Probability\\ \hline \hline
No Error & 000 & $s_0$ & $(1-p)^3 + p^3$ & $p_0$\\ \hline
Error on bit 1 & 001 & $s_1$ & $p(1-p)^2 + p^2(1-p)$ & $p_1$\\\hline
Error on bit 2 & 010 & $s_2$ & $p(1-p)^2 + p^2(1-p)$ & $p_2$\\\hline
Error on bit 3 & 100 & $s_3$ & $p(1-p)^2 + p^2(1-p)$ & $p_3$\\
\hline
\end{tabular}
\caption{Channel model for the [3,1,3] repetition code}
\label{table1}
\end{table}
Alice and Bob have access to a shared binary string (the secret key) with a uniform probability distribution.  They also have a function that takes as input shared secret bits from this key, and outputs an ordered pair $(j,k)$ where $j,k \in \{0,1,2,3\}$ with a non-uniform probability distribution, $q_{jk}$.  The ordered pair $(j,k)$ provides both parties with the following map:
\begin{eqnarray}
\label{eqn2}
0 \rightarrow s_j~, \\
1 \rightarrow s_k~.
\end{eqnarray}
If Alice wants to transmit a steganographic bit (say the bit 0), she applies the single-bit error corresponding to syndrome $s_j$ with probability $q_{jk}$, or $s_k$ with probability $q_{kj}$.  When Bob gets the codeword, he knows how to decode the syndromes because he has access to the same mapping~(\ref{eqn2}) that Alice does.

As an example, imagine that Alice and Bob have access to a shared secret key and a function which they use to generate, say, the ordered pair $(1,0)$ with probability $q_{10}$.  If Alice wants to hide the bit 0, she applies a single bit-flip error to the least-significant-bit (LSB) of her codeword, (or two bit-flips to the second and third bits of her codeword) and transmits it to Bob.  This corresponds to syndrome $s_1$.  If she wants to send the bit 1, then she does nothing to her codeword, which corresponds to the syndrome $s_0$.  (Of course, Alice could have picked the reverse map as well with probability $q_{01}$, where she would have hidden bit 0 in syndrome $s_0$ and bit 1 in syndrome $s_1$.)
The point of steganography is not to apprise Eve that Alice is sending steganographic bits to Bob over codewords.  Imagine the scenario where Alice sends several single-bit steganographic messages to Bob.  It would be a disaster if she kept using the same syndromes for her steganographic bits.  After a few rounds Eve would immediately guess that the channel was being used to send steganographic messages, for each time Eve would just see either no error or error on the first bit.  This would be a very unnatural error model for the bit-flip channel.  So Alice has to apply errors in such as a way as to emulate a natural error model.  In applying errors she must match the probability distribution detailed in Table~\ref{table1}.  This amounts to finding a solution to the following system of linear equations:
\begin{equation}
\label{eqn4}
p_j = \frac{1}{2}\sum_{k = 0}^3 \left(q_{jk} + q_{kj}\right),
\end{equation}
where $j,k \in \{0,1,2,3\}$.  If the probability $p$ of an error is very low, then Alice will not be able to hide information without apprising Eve because the Shannon entropy of the channel will be very low.  For the probability distribution given in Table~\ref{table1}, the maximum entropy $H(p) = 2$ is achieved when $p = 1/2$.  Therefore, this is not a good code to send steganographic information to Bob, because most of the time Eve will be expecting no errors on the codewords.  So Alice can only occasionally transmit a steganographic bit in a codeword in order not to alert Eve to any covert activity on the channel.

Alice and Bob get around this problem by including the possibility that they will send no bit in a given code-block.  This is represented by the diagonal elements $q_{jj}$.  In this case Alice just applies the syndrome $s_j$, and sends the codeword (with no steganographic bit) to Bob.  Since Bob shares the secret key, he knows when to expect a steganographic bit.

How should Alice choose $q_{jk}$ and $q_{kj}$ in order to emulate the probability distribution on the channel?  We proceed with the solution to Equation~\ref{eqn4} as follows.  We first set all the diagonal terms except $q_{00}$ to zero: i.e., $q_{11} = q_{22} = q_{33} = 0$.  We do this to maximize the probability of sending a steganographic bit.  Now we expand Equation~(\ref{eqn4}) for $j = 1$, to get the following:
\begin{eqnarray}
\label{eqn5}
&& p_1 = \frac{1}{2}\bigg(q_{10}+q_{11}+q_{12}+q_{13}+q_{01}+q_{11}+q_{21}+q_{31}\bigg)~, \nonumber \\
&& p(1-p)^2 + p^2(1-p) = \frac{1}{2}\bigg(q_{10}+q_{12}+q_{13}+q_{01}+q_{21}+q_{31}\bigg)~.
\end{eqnarray}
A solution to the above equation can be:
\begin{equation}
\label{eqn6}
q_{10} = q_{12} = q_{13} = q_{01} = q_{21} = q_{31} = \frac{p(1-p)^2 + p^2(1-p)}{3}~.
\end{equation}
The above solution gives Alice a better chance of sending a steganographic bit without apprising Eve, because Alice is spreading the likelihood of transmitting a steganographic bit among different syndromes.  So there is a good chance of applying a different error each time that Alice sends a steganographic bit to Bob, and moreover, the probability of applying this error matches the expected distribution on the channel.  The solution for $p_2$ and $p_3$ can be worked out similarly:
\begin{eqnarray}
\label{eqn7}
q_{20} = q_{21} = q_{23} = q_{02} = q_{12} = q_{32} = \frac{p(1-p)^2 + p^2(1-p)}{3}~, \\
q_{30} = q_{31} = q_{32} = q_{03} = q_{13} = q_{23} = \frac{p(1-p)^2 + p^2(1-p)}{3}~.
\end{eqnarray}
Now we go back and find the value for $q_{00}$ by solving Equation~(\ref{eqn4}).  For $j = 0$ this corresponds to:
\begin{equation}
\label{eqn7a}
(1-p)^3 + p^3 = \frac{1}{2}\bigg(q_{00}+q_{01}+q_{02}+q_{03}+q_{00}+q_{10}+q_{20}+q_{30}\bigg).
\end{equation}
We already have the values for each of the terms in Equation~(\ref{eqn7a}) except $q_{00}$.  We, therefore, can solve for $q_{00}$:
\begin{equation}
\label{eqn8}
\fbox{$q_{00} = 4p^2 - 4p + 1$}
\end{equation}
We would like the probability of sending no steganographic bits to be as low as possible and so we would like $q_{00}$ to be as low as possible.  The average steganographic information $N_{avg}$ that Alice transmits to Bob is the sum of the non-diagonal q-terms:

\begin{eqnarray}
\label{eqn8a}
N_{avg} & = & \sum_{j\neq k} q_{jk}~, \\
 & = & q_{01}+q_{02}+q_{03}+q_{10}+q_{12}+q_{13}+q_{20}+q_{21}+q_{23}+q_{30}+q_{31}+q_{32}~, \nonumber\\
 & = & 12\bigg(\frac{p(1-p)^2 + p^2(1-p)}{3}\bigg)~, \nonumber \\
 & = & 4(p(1-p)^2 + p^2(1-p))~, \nonumber \\
 & = & 4p(1-p)~. \nonumber
\end{eqnarray}
We now proceed to give a more efficient solution to Equation~(\ref{eqn4}).  As in the previous solution we would like to make $q_{00}$ as small as possible.  So let $q_{01},q_{02},q_{03},q_{10},q_{20}$, and $q_{30}$ be non-zero and let all other q-terms be equal to zero.  In order to solve for $q_{00}$, we must solve:
\begin{equation}
\label{eqn9}
p_0 = \frac{1}{2}\bigg(q_{00}+q_{00}+\sum_{j=1}^3 \bigg(q_{0j}+q_{j0}\bigg)\bigg)~.
\end{equation}
Equation~(\ref{eqn9}) can be simplified as follows:
\begin{equation}
\label{eqn10}
p_0 = q_{00}+ \sum_{j=1}^3 \frac{1}{2}\bigg(q_{0j}+q_{j0}\bigg)~.
\end{equation}
We need the quantity within the sum in Equation~(\ref{eqn10}) in order to find $q_{00}$.  We set all the non-zero q-terms $q_{01},q_{02},q_{03},q_{10},q_{20}$, and $q_{30}$ equal to the probability of a single bit-flip error, $p^2(1-p) + (1-p)^2p$, as follows,
\begin{eqnarray}
\label{eqn11a}
p_j & = & \frac{1}{2}\bigg(q_{0j}+q_{j0}\bigg) = p^2(1-p)+(1-p)^2p~, \\
\label{eqn11b}
q_{0j} & = & q_{j0} =  p^2(1-p)+(1-p)^2p~,
\end{eqnarray}
for $j = 1,2,3$.  Substituting Equation~(\ref{eqn11b}) into Equation~(\ref{eqn10}) we obtain:
\begin{equation}
\label{eqn12}
\fbox{$q_{00} = 6p^2 - 6p + 1$}
\end{equation}
Since the Shannon entropy for this channel is very low, $p \ll 1$, most of the time the codewords will go uncorrupted through the channel.  We know that Eve is expecting to see this behavior on the channel.  Any large deviation from this expectation will apprise Eve.  So Alice takes advantage of this knowledge and uses the syndrome $s_0$ as often as she can in order to hide single bits.

We can calculate the average amount of information that Alice transmits to Bob:
\begin{eqnarray}
\label{eqn12a}
N_{avg} & = &\sum_{j \neq k} q_{jk}~, \nonumber \\
 & = & q_{01}+q_{02}+q_{03}+q_{10}+q_{12}+q_{13}+q_{20}+q_{21}+q_{23}+q_{30}+q_{31}+q_{32}~, \nonumber \\
 & = & 6(p^2(1-p)+(1-p)^2p)~, \nonumber \\
 & = & 6p(1-p)~.
\end{eqnarray}
Comparing with the first example, we can see that Alice is able to send more steganographic information in the latter example because she is utilizing the syndrome $s_0$ which corresponds to no error on the codeword.  This is evident from Figure~\ref{fig:threebitplot}.  When she hides one bit in this syndrome, she's able to do more because most of the time Eve expects this channel to not corrupt the codewords and so Eve will be expecting codewords $000$ or $111$.  Alice knows this and thus takes advantage by packing in more steganographic information.
\begin{figure}[htp]
  \begin{center}
    	 \includegraphics[width = 4.5in]{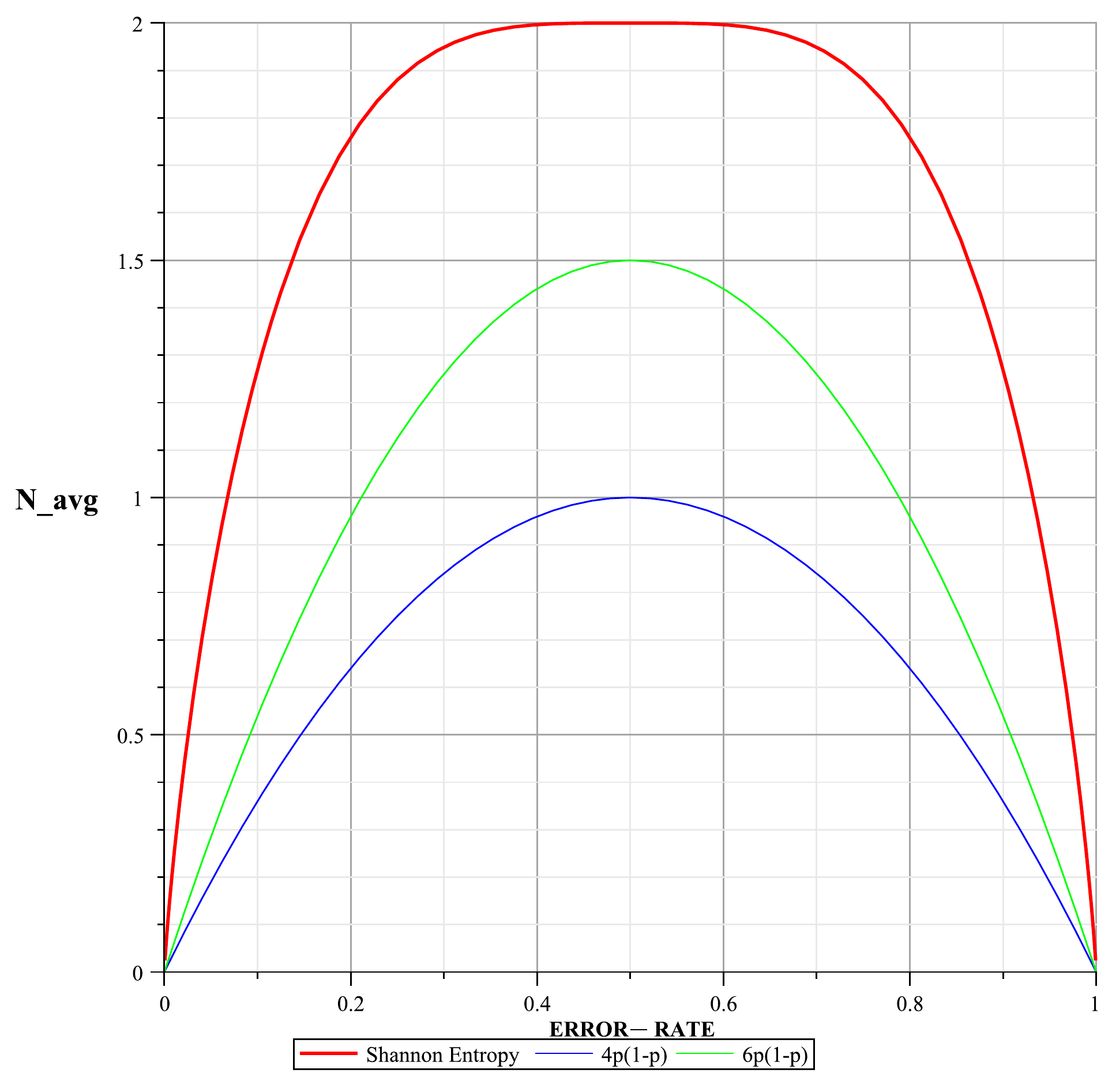}
  \end{center}
  \caption{Average Number of Steganographic Bits vs. Error-Rate $p$}
  \label{fig:threebitplot}
\end{figure}
\subsection{Noisy Bit-Flip Channel}
\label{exampleII}
We now investigate the case where Alice would like to transmit a steganographic bit to Bob in the presence of extra bit-flip noise on the channel.  In this example, as opposed to the previous noiseless case, when Alice sends her codewords through the channel, the channel may flip a bit and change the syndrome that corresponds to a steganographic bit.  As in the noiseless example in Section~\ref{sec:exampleI}, Alice and Bob have access to a shared binary string and a generating function which inputs some of these key bits and outputs an ordered pair $(j,k)$ with a non-uniform probability distribution $q_{jk}$.  Alice and Bob both have access to the following map:
\begin{eqnarray}
\label{eqn:exII-1}
0 \rightarrow s_j~, \\
1 \rightarrow s_k~.
\end{eqnarray}
If Alice wants to hide the bit 0, she applies the error corresponding to syndrome $s_j$ on her codeword with probability $q_{jk}$.  Unlike the example in Section~\ref{sec:exampleI}, now when the codeword passes through the channel there is a non-zero probability that the channel will apply its own error to the corrupted codeword that Alice is transmitting.  This gives rise to the following map:
\begin{eqnarray}
\label{eqn:exII-2}
s_j \rightarrow s_{l}~, \\
s_k \rightarrow s_{l'}~.
\end{eqnarray}
There is a non-zero conditional probability $p(l|j)$ and $p(l|k)$ with which the channel transforms Alice's intended syndrome $s_j$ (if she's hiding bit 0) or $s_k$ (if she's hiding bit 1) to $s_l$, respectively.  Alice would like to simulate this binary symmetric channel which now has additional noise from the environment.  In order to achieve this, we would have to solve the following set of equations:
\begin{equation}
\label{eqn:exII-3}
p_{l}' = \frac{1}{2}\sum_{j,k = 0}^3 q_{jk}\bigg(p(l|j) + p(l|k)\bigg)~,
\end{equation}
where $l = 0,1,\ldots,3$.  $p' > p$ is the error rate of the effective channel.  The effective probability distribution, corresponding to no flip, and bit-flips on bits one, two, and three of the noisy channel may be written as:
\begin{eqnarray}
\label{eqn:exII-3a}
p_{0}' & = & p'^3 + (1-p')^3~, \\
p_{1}' & = & p_{2}' = p_{3}' = p'(1-p')^2 + p'^2(1-p') = p'(1-p')~.
\end{eqnarray}

As in the example in Section~\ref{sec:exampleI}, we would like to determine the probability $q_{00}$ which corresponds to Alice not sending any bit to Bob.  Our goal is to minimize this probability, because we want Alice to transmit steganographic bits to Bob.  We begin by rewriting Equation~(\ref{eqn:exII-3}) for $l = 0$.  Note that we do not include diagonal q-terms (except $q_{00}$) in the following equation, as that corresponds to sending no steganographic bit.
\begin{equation}
\label{eqn:exII-4}
p_{0}' = \frac{1}{2}\sum_{j,k = 0}^3 q_{jk}\bigg(p(0|j) + p(0|k)\bigg)~.
\end{equation}
After expanding Equation~(\ref{eqn:exII-4}), we obtain the following equation:
\begin{equation}
\label{eqn:exII-5}
p_{0}' = \frac{1}{2}\bigg(q_{00}\bigg(p(0|0) + p(0|0)\bigg) + \ldots + q_{32}\bigg(p(0|3) + p(0|2)\bigg)\bigg)~.
\end{equation}
In order to solve Equation~(\ref{eqn:exII-5}), we need to determine the conditional probabilities, $p(0|0), p(0|1), p(0|2)$, and $p(0|3)$ which we have calculated in Table~\ref{tbl:CondProbTable}.
\begin{table}
\begin{tabular}
[c]{|c|c|c|c|}
\hline
Conditional Probability & Corrupted Codeword & Probability of Error  \\ \hline \hline
$p(0|0)$ & $000\stackrel{III}{\rightarrow}000 + 111\stackrel{XXX}{\rightarrow}000$ & $(1-p)^3 + p^3 = p_0$\\ \hline
$p(0|1)$ & $001\stackrel{IIX}{\rightarrow}000 + 110\stackrel{XXI}{\rightarrow}000$ & $p(1-p)^2 + p^2(1-p) = p_1$\\\hline
$p(0|2)$ & $010\stackrel{IXI}{\rightarrow}000 + 101\stackrel{XIX}{\rightarrow}000$ & $p(1-p)^2 + p^2(1-p) = p_1$\\\hline
$p(0|3)$ & $100\stackrel{XII}{\rightarrow}000 + 011\stackrel{IXX}{\rightarrow}000$ & $p(1-p)^2 + p^2(1-p) = p_1$\\
\hline
\end{tabular}
\caption{Conditional Probabilities - [3,1,3] Repetition Code}
\label{tbl:CondProbTable}
\end{table}
After substituting the value of the conditional probabilities from Table~\ref{tbl:CondProbTable} into Equation~(\ref{eqn:exII-5}) and rearranging the terms, we get the following:
\begin{equation}
\label{eqn:exII-6}
q_{00} = \frac{p_{0}'}{p_{0}} - \frac{1}{2}\frac{(p_0+p_1)}{p_0}\bigg(q_{01}+q_{02}+q_{03}+q_{10}+q_{20}+q_{30}\bigg)
- \frac{p_1}{p_0}\bigg(q_{12}+q_{13}+q_{21}+q_{23}+q_{31}+q_{32}\bigg)~.
\end{equation}
In order to solve Equation~(\ref{eqn:exII-6}), we need to determine the values for the various q-terms.  We begin by expanding the following equation:
\begin{eqnarray}
\label{eqn:exII-7}
p_{1}' & = & \frac{1}{2}\sum_{j = 1}^3 \sum_{k = 0}^3 q_{jk}\bigg(p(l|j) + p(l|k)\bigg)~, \nonumber \\
p_{1}' & = & \frac{1}{2}\bigg(q_{10}\bigg(p(1|1)+p(1|0)\bigg) + \ldots + q_{32}\bigg(p(1|3)+p(1|2)\bigg)\bigg)~.
\end{eqnarray}
We determine the conditional probabilities appearing in Equation~(\ref{eqn:exII-7}) in the same way as in Table~\ref{tbl:CondProbTable}.  After calculating and substituting the values for the conditional probabilities into Equation~(\ref{eqn:exII-7}) we get the following:
\begin{equation}
\label{eqn:exII-8}
p_{1}' = \frac{1}{2}(p_0+p_1)(q_{10}+q_{12}+q_{13}+q_{21}+q_{31})+p_1(q_{20}+q_{23}+q_{30}+q_{32})~.
\end{equation}
From Equation~(\ref{eqn:exII-8}), we can easily set the q-terms to the following:
\begin{eqnarray}
\label{eqn:exII-9}
&& q_{10}=q_{12}=q_{13}=q_{21}=q_{31} = \frac{1}{5}\frac{p_{1}'}{(p_0 + p_1)}~, \\
&& q_{20} = q_{23} = q_{30} = q_{32} = \frac{1}{8}\frac{p_{1}'}{p_{1}}~.
\end{eqnarray}
We still need to determine the values for $q_{01},q_{02}$, and $q_{03}$.  We begin by writing the following equation:
\begin{equation}
\label{eqn:exII-10}
p_{2}' = \frac{1}{2}\sum_{j = 0}^3\sum_{k = 1}^3 q_{jk}\bigg(p(2|j) + p(2|k)\bigg)~.
\end{equation}
After expanding Equation~(\ref{eqn:exII-10}), and rewriting, we get the following equation:
\begin{equation}
\label{eqn:exII-12}
p_{1}(q_{01}+q_{03}) + \frac{1}{2}\bigg(p_{0}+p_{1}\bigg)q_{02} = p_{2}'-\frac{1}{2}(p_{0}+p_{1})\bigg(q_{12}+q_{21}+q_{23}+q_{32}\bigg)-p_{1}(q_{13}+q_{31})~.
\end{equation}
But we already know the values of all the q-terms that appear on the right side of Equation~(\ref{eqn:exII-12}).  We substitute the values of these q-terms from Equation~(\ref{eqn:exII-9}):
\begin{equation}
\label{eqn:exII-13}
p_{1}q_{01} + \frac{1}{2}\bigg(p_{0}+p_{1}\bigg)q_{02}+p_{1}q_{03} = p_{2}'-\frac{1}{2}\bigg(p_{0}+p_{1}\bigg)\bigg(\frac{2}{5}\frac{p_{1}'}{(p_{0}+p_{1})}+\frac{1}{4}\frac{p_{1}'}{p_{1}}\bigg)-p_{1}\bigg(\frac{2}{5}\frac{p_{1}'}{(p_{0}+p_{1}}\bigg)~.
\end{equation}
Let us rewrite Equation~(\ref{eqn:exII-13}) as follows:
\begin{equation}
\label{eqn:exII-14}
p_{2}' = p_{1}q_{01} + \frac{1}{2}\bigg(p_{0}+p_{1}\bigg)q_{02}+p_{1}q_{03} + \gamma ~,
\end{equation}
where
\begin{equation}
\label{eqn:gamma}
\gamma = \frac{1}{2}\bigg(p_{0}+p_{1}\bigg)\bigg(\frac{2}{5}\frac{p_{1}'}{(p_{0}+p_{1})}+\frac{1}{4}\frac{p_{1}'}{p_{1}}\bigg)+p_{1}\bigg(\frac{2}{5}\frac{p_{1}'}{(p_{0}+p_{1}}\bigg)~.
\end{equation}
Now it is quite easy to set the values for $q_{01},q_{02},$ and $q_{03}$.
\begin{eqnarray}
\label{eqn:exII-15}
&& q_{01} = q_{03} = \frac{1}{3p_{1}}\bigg(p_{2}'-\gamma \bigg) = \frac{1}{3p_{1}}\bigg(p_{1}'-\gamma \bigg)~,\\
&& q_{02} = \frac{2}{3}\frac{p_{2}'-\gamma}{p_{0}+p_{1}} = \frac{2}{3}\frac{p_{1}'-\gamma}{p_{0}+p_{1}}~.
\end{eqnarray}
We can now finally substitute the values of the q-terms into Equation~(\ref{eqn:exII-6}):
\begin{equation}
\label{eqn:exII-16}
\fbox{$q_{00} = \frac{p_{0}'}{p_{0}} - \frac{1}{5}(\alpha+4\beta)\bigg(\frac{p_{1}'}{p_{0}+p_{1}}\bigg) - \frac{1}{4}(\alpha+\beta)\frac{p_{1}'}{p_{1}} - \frac{2}{3}\alpha\bigg(\frac{p_{1}'-\gamma}{p_{1}}\bigg)\bigg(\frac{p_{0}}{p_{0}+p_{1}}\bigg)$}
\end{equation}
where
\begin{eqnarray}
\label{eqn:exII-17}
&& \alpha = \frac{1}{2}\bigg(\frac{p_{0}+p_{1}}{p_{0}}\bigg)~, \\
&& \beta = \frac{p_{1}}{p_{0}}~,
\end{eqnarray}
and where we obtain $p_{0}'$ and $p_{1}'$ in terms of $p$ and $\delta p$ by expanding Equation~(\ref{eqn:exII-3a}):
\begin{eqnarray}
\label{eqn:exII-18}
&& p_{0}' = 1-3p(1-p)-3\delta p(1-2p)~, \\
&& p_{1}' = p_{2}' = p_{3}' = p(1-p-\delta p)+\delta{p}(1-p)~.
\end{eqnarray}
Next we would like to determine the amount of average information that Alice can transmit to Bob if the effective error rate $p' = p+ \delta p > p$.  The sum of the non-diagonal q-terms will give us the average information:
\begin{equation}
\label{eqn:exII-19}
N_{avg} = \sum_{j \neq k}^3 q_{jk}~.
\end{equation}
After expanding Equation~(\ref{eqn:exII-19}), we get:
\begin{equation}
\label{eqn:exII-20}
\fbox{$N_{avg} = \frac{2}{3}(p_{1}'-\gamma)\bigg(\frac{p_{0}+2p_{1}}{p_{1}(p_{0}+p_{1})}\bigg) + p_{1}'\bigg(\frac{p_{0}+3p_{1}}{2p_{1}(p_{0}+p_{1})}\bigg)$}
\end{equation}

\subsection{Numerical Simulation}
\label{numericsimI}
We would like to optimize the average number of bits that Alice can hide in a three-bit repetition code for both the noiseless and the noisy bit-flip channel.  In order to do this we set up the problem as a linear-optimization one, where the objective function (also known as the cost or profit function) that we would like to optimize is the average number of steganographic bits that Alice can transmit to Bob.
We would further like to extend this example to the five-bit repetition code that encodes one logical bit into five physical bits, and can correct up to two arbitrary bit-flips.  While we can implement the three-bit repetition code from Section~\ref{sec:exampleI}, it turns out that if we extend this naively to the five-bit case, we run into the problem of optimizing over several thousand variables.  This happens because unlike the three-bit case, the five-bit code has sixteen syndromes.  This gives Alice and Bob many ways to encode their steganographic bits.
With the three-bit code Alice and Bob can in addition to hiding a single bit, hide two bits.  They would use the following encoding map:
\begin{eqnarray}
\label{eqn:numericsimI-1}
00 \rightarrow s_{0}~, \\
01 \rightarrow s_{1}~, \\
10 \rightarrow s_{2}~, \\
11 \rightarrow s_{3}~.
\end{eqnarray}
However, if the error-rate of the channel is very low, then Alice will rarely transmit two steganographic bits.  On the average Alice will transmit one steganographic bit to Bob.  If Eve knows that the error-rate of the channel is low, where a single bit is flipped occasionally, she would be apprised of covert activity on the channel if Alice applies the two-bit encoding most of the time.  The ideal situation for Alice would be where she minimizes sending no bit at all through the channel, but where she can occasionally (rarely) hide two bits, thereby boosting her transmission rate to little over one bit over several uses of the channel.  The latter depends on the value of $p$.  We summarize the linear equality/inequality constraints along with the objective function $N_{avg}$ (average number of hidden bits) that we optimize, below:
\begin{eqnarray}
\label{eqn:numericsumI-2}
p_j & = &\frac{1}{2}\sum_{k = 0}^3 \left(q_{jk} + q_{kj}\right), \\
\sum_{j,k} q_{jk} & = & 1, \\
 q_{jk} & \geq & 0, \forall j,k, \\
 N_{avg} & = & \sum_{j \neq k} q_{jk} + 2\sum_{j \neq k \neq l \neq m}q_{jklm}.
\end{eqnarray}
\begin{figure}[htp]
  \begin{center}
    	 \includegraphics[height = 4.0in,width = 4.0in]{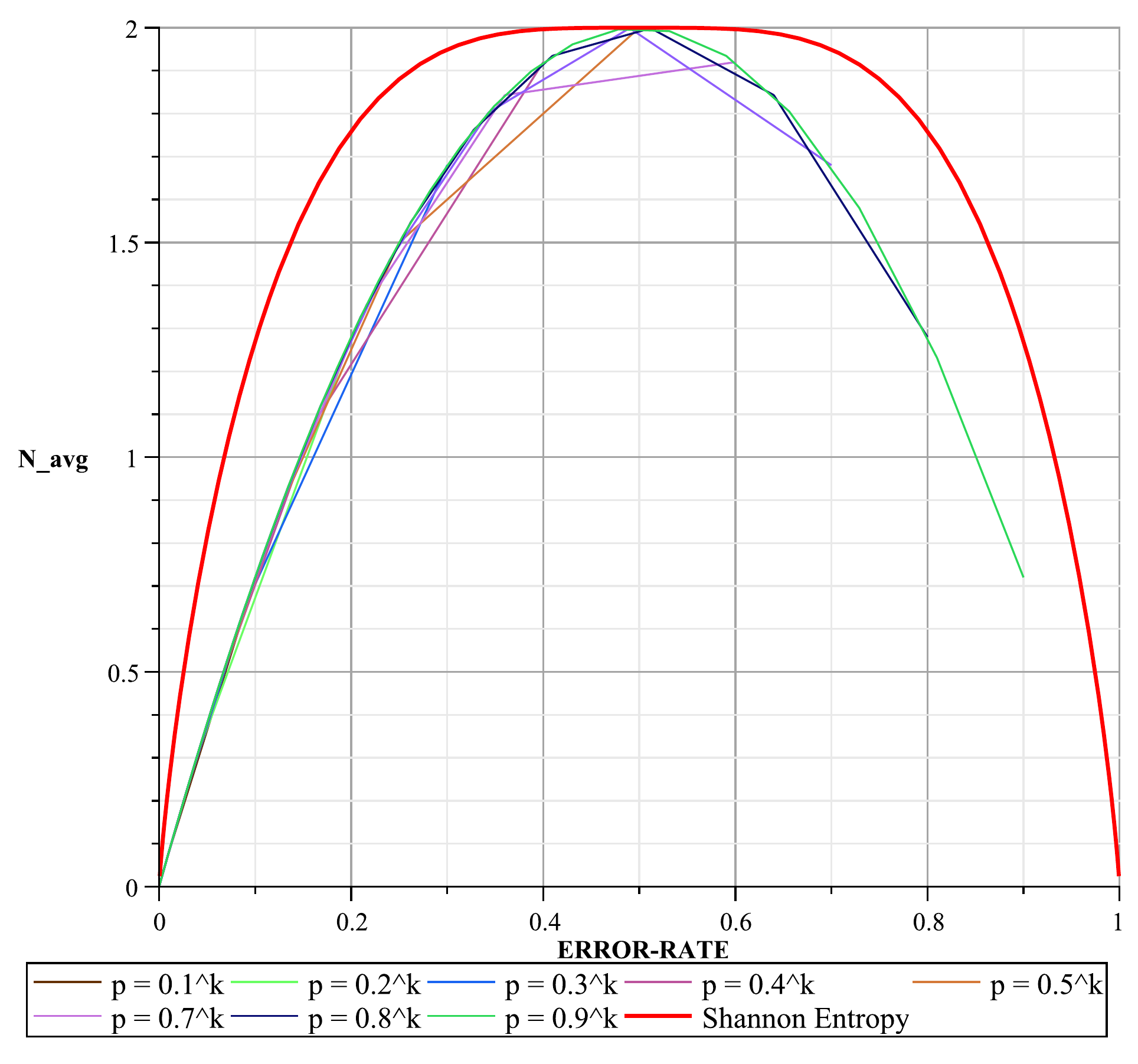}
  \end{center}
  \caption{Noiseless Three-Bit Code. We plot the error-rate on the X-axis and the average number of bits, $N_{avg}$, on the Y-axis for various error sequences.  Here $k$ varies from 1 to 100.  As the error-rate increases Alice can hide up to two bits.  The red curve plots the Shannon entropy of the probability distribution of the syndromes.}
  \label{fig:NoiselessThreeBits}
\end{figure}
So we end up with five linear equality constraints and sixteen linear inequality constraints.  We can already see that with the three-bit code we are optimizing over sixteen different variables, with several linear constraints.  When we try to extend this naive approach to the five-bit problem, we end up optimizing over several thousand variables as we mentioned earlier.

In the previous approach Alice and Bob use the sixteen variables $q_{jk}$ to keep track of the various ways in which they can encode their steganographic bits.  We also know that Alice must apply errors in such a way that she matches the probability distribution $p$ of the bit-flip channel.  In the second approach Alice groups her syndromes by probabilities.  From Table~\ref{table1} it is clear that syndromes $s_1, s_2$, and $s_3$ have the same probability $p^2(1-p)+p(1-p)^2$, while syndrome $s_0$ has the probability $p^3 + (1-p)^3$.  We reduce the number of variables over which we must optimize by introducing the notion of a syndrome class.  Instead of the variables $q_{jk}$, we now use $\CQ_{j_{0},...,j_{M-1}}$, where $j_{0}$ is the number of syndromes Alice uses from the first group, and the number of syndromes she uses from the second group and so on, and where $M = 2$, tells us the total number of syndrome classes.  So the $\CQ$ variables give us the total probability of using an encoding scheme to hide a bit.  Moreover, we use the variables $n_{0}, n_{1}, ..., n_{M-1}$ to denote the various syndrome classes, and to keep track of the total number of syndromes in each group.  So for the three-bit case we have $n_{0} = 1$ for we only have the syndrome $s_{0}$, and $n_{1} = 3$ for we have three different syndromes $s_{1},s_{2}$, and $s_{3}$.  With this new definition the linear equality constraints that Alice must now satisfy are as follows:
\begin{equation}
\label{eqn:numericsimI-3}
p_{k} = \sum_{j_{0},j_{1},\ldots,j_{M-1}}\frac{j_{k}/n_{k}}{j_{0}+j_{1}+\ldots+j_{M-1}}\CQ_{j_{0}j_{1} \ldots j_{M-1}}~,
\end{equation}
and
\begin{equation}
\sum_{j_{0},j_{1},\ldots,j_{M-1}}\CQ_{j_{0},j_{1},\ldots,j_{M-1}} = 1~.
\end{equation}
While as the linear inequality constraints that she must satisfy are:
\begin{equation}
\label{eqn:numericsimI-4}
\CQ_{j_{0},j_{1},\ldots,j_{M-1}}  \geq  0~. \\
\end{equation}
The average number of bits that she can hide without apprising Eve of steganographic activity are:
\begin{equation}
\label{eqn:numericsimI-5}
N_{avg} = \sum_{j_{0},j_{1},\ldots,j_{M-1}}\left\lfloor \log_{2}(j_{0}+j_{1}+\ldots+j_{M-1})\right\rfloor \CQ_{j_{0}j_{1} \ldots j_{M-1}}~.
\end{equation}
With the above definition we can frame our noiseless three-bit linear optimization problem as follows:
\begin{eqnarray}
\label{eqn:numericsimI-6}
&& p_{0}  =  \CQ_{10} + \frac{1}{2}\CQ_{11} + \frac{1}{4}\CQ_{13}~, \\
&& p_{1}  =  \frac{1}{3}\CQ_{01}+\frac{1}{6}\CQ_{11}+\frac{1}{3}\CQ_{02}+\frac{1}{4}\CQ_{13}~, \\
&& \sum_{j_{0},j_{1}}\CQ_{j_{0}j_{1}}  =  1~,\\
&& \CQ_{j_{0}j_{1}}  \geq  0~,~\forall j_{0}, j_{1}~,\\
&& N  =  \CQ_{11} + \CQ_{02} + 2\CQ_{13}~.
\end{eqnarray}
We have reduced our problem from optimizing over sixteen variables to optimizing over five variables.  We summarize our result from this optimization in Figure~\ref{fig:NoiselessThreeBits}.  In this example Alice can hide a maximum of two bits at an error-rate of 0.5.  We plot the Shannon entropy of the various syndrome probabilities in red.  In the limit of large block codes, we can essentially close the gap between the Shannon curve and the curves for the optimal solution.

We now detail the five-bit code, [5,1,5] example.  This is a distance-five code and can correct up to two bit-flip errors.  The [5,1,5] code has sixteen distinct syndromes, $s_0, s_1,\ldots,s_{15}$.  As in the previous example, $s_0$ corresponds to no error.  Syndromes $s_1, s_2,\ldots,s_5$ correspond to single bit-flip errors on each of the five physical bits of the corrupted codeword.  Syndromes $s_6,s_7,\ldots,s_{15}$ correspond to double bit-flip errors. We summarize the error model for the bit-flip channel for the [5,1,5] code in Table~\ref{tbl:five-bit}.  Notice that as in Table~\ref{table1} we have included the probabilities for both correctable and uncorrectable errors in each of the rows of Table~\ref{tbl:five-bit}.
\begin{table}
\begin{center}
\begin{tabular}
[c]{|c|c|c|c|c|}
\hline
Error & Codeword & Syndrome & Prob of Error & Error\\ \hline \hline
No Error & 00000 & $s_0$ & $(1-p)^5 + p^5$ & $p_0$\\ \hline
Error on bit 1 & 00001 & $s_1$ & $(1-p)^{4}p + p^4(1-p)$ & $p_1$\\\hline
Error on bit 2 & 00010 & $s_2$ & $(1-p)^{4}p + p^4(1-p)$ & $p_2$\\\hline
Error on bit 3 & 00100 & $s_3$ & $(1-p)^{4}p + p^4(1-p)$ & $p_3$\\ \hline
Error on bit 4 & 01000 & $s_4$ & $(1-p)^{4}p + p^4(1-p)$ & $p_4$\\ \hline
Error on bit 5 & 10000 & $s_5$ & $(1-p)^{4}p + p^4(1-p)$ & $p_5$\\ \hline
Error on bits 1 and 2 & 00011 & $s_6$ & $(1-p)^{3}p^{2} + p^3(1-p)^{2}$ & $p_6$\\ \hline
Error on bits 1 and 3 & 00101 & $s_7$ & $(1-p)^{3}p^{2} + p^3(1-p)^{2}$ & $p_7$\\ \hline
Error on bits 1 and 4 & 01001 & $s_8$ & $(1-p)^{3}p^{2} + p^3(1-p)^{2}$ & $p_8$\\ \hline
Error on bits 1 and 5 & 10001 & $s_9$ & $(1-p)^{3}p^{2} + p^3(1-p)^{2}$ & $p_9$\\ \hline
Error on bits 2 and 3 & 00110 & $s_{10}$ & $(1-p)^{3}p^{2} + p^3(1-p)^{2}$ & $p_{10}$\\ \hline
Error on bits 2 and 4 & 01010 & $s_{11}$ & $(1-p)^{3}p^{2} + p^3(1-p)^{2}$ & $p_{11}$\\ \hline
Error on bits 2 and 5 & 10010 & $s_{12}$ & $(1-p)^{3}p^{2} + p^3(1-p)^{2}$ & $p_{12}$\\ \hline
Error on bits 3 and 4 & 01100 & $s_{13}$ & $(1-p)^{3}p^{2} + p^3(1-p)^{2}$ & $p_{13}$\\ \hline
Error on bits 3 and 5 & 10100 & $s_{14}$ & $(1-p)^{3}p^{2} + p^3(1-p)^{2}$ & $p_{14}$\\ \hline
Error on bits 4 and 5 & 11000 & $s_{15}$ & $(1-p)^{3}p^{2} + p^3(1-p)^{2}$ & $p_{15}$\\
\hline
\end{tabular}
\end{center}
\caption{Channel model for the [5,1,5] code}
\label{tbl:five-bit}
\end{table}
In the five-bit example we have three distinct syndrome classes, no error, a single bit-flip error, and two bit-flip errors.  We label these classes by $n_{0},n_{1}$, and $n_{2}$ respectively.  If we observe Table~\ref{tbl:five-bit}, we see that $n_{0} = 1, n_{1} = 5$, and $n_{2} = 10$.  The probability of using a particular encoding is represented by the variable $\CQ_{j_{0}j_{1}j_{2}}$.  With the five-bit code, Alice can hide up to four steganographic bits.  We summarize the various encoding schemes that Alice can employ to send zero, one, two, three, and four bits to Bob in Table~\ref{tbl:fivebitencscheme}.
The first linear equality that Alice must satisfy is that the total probability of using all the various encoding schemes must sum up to 1:
\begin{equation}
\label{eqn:fivebiteqlconstraint1}
\sum_{j_{0}j_{1}j_{2}}\CQ_{j_{0}j_{1}j_{2}} = 1~.
\end{equation}
The second set of linear equality constraints deals with satisfying the probability distribution of the channel, as follows:
\begin{eqnarray}
\label{eqn:fivebiteqlconstraint2}
&& p_{0} = \CQ_{100} + \frac{1}{2}\bigg(\CQ_{110}+\CQ_{101}\bigg) + \frac{1}{4}\bigg(\CQ_{130}+\CQ_{121}+\CQ_{103}+\CQ_{112}\bigg)\nonumber \\
&& + \frac{1}{8}\bigg(\CQ_{152} + \CQ_{143} + \CQ_{134} + \CQ_{125} + \CQ_{116} + \CQ_{107}\bigg) + \frac{1}{16}\CQ_{1,5,10}~,  \\
\nonumber \\
&& p_{1} = \frac{1}{5}\bigg(\CQ_{010}+\CQ_{020}+\CQ_{040}\bigg) + \frac{1}{10}\bigg(\CQ_{110}+\CQ_{011}+\CQ_{121}+\CQ_{022}+\CQ_{044}+\CQ_{143}\bigg) \nonumber \\
&& + \frac{3}{20}\bigg(\CQ_{130} + \CQ_{031}\bigg) + \frac{1}{8}\bigg(\CQ_{152}+\CQ_{053}\bigg) +\frac{3}{40}\bigg(\CQ_{035}+\CQ_{134}\bigg)\nonumber \\
&& +\frac{1}{20}\bigg(\CQ_{013}+\CQ_{026}+\CQ_{125}+\CQ_{112}\bigg) + \frac{1}{40}\bigg(\CQ_{017}+\CQ_{116}\bigg) + \frac{1}{16}\CQ_{1,5,10}~,\\
\nonumber \\
&& p_{2} = \frac{1}{10}\bigg(\CQ_{001}+\CQ_{004}+\CQ_{008}+\CQ_{002}\bigg) + \frac{3}{80}\bigg(\CQ_{053} + \CQ_{143}\bigg) \nonumber \\
&& + \frac{1}{20}\bigg(\CQ_{011}+\CQ_{101}+\CQ_{022}+\CQ_{044}+\CQ_{112}\bigg) + \frac{7}{80}\bigg(\CQ_{017}+\CQ_{107}\bigg) \nonumber \\
&& + \frac{1}{40}\bigg(\CQ_{121}+\CQ_{031}+\CQ_{152}\bigg) + \frac{3}{40}\bigg(\CQ_{013}+\CQ_{026}+\CQ_{116}+\CQ_{103}\bigg) \nonumber \\
&& + \frac{1}{16}\bigg(\CQ_{035}+\CQ_{125}+\CQ_{1,5,10}\bigg)~.
\end{eqnarray}
The linear inequality constraints that Alice must satisfy are:
\begin{table}
\begin{center}
\begin{tabular}
[c]{|c|c|c|c|c|c|c|c|c|c|}
\hline
$N_{s}$ & Encoding & $N_{s}$ & Encoding &$N_{s}$ & Encoding & $N_{s}$ & Encoding &$N_{s}$ & Encoding\\ \hline \hline
0 & $\CQ_{010}$ & 1 & $\CQ_{011}$ & 2 & $\CQ_{031}$ & 3 & $\CQ_{053}$ & 4 & $\CQ_{1,5,10}$ \\ \hline
0 & $\CQ_{001}$ & 1 & $\CQ_{020}$ & 2 & $\CQ_{022}$ & 3 & $\CQ_{044}$ && \\ \hline
0 & $\CQ_{100}$ & 1 & $\CQ_{002}$ & 2 & $\CQ_{013}$ & 3 & $\CQ_{035}$ && \\ \hline
&&                1 & $\CQ_{101}$ & 2 & $\CQ_{004}$ & 3 & $\CQ_{026}$ && \\ \hline
&&                1 & $\CQ_{110}$ & 2 & $\CQ_{040}$ & 3 & $\CQ_{017}$ && \\ \hline
&& &&                               2 & $\CQ_{121}$ & 3 & $\CQ_{008}$ && \\ \hline
&& &&                               2 & $\CQ_{130}$ & 3 & $\CQ_{143}$ && \\ \hline
&& &&                               2 & $\CQ_{103}$ & 3 & $\CQ_{143}$ && \\ \hline
&& &&                               2 & $\CQ_{112}$ & 3 & $\CQ_{125}$ && \\ \hline
&& && &&                                             3 & $\CQ_{116}$ && \\ \hline
&& && &&                                             3 & $\CQ_{107}$ && \\ \hline
&& && &&                                             3 & $\CQ_{152}$ && \\ \hline
\hline
\end{tabular}
\end{center}
\caption{$N_{s}$ is the number of steganographic bits that Alice can hide and subsequently transmit to Bob. Column 2 shows the
various encoding schemes that Alice can employ to successfully transmit steganographic bits to Bob.}
\label{tbl:fivebitencscheme}
\end{table}
\begin{equation}
\label{eqn:fivebitineqconstraint}
\CQ_{j_{0}j_{1}j_{2}} \geq 0, \forall j_{0},j_{1},j_{2}~.
\end{equation}
The objective function that Alice must maximize is the average number of steganographic bits that she can transmit to Bob:
\begin{eqnarray}
\label{eqn:fivebitavgbits}
&& N_{avg} = \sum_{j{0}+j_{1}+j_{2} = 2}\CQ_{j_{0}j_{1}j_{2}} + \sum_{j{0}+j_{1}+j_{2} = 4}2\CQ_{j_{0}j_{1}j_{2}} \nonumber \\
&& + \sum_{j{0}+j_{1}+j_{2} = 8}3\CQ_{j_{0}j_{1}j_{2}} + \sum_{j{0}+j_{1}+j_{2} = 16}4\CQ_{j_{0}j_{1}j_{2}}~.
\end{eqnarray}

\begin{figure}[htp]
  \begin{center}
    	 \includegraphics[height = 4.0in,width = 4.0in]{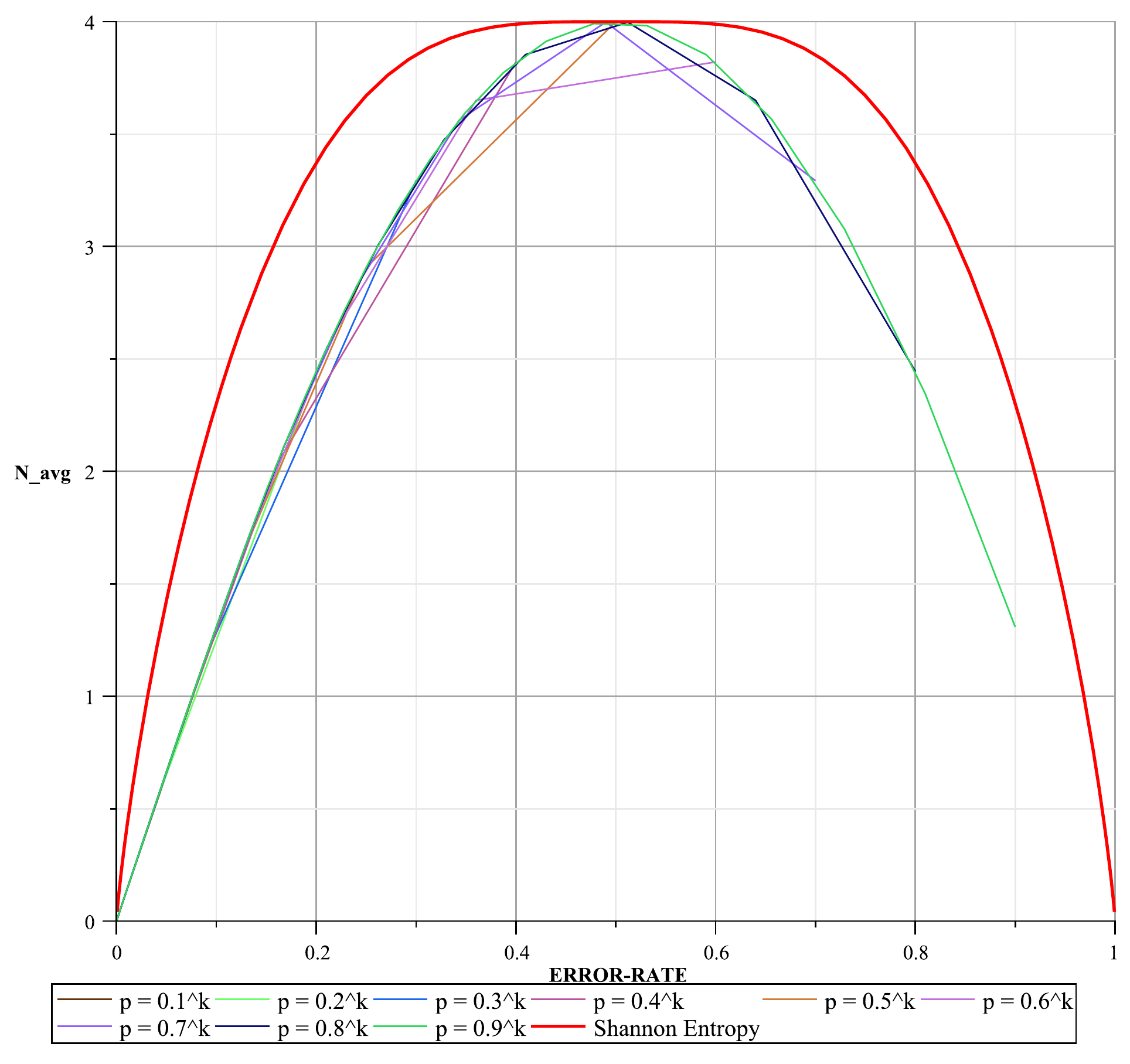}
  \end{center}
  \caption{Noiseless Five-Bit Code. We plot the error-rate on the X-axis and the average number of bits, $N_{avg}$, on the Y-axis for various error sequences.  Here $k$ varies from 1 to 100.  As the error-rate increases Alice can hide up to six bits.  The red curve plots the Shannon entropy of the probability distribution of the syndromes.}
  \label{fig:NoiselessFiveBits}
\end{figure}
The result of optimizing the five-bit code is shown in Figure~\ref{fig:NoiselessFiveBits}.  The figure is qualitatively similar to the noiseless three-bit example.  Once again the maximum number of steganographic bits that Alice can hide is at an error-rate of 0.5.  This happens because at that rate the Shannon entropy of the channel is maximized which means that almost all of the codewords are experiencing bit-flips, and so Alice can on the average can transmit more steganographic bits.  In the limit of large code-blocks, the number of steganographic bits that Alice can transmit to Bob reaches the Shannon capacity.  Even though we don't prove this statement, one can infer this by eye-balling the gap between the steganographic curves and the Shannon curve in Figures~\ref{fig:NoiselessThreeBits}, and~\ref{fig:NoiselessFiveBits}.
\begin{figure}[htp]
  \begin{center}
    	 \includegraphics[height = 4.0in,width = 4.0in]{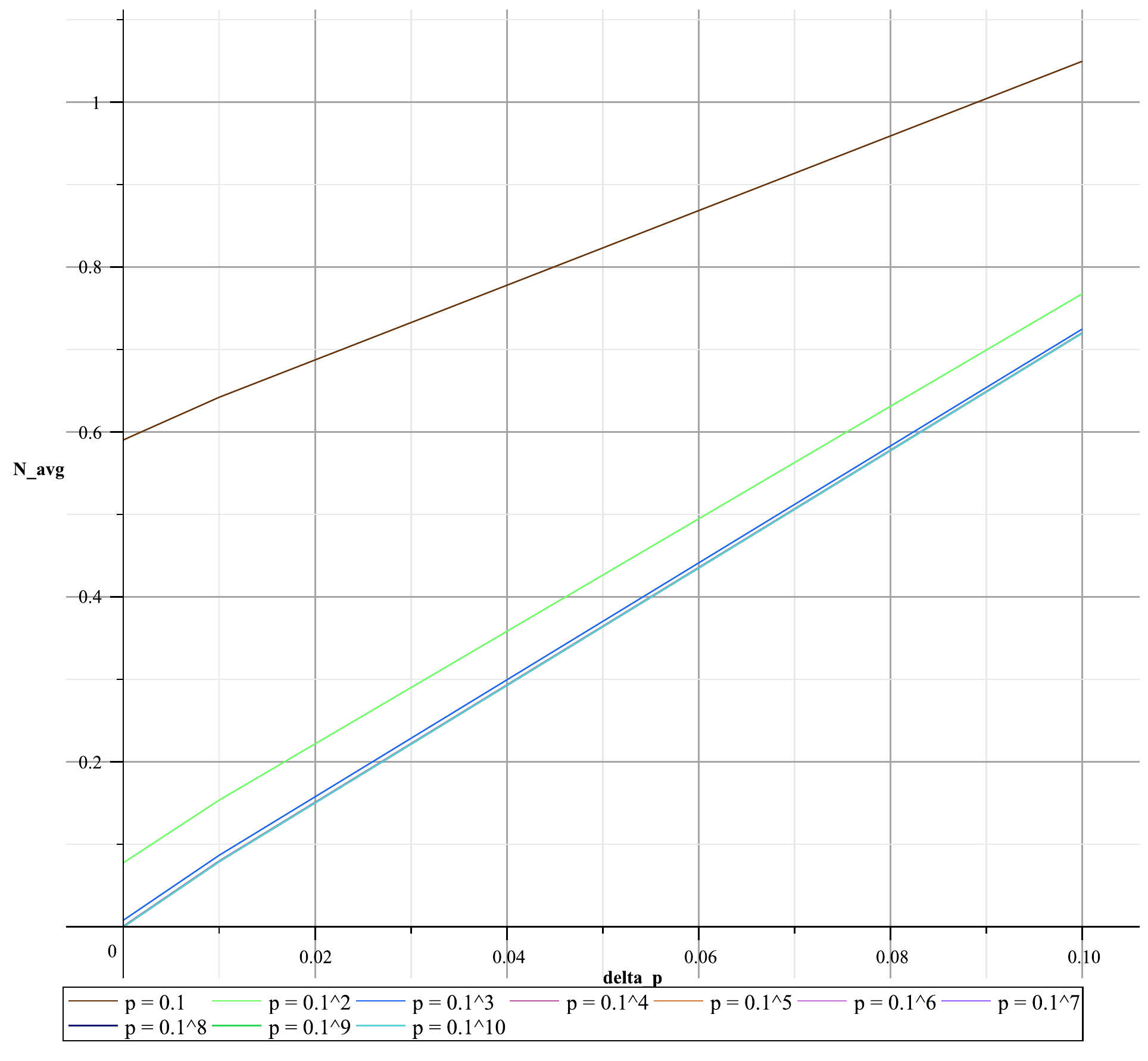}
  \end{center}
  \caption{Noisy Three-Bit Code. We plot the channel noise $\delta{p}$ on the X-axis and the average number of bits correctly received by Bob, $N_{avg}$, on the Y-axis for various error-rates $p$ which decrease by an order of magnitude.}
  \label{fig:NoisyThreeBits}
\end{figure}
We now move on to the case where Alice hides up to two steganographic bits in a three-bit code and transmits it through a noisy bit-flip channel with error-rate $p' = p + \delta p$, where $\delta p \geq 0$ is the noise that the channel experiences from the environment, and where $p$ is the error-rate that Eve expects of the channel.  The idea here is that Alice ought to be able to utilize the extra noise in the channel to send a steganographic bit to Bob more frequently.  We implement this linear optimization problem in terms of encoding classes following along the lines of noiseless three-bit example stated previously.  In Table~\ref{tbl:noisytable}, we show the number of correct steganographic bits that Bob receives from Alice, when she encodes two steganographic bits in the three-bit code.
\begin{table}
\begin{center}
\begin{tabular}
[c]{|c|c|c|}
\hline
Transmitted Bits & Received Bits &  $\#$ of Correct Bits \\ \hline \hline
&    00 & 2 \\
00 & 01 & 1 \\
&    10 & 1 \\
&    11 & 0 \\ \hline
&    00 & 1 \\
01 & 01 & 2 \\
&    10 & 0 \\
&    11 & 1 \\ \hline
&    00 & 1 \\
10 & 01 & 0 \\
&    10 & 2 \\
&    11 & 1 \\ \hline
&    00 & 0 \\
11 & 01 & 1 \\
&    10 & 1 \\
&    11 & 2 \\ \hline
\hline
\end{tabular}
\end{center}
\caption{The first column shows the two bits that Alice transmits to Bob, whereas the second column shows the received bits by Bob, some of which have experienced an error due to the noisy bit-flip channel.}
\label{tbl:noisytable}
\end{table}
Suppose Alice and Bob agree on a particular encoding, say, $(j,k)$ with probability $q_{jk}$, the encoding that we state in Equation~(\ref{eqn2}).  One can divide this scenario into two cases: the first case where Alice transmits the codeword corresponding to syndrome $j$, and Bob receives the codeword corresponding to syndrome $l$, and the second case where Alice sends the codeword corresponding to syndrome $k$, but Bob receives the codeword corresponding to syndrome $l$.  Based on these two notions we define the probability of receiving a correct bit by Bob as follows:
\begin{equation}
\label{eqn:numericsimI-8}
p_{corr}(j,k) = \frac{1}{2}\bigg(\sum_{p\left(l|j\right)>p\left(l|k\right)} p(l|j) + \sum_{p\left(l|j\right)<p\left(l|k\right)}p(l|k)\bigg)~.
\end{equation}
The linear equality and inequality constraints for the noisy example are exactly the same as those for the noiseless case in Equation~(\ref{eqn:numericsimI-6}).  The objective function $N_{avg}$ that Alice must optimize is:
\begin{equation}
\label{eqn:numericsimI-9}
N_{avg} = \sum_{j \neq k}p_{corr}(j,k)q_{jk}+\sum_{j \neq k \neq l\neq m}2\left(p_{0}+p_{1}\right)q_{jklm}~.
\end{equation}
The above objective function, in the encoding classes setting, gets transformed to:
\begin{equation}
\label{eqn:numerisimI-10}
N_{avg} = p_{0}\left(\CQ_{11}+\CQ_{02}\right) + 2(p_{0}+p_{1})\CQ_{13}~.
\end{equation}
We expect that with the extra noise Alice ought to be able to sneak in steganographic bits more frequently, and that is what we our results show in Figure~\ref{fig:NoisyThreeBits}

\section{Inner-Outer Codes}
\label{sec:inner-outercode}
\begin{saying}
Information is the resolution of uncertainty.
---\textit{Claude Shannon}
\end{saying}
In this section we generalize the notion of transmitting steganographic bits over bit-flip channels by utilizing coding theory to our advantage.  The strategy of straightforwardly encoding steganographic bits into the syndromes of an error-correcting code can only go so far.  Transmitting steganographic bits over a five-bit code by utilizing syndromes becomes infeasible because the noisy environment can transform Alice's intended syndrome to any of the other fifteen syndromes.  The main reason that the strategy worked for the three-bit case is because it is a very simple example, the conditional probabilities are easy to calculate, and there are not too many ways by which the noise can transform the syndromes of the three-bit code.  In this chapter we present a different protocol for transmitting steganographic bits over a bit-flip channel.  The protocol still uses error-correcting codes, but Alice and Bob no longer encode steganographic bits into syndromes.  One of the main advantages of developing this protocol is that it generalizes in a straightforward fashion to the quantum setting.

\subsection{Classical Steganographic Protocol}
\label{sec:classicalstegprotocol}
In any steganographic protocol the sender, Alice, embeds her steganographic information into an innocent-looking cover-text message, and transmits it to the receiver Bob over a binary symmetric channel with bit-flip rate $p$.  She could for example use a passage from Hamlet to encode steganographic information into Shakespeare's poetry via a secret key shared between Bob and her.  When this message is intercepted by Eve she does not suspect that there is any covert activity between them because, after all, the message is just a few passages from Hamlet.  Our biggest challenge was to extend this latter notion of an innoncent-looking message to the information-theoretic setting.  As we demonstrated in the previous chapter, we achieve this via a set linear equality constraints that Alice must satisfy so that she can match the probability distribution $p$ of the bit-flip channel.  The protocol begins as follows:
\begin{enumerate}
	\item Alice chooses an $[N,k_{c}]$ error-correcting code, where $N \gg 1$, as her cover-code.  We also refer to this cover-code as           an ``outer'' code.
	\item Alice chooses a random $M$-bit subset from among the $N$ bits of the outer code, where $M < N$, and $M \gg 1$.  The locations of these $M$ bits is specified by the secret key that Alice and Bob share with each other.
	\item Alice applies errors to the subset of $M$ bits using an $[M,k_{s}]$ code, which we refer to as the ``inner'' code, thereby encoding $k_{s}$ steganographic bits.  An average codeword of an $[M,k_{s}]$ code has weight approximately $\frac{M}{2}$.  By applying these errors on the $N$-bit code-block, the effective binary symmetric channel has a bit-flip rate of $q \approx \frac{M}{2N}$.
	\item For an extra layer of protection Alice applies a one-time pad to the $M$-bit subset.  In addition to the original key, Alice and Bob also share these extra $M$ bits of their one-time pad.
	\item Alice transmits her $N$-bit code-block through the channel.
	\item Bob measures the syndromes of the cover-code and extracts the pattern of bit-flip errors.
	\item Bob reapplies the one-time pad to the $M$-bit subset and recovers the errors on the $M$ bits.
	\item Bob decodes the $M$-bit codeword and correctly recovers $k_{s}$ steganographic bits.
\end{enumerate}

Let us clarify the protocol above with a simple example which is shown in Figure~\ref{fig:classicalstegoprotocol}.  Suppose Alice chooses the ten-bit repetition code, $[10,1,10]$, as her cover-code.  She encodes 1 into 1111111111 which is one of the codewords of the $[10,1,10]$ code.  So the $[10,1,10]$ is her ``outer'' code.  She then chooses three bits randomly from this ten-bit code-block.  In this example she chooses the third, sixth, and eighth bits to represent the codeword for her $[3,1,3]$ inner-code.  She then flips the rest of the bits to 0.  She ends up with 0010010100.  Suppose the one-time pad is 0010000100. She then applies this one-time pad to 0010010100, and gets 0000010000.  She finally encodes this into the cover-code by XORing the latter with 1111111111.  She ends up with the bit-string 1111101111 which she then transmits through a binary-symmetric channel to Bob.  Let us say that this bit-string experiences bit-flips on fourth and eighth bits.  Bob receives the bit-string 1101100111.  Since he knows that Alice used the $[10,1,10]$ code as her cover-code, he knows that to unmask the cover he needs to XOR with 1111111111, which he does and obtains the string 0010011000 to which he now applies the one-time pad to recover 0000011100.  He now uses the shared secret key in order to determine the bit locations of the inner $[3,1,3]$ code.  He knows that the three-bit code comprises of the third, sixth, and eighth bits.  So he recovers the bit-string 011, to which he applies error-correction and obtains 111, which he correctly decodes as 1.  If the channel had flipped two bits of the inner code, then Bob would have incorrectly decoded the steganographic bit.  Alice can avoid this problem by using a bigger code for her inner-code block.  In this example had she used the $[5,1,5]$ code as her inner-code, then Bob would have been able to correct for up to two errors.  Moving to larger block codes unfortunately leads to a decrease in the transmission rate.
\begin{figure}
[ptb]
\begin{center}
\includegraphics[
natheight=11.0in,
natwidth=8.5in,
height=5.0in,
width=3.5in
]%
{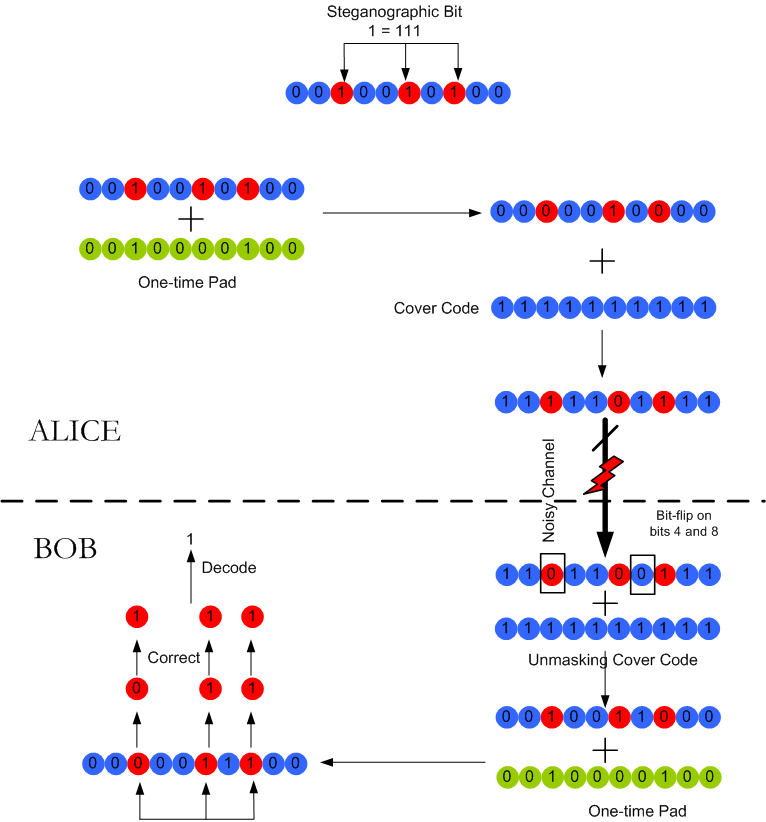}
\caption{Classical Steganographic Protocol.  The protocol begins in the top half.  Alice's steganographic bits are shown in red.  The one-time pad is shown in as green discs.  The bold dashed line shows that Alice and Bob are separated over space and Alice uses a BSC to transmit the codeword to Bob.  The bits colored red comprise the inner [3,1,3] block code, and the blue discs comprise the [10,1,10] outer block code.  The noisy channel flips the fourth and eighth bits.  The shared key is used by Alice and Bob to specify the location of the steganographic bits and to generate the one-time pad.}
\label{fig:classicalstegoprotocol}
\end{center}
\end{figure}

\begin{table}
\begin{center}
\begin{tabular}
[c]{|c|c|c|}
\hline
Steganographic Bit & Encoding Bit &  Bit String \\ \hline \hline
0 & 0 & 000 \\ \hline
1 & 0 & 111 \\ \hline
&&      001 \\
0 & 1 & 010 \\
&&      001 \\ \hline
&&      110 \\
1 & 1 & 101 \\
&&      011 \\ \hline
\hline
\end{tabular}
\end{center}
\caption{Encoding for transmitting one steganographic bit over a bit-flip channel with $[3,1,3]$ as an inner-code and $[N,1,N]$ as the outer code.}
\label{tbl:threeinN}
\end{table}

\subsection{Three Bits In $N$ Bits Encoding}
\label{sec:ThreeInN}
We now give a specific example where Alice uses the $[3,1,3]$ repetition code as her inner code, while using an $N$-bit repetition code as her outer code.  As in the previous chapter, we frame the problem as a linear-optimization one.  We present the encoding in Table~\ref{tbl:threeinN}.  We introduce the variable $\CY_{ijk}$, which is the probability of choosing to hide steganographic bit $i$ using the encoding bit $j$, and where we allow $k$ bit-flips on the outer-code.  So $i$ refers to the first column of Table~\ref{tbl:threeinN}, while $j$, and $k$ refer to the second and third columns, respectively.  With probability $\CQ$ Alice uses the encoding $\CY_{ijk}$ to transmit a steganographic bit, and
\begin{figure}[htp]
  \begin{center}
    	 \includegraphics[height = 4.0in,width = 4.0in]{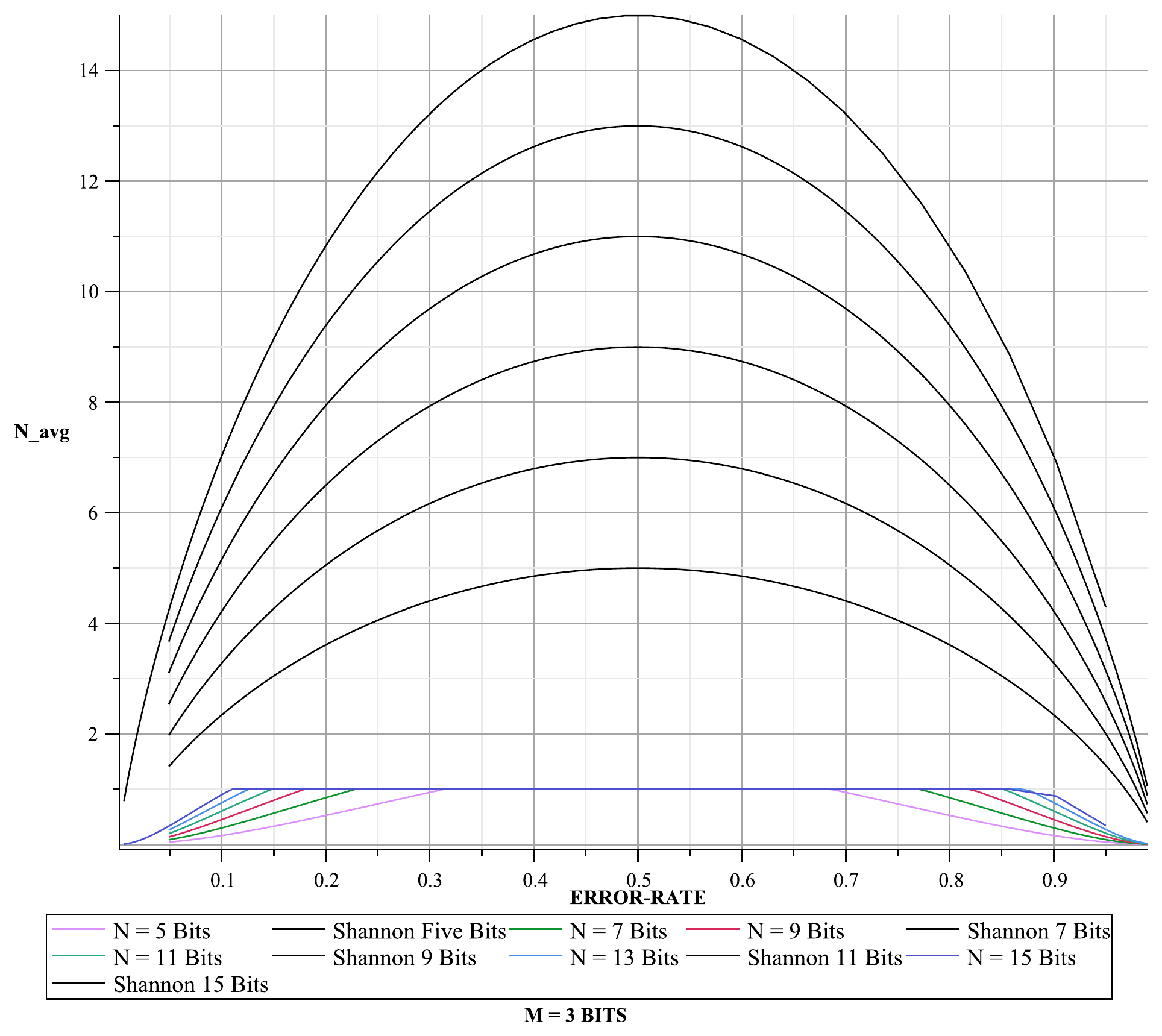}
  \end{center}
  \caption{Three Bits in N Bits Encoding. We plot the channel error-rate $p$ on the X-axis and the average number of transmitted steganographic bits , $N_{avg}$, on the Y-axis.  Here $M$ refers to the length of the inner code while $N$ refers to the length of the entire block code.  The curves in black show the Shannon entropy of the channel with the corresponding block code.}
  \label{fig:ThreeInN}
\end{figure}
with probability $1-\CQ$ she sends no steganographic bit to Bob.  However, when she chooses to send no steganographic bit to Bob, she applies $l$ errors randomly to the $N$-bit code block with probability $q_{l}$.  Alice does this because in case the error-rate of the channel is small and she cannot hide no steganographic bit, she may as well fool Eve by applying random errors which would still match the probability distribution $p$ of the channel.  The set of linear equality constraints that she must satisfy are:
\begin{eqnarray}
\label{eqn:threeinN-1}
p_{0} & = & {N\choose 0}(1-p)^{N} = \CY_{000} + q_{0}~,\\
p_{1} & = & {N\choose 1}p(1-p)^{N-1} = \CY_{010}+\CY_{001} + q_{1}~, \\
p_{2} & = & {N\choose 2}p^{2}(1-p)^{N-2} = \CY_{010} + \CY_{001} + \CY_{002} + q_{2}~,\\
&& \vdots \nonumber \\
p_{N} & = & {N \choose N}p^{N} = \CY_{10(N-3)} + q_{N}~.
\end{eqnarray}
We also insist that the steganographic bit that Alice chooses to transmit to Bob should be independent of her choice of encoding.  To match the latter constraints Alice must satisfy:
\begin{eqnarray}
\label{eqn:threeinN-2}
\tilde{\CY}_{00} = \tilde{\CY}_{10}~, \\
\tilde{\CY}_{01} = \tilde{\CY}_{11}~,
\end{eqnarray}
where
\begin{equation}
\label{eqn:threeinN-3}
\tilde{\CY}_{ij} = \sum_{k = 0}^{N-3}\CY_{ijk}~.
\end{equation}
In addition to the above equality constraints, we also want conservation of probabilities:
\begin{equation}
\label{eqn:threeinN-4}
\sum_{ijk} \CY_{ijk} + \sum_{l} q_{l} = 1~.
\end{equation}
Alice must satisfy the following linear inequality constraints:
\begin{eqnarray}
\label{eqn:threeinN-5}
\CY_{ijk} \geq 0, \forall i,j,k~, \\
q_{l} \geq 0, \forall l~.
\end{eqnarray}
Alice wants to maximize the probability of transmitting a stegaongraphic bit to Bob with each use of the channel.  So she must maximize the following objective function:
\begin{equation}
\label{eqn:threeinN-6}
N_{avg} = \sum_{ijk}\CY_{ijk}~.
\end{equation}
We simulated the above protocol keeping the inner [3,1,3] code fixed, while using the outer code $[N,1,N]$ for $N = 5,7,9,11,13$, and $15$.  From Figure~\ref{fig:ThreeInN} it is clear that as we increase the length of the outer code block from $N = 5$ to $N = 15$, she can utilize more of the error-rate region in order to transmit a steganographic bit to Bob.   In the next section we generalize this by allowing an arbitrary inner-code.

\subsection{$M$ Bits In $N$ Bits Encoding}
\label{sec:MinN}
In this section we detail how Alice can transmit $M$ steganographic bits to Bob over a noisless binary symmetric channel with a bit-flip rate of $p$.  Once again we frame this problem as a linear optimization one, where we are optimizing over the average number of steganographic bits that Alice can transmit to Bob without apprising Eve of any suspicious activity---making the whole protocol look innoncent to Eve.  Let $S_{M}$ be the subset of the bits that comprise the inner-code.  Let $\bar{S} = S_{N}$ be the subset of the bits that comprise the outer $[N,k_{c}]$ code.  Let $\CY_{ij}$ be the probability of flipping $i$ bits of $S_{M}$, and
\begin{figure}[htp]
  \begin{center}
    	 \includegraphics[height = 4.0in,width = 4.0in]{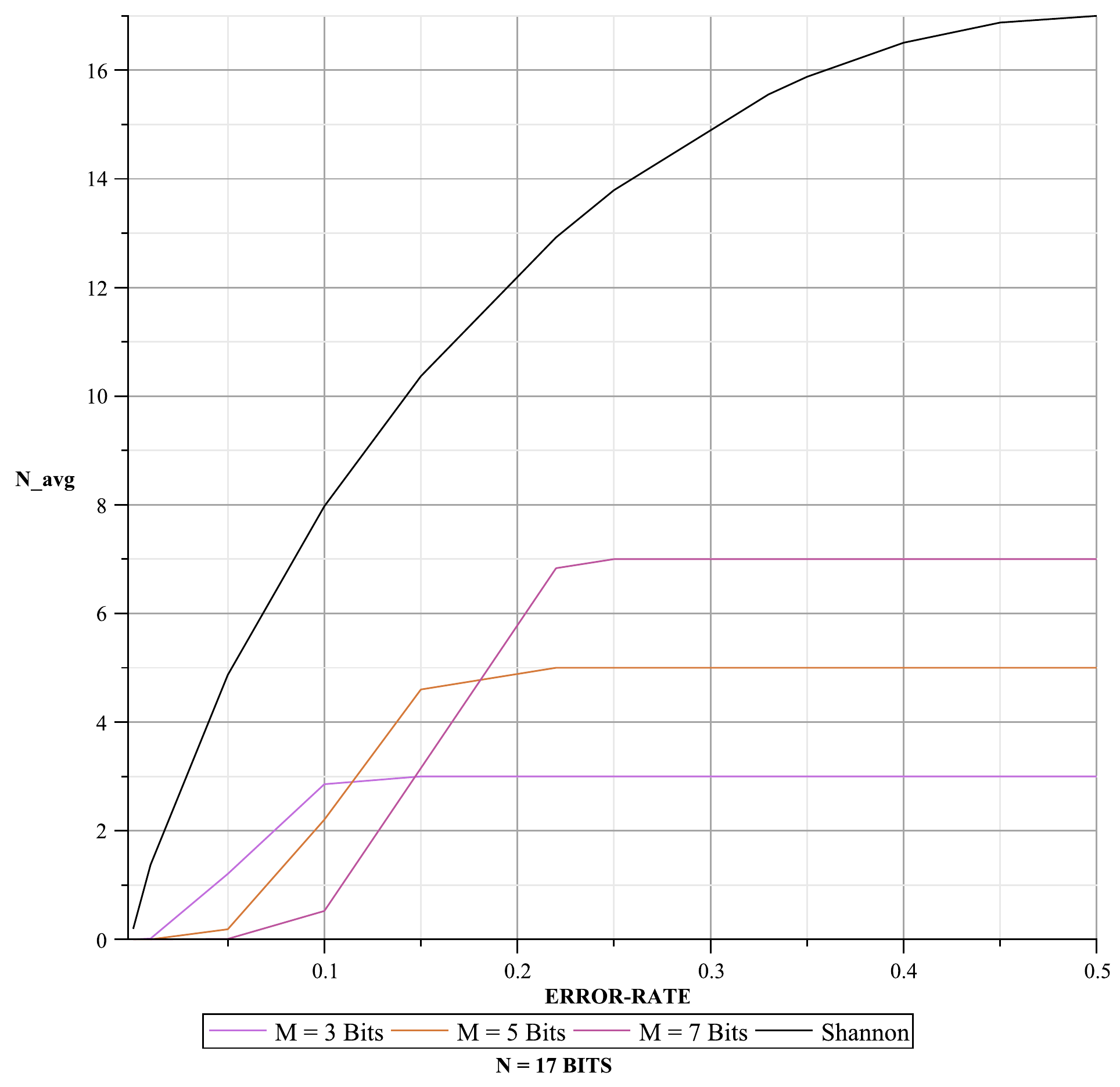}
  \end{center}
  \caption{$M$ Bits in $N$ Bits Encoding. We plot the channel error-rate $p$ on the X-axis and the average number of transmitted steganographic bits , $N_{avg}$, on the Y-axis.  Here $M$ refers to the length of the inner code while $N$ refers to the length of the entire block code.  The black curve shows the Shannon entropy of the channel when $N = 17$ bits.}
  \label{fig:MInN}
\end{figure}
$j$ bits of $S_{N}$ where $0 \leq j \leq N-M$.  When Alice decides to send no steganographic bits to Bob she applies random bit-flips to her $N$-block code with probability $q_{l}$, where $0 \leq l \leq N$.  Let $\CQ = \sum_{ij}\CY_{ij}$, be the total probability of sending a $M$ steganographic bits to Bob, and let $1-\CQ = \sum_{l}q_{r}$ be the total probability of transmitting no steganographic bits to Bob.  We would like to minimize the latter quantity or equivalently maximize $\CQ$.  For this scheme to be useful Bob must be able to decode Alice's messages, which means that the outer code should be able to correct up to $t_{s} = \left\lfloor \frac{d_{s}-1}{2}\right\rfloor$ errors.

The first set of linear equality constraints that Alice must satisfy is the probability distribution that must match the error-rate of the bit-flip channel.  These are:
\begin{eqnarray}
\label{eqn:MinN-1}
p_{0} & = & {N\choose 0}(1-p)^{N} = \CY_{00} + q_{0}~,\\
p_{1} & = & {N\choose 1}p(1-p)^{N-1} = \CY_{01}+\CY_{10} + q_{1}~, \\
p_{2} & = & {N\choose 2}p^{2}(1-p)^{N-2} = \CY_{20} + \CY_{11} + \CY_{02} + q_{2}~,\\
&& \vdots \nonumber \\
p_{N} & = & {N \choose N}p^{N} = \CY_{M(N-M)} + q_{N}~.
\end{eqnarray}
The second set of linear equality constraints is on the bits that comprise the inner-code:
\begin{equation}
\label{eqn:MinN-2}
\tilde{\CY}_{k} = \frac{1}{2^{M}} {M \choose k} \sum_{i = 0}^{M}\sum_{j = 0}^{N-M}\CY_{ij}~,
\end{equation}
where
\begin{equation}
\label{eqn:MinN-3}
\tilde{\CY}_{k} = \sum_{j = 0}^{N-M} \CY_{kj}~.
\end{equation}
Alice must also satisfy the following:
\begin{equation}
\label{eqn:MinN-4}
\sum_{ij}\CY_{ij} + \sum_{l}q_{l} = 1~.
\end{equation}
As in the previous optimization problems, Alice must satisfy the following linear inequality constraints:
\begin{eqnarray}
\label{eqn:MinN-5}
\CY_{ij} \geq 0, \forall i,j~, \\
q_{l} \geq 0, \forall l~.
\end{eqnarray}
Alice wishes to maximize the average number of steganographic bits that she can transmit to Bob.  So the objective function is:
\begin{equation}
\label{eqn:MinN-6}
N_{avg} = \sum_{ij}\CY_{ij}~.
\end{equation}
From Figure~\ref{fig:MInN} it is clear that the three-bit code outperforms both the five and seven bit codes initially, but is then outperformed by the higher bit codes when the channel bit-flip rate increases.  One can also infer from the plot (though we have not proved this) that with increasing block-size the average number of steganographic bits will reach the Shannon capacity of the channel.

The protocol above is for a bit-flip channel which has no intrinsic noise.  The flips that Eve observes are artificially induced by Alice and are incorporated in such a way as to match with Eve's belief of the error-rate of the channel.  For the inner-code above, Alice and Bob do not use an actual error-correcting code.  However, if there is intrinsic noise in the channel due to the environment, then they must use an error-correcting code for their inner-bits.  This could be an $[M, k_{s}]$ code, where $M \ll N$, and where now Alice instead of transmitting $M$ bits to Bob, can only transmit $k_{s}$ bits.  Clearly the transmission rate goes down, but that is the price Alice must pay for Bob to have any chance of decoding the message properly, and correctly acquiring the right steganographic bits.  We now proceed to give an exact formula for the amount of secret key consumed in each round of the protocol as a function of $p$ and the intrinsic noise $\delta{p}$.

\subsection{Key Consumption Rate}
\label{sec:KCR}
In this section we detail how many bits of the shared secret key Alice and Bob consume per block of code that Alice transmits to Bob.  We define this quantity as the key-consumption rate $\CK_{p}^{\delta{p}}$.  There are ${N\choose{M}}$ ways of choosing $M$ locations out of $N$ in an $N$ bit code-block.  So the total number of bits that Alice and Bob need to specify the subset $S_{M}$ which comprises the bits for
\begin{figure}[htp]
  \begin{center}
    	 \includegraphics[height = 4.0in,width = 4.0in]{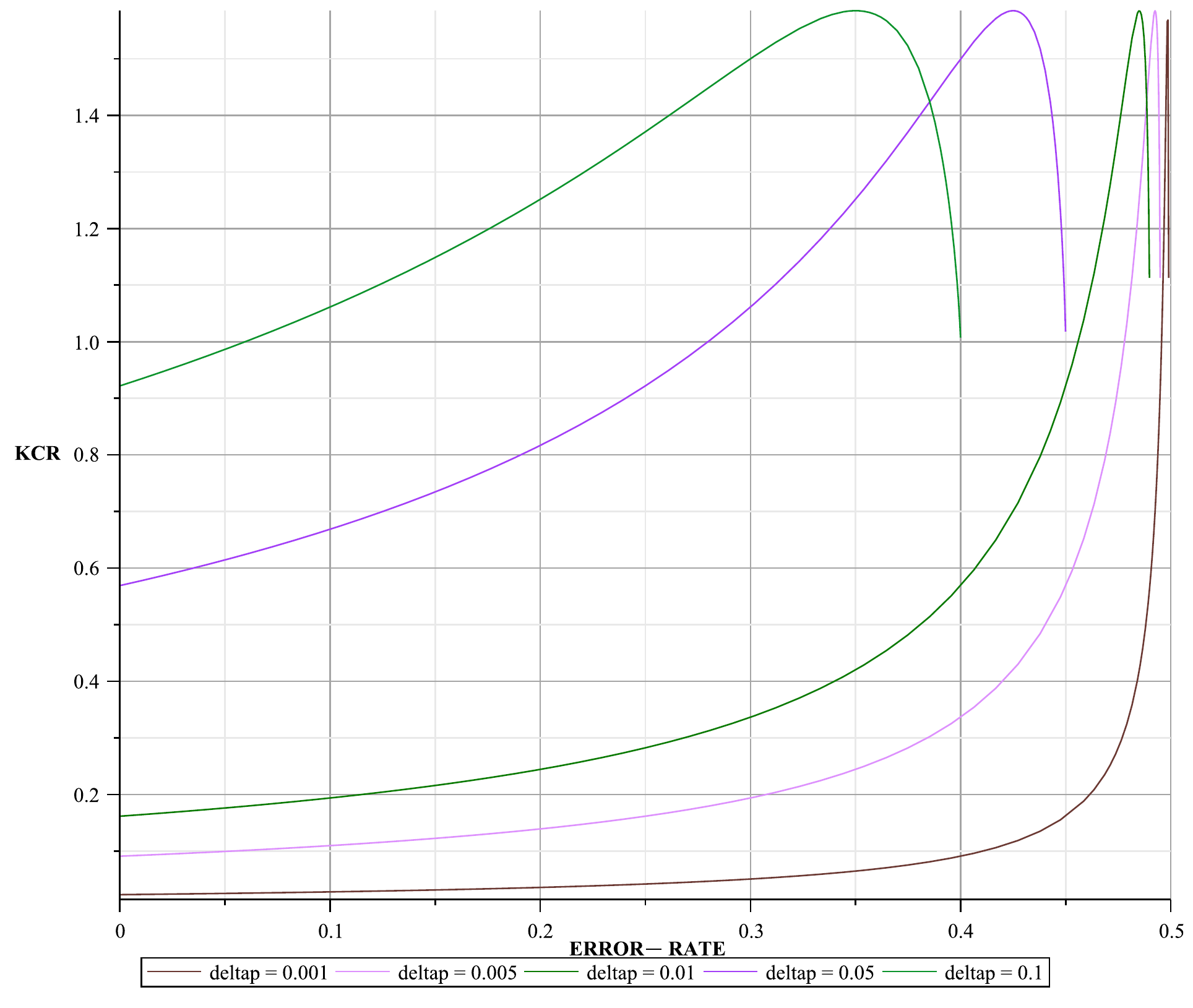}
  \end{center}
  \caption{Key Consumption Rate. We plot the channel error-rate $p$ on the X-axis and the key consumption rate \textbf{KCR} on the Y-axis.  Here $\delta{p} = 0.001,0.005,0.01,0.05$, and $0.1$.}
  \label{fig:KCR}
\end{figure}
the inner code is $\log{N \choose M}$.  They also need $M$ bits of key for the one-time pad to encrypt their steganographic bits.  So the total bit-length $\left| \CK \right|$ of the shared secret key is:
\begin{equation}
\label{eqn:kcr1a}
\left| \CK \right| = \log{N \choose M} + M~.
\end{equation}
After applying Stirling's approximation to the right hand side of Equation~(\ref{eqn:kcr1a}), we get the following:
\begin{equation}
\label{eqn:kcr-1}
\left| \CK \right| \approx N\log N - M \log M - (N-M)\log(N-M) + M~.
\end{equation}
Replace $M$ in Equation~(\ref{eqn:kcr-1}) by $M \approx 2Nq$, where $q$ is the rate at which Alice applies errors to the $N$-bit code-block in order to emulate the actual probability distribution of the binary symmetric channel with bit-flip rate $p$.  We get the following:
\begin{eqnarray}
\label{eqn:kcr-2}
\left| \CK \right| & \approx & N\log N -2Nq\log(2Nq) - (N-2Nq)\log(N-2Nq) + 2Nq \nonumber \\
& \approx & N \log N - 2Nq\left(\log 2 + \log N + \log q\right) - N(1-2q)\log N(1-2q) + 2Nq \nonumber \\
& \approx & N \log N - 2Nq - 2Nq\log N -2Nq \log q \nonumber \\
& - & N(1-2q)\log N - N(1-2q)\log(1-2q) + 2Nq~.
\end{eqnarray}
Now collect all the $\log N$ terms in Equation~(\ref{eqn:kcr-2}).
\begin{eqnarray}
\label{eqn:kcr-3}
&& \approx \left(N - 2Nq - N + 2Nq\right)\log N \nonumber \\
&& - 2Nq - 2Nq\log q - N\log(1-2q) + 2Nq\log(1-2q) + 2Nq~.
\end{eqnarray}
The first term in Equation~(\ref{eqn:kcr-3}) is zero. Moreover the two $2Nq$ terms cancel each other.  We proceed further with the simplification below:
\begin{eqnarray}
\label{eqn:kcr-4}
\left| \CK \right| & \approx & 2Nq \bigg(\log(1-2q) - \log q - \frac{1}{2q}\log(1-2q)\bigg) \nonumber \\
& \approx & 2Nq \bigg(\log(1-2q) - \log q - \log(1-2q)^{\frac{1}{2q}}\bigg) \nonumber \\
& \approx & 2Nq \left(\log(1-2q) - \left(\log q + \log(1-2q)^{\frac{1}{2q}}\right)\right) \nonumber \\
& \approx & 2Nq \left(\log(1-2q) - \left(\log q(1-2q)^{\frac{1}{2q}}\right)\right) \nonumber \\
& \approx & 2Nq\bigg(\log \frac{(1-2q)}{q(1-2q)^{\frac{1}{2q}}}\bigg) \nonumber \\
& \approx & N\bigg(2q\log \frac{(1-2q)}{q(1-2q)^{\frac{1}{2q}}}\bigg)~.
\end{eqnarray}
We want to express the key-consumption rate in terms of the error-rate $p$ of the channel, and the extra noise $\delta p$.  We can substitute $q = \frac{\delta{p}}{1-2p}$ in Equation~(\ref{eqn:kcr-4}):
\begin{equation}
\label{eqn:kcr-5}
\left| \CK \right| \approx N\Bigg(\frac{2\delta{p}}{1-2p}\log{\frac{\beta}{\big(\frac{\delta{p}}{1-2p}\big)\beta^{\big(\frac{1-2p}{2\delta{p}}\big)}}}\Bigg)~.
\end{equation}
where $\beta = \frac{1-2(p+\delta{p})}{1-2p}$.  So the key-consumption rate $\CK_{p}^{\delta{p}} = \frac{\left| \CK \right|}{N}$ is:
\begin{equation}
\label{eqn:kcr-6}
\fbox{$\CK_{p}^{\delta{p}} = \Bigg(\frac{2\delta{p}}{1-2p}\log{\frac{\beta}{\big(\frac{\delta{p}}{1-2p}\big)\beta^{\big(\frac{1-2p}{2\delta{p}}\big)}}}\Bigg)$}
\end{equation}
From Figure~(\ref{fig:KCR}) one can see that as the error-rate increases Alice and Bob consume key at a faster rate.  They need more bits to specify the subset $S_{M}$ as Alice needs to flip more bits in order to emulate the additional noise.  For a higher $\delta{p}$ our intuition tells us that Alice and Bob ought to use a longer key-length, as Alice would be choosing more bits to flip (on the average).

\section{Concluding Remarks}
We presented two different protocols on how to hide classical bits by first encoding Alice's steganographic information into the syndromes of an error-correcting code and then later using the codewords themselves to encode the hidden information.  Both techniques have their pros and cons.  The syndrome encoding quickly fails if we try to determine the optimal number of stego bits that Alice would like to send to Bob.  The second technique of encoding stego bits as codewords of an error-correcting code gives gives us a better handle on the steganographic rate and the rate of key consumption.  While we have presented these two encodings for a finite sized code block, we would like to extend these ideas to the case where $N$ is large.  In Chapter~\ref{chap:quantumsteg} we present results on the steganographic rate using syndrome encoding in the asymptotic sense.  We have pushed the asymptotic results to a later chapter because we need to introduce the notion of the diamond norm.  We use this norm to define the security of our quantum steganographic protocols.  Our main aim in constructing a model of classical steganography was to gain intuition on how we can extend our classical model to a quantum one.  In the next chapter we present our formalism for quantum steganography.

\chapter{Quantum Steganography}
\label{chap:quantumsteg}
\begin{saying}
Marvelous, what ideas the young people have these days. \\
But I don't believe a word of it.\\
--- \textit{Albert Einstein}
\end{saying}
\lettrine{I}n this chapter we introduce our formalism for quantum steganography.  It was important to establish an intuition, a deep understanding of a classical model of steganography before we could proceed to cast the problem in its quantum form.  Considering the number of papers being published in quantum information science each day, it came as a surprise to us that only four papers had been published in this subject area!  Before we present our protocols for quantum steganography, we provide a brief history of the work that has been done in quantum steganography by other researchers.

Julio Gea-Banacloche introduced the idea of hiding secret messages in the form of error syndromes by deliberately applying correctable errors in a quantum state encoded using the three-bit repetition code, $[[3,1,3]]$~\cite{Banacloche}.  He also discussed the idea of ``reverse encoding'' to transmit a quantum steganographic message which is obtained by applying a superposition of bit-flip errors on the codeword of the three-bit code.  If the receiver Bob knows exactly which error-correcting code was used to encode the quantum state he can read the hidden message in the error syndromes of the code.  Eve can monitor the channel and can perform measurements to diagnose errors.  The downside to this protocol is that if Eve knows the exact code that Alice and Bob are using then she can measure the codewords and undo the errors introduced by Alice and thereby transmit the wrong steganographic message to Bob.  Gea-Banacloche suggested a way around this problem by insisting that Alice and Bob share a secret bit.  Alice uses this secret bit to encode her codewords in either of two mutually non-orthogonal bases.  If Eve now proceeds to measure the codewords, she will be able to diagnose the correct errors only half of the time, and half of the time she will disturb the state of the codewords which Bob will be able to detect and thereby abort the protocol.  One can strengthen this protocol by insisting Alice choose from a larger set of mutually non-orthogonal states rather than just two.  In such a case Alice and Bob will have to share a longer secret key, instead of a single bit.  One of the central notions in steganography is that a message should look innocent to Eve.  We tackle this issue extensively in this thesis, but which Gea-Banacloche does not address in his paper.

Curty~et.~al. propose three different quantum steganographic protocols~\cite{curty-santos2004}.  In the first protocol they hide bits 0 or 1, by replacing a noisy qubit with $\ket{+}$ or $\ket{-}$ respectively.  In a modified protocol they hide two classical bits into a noisy qubit using super-dense coding.  In a third protocol they teleport a steganographic qubit by communicating the classical bits that Bob needs over a classical steganographic channel in order to perform the correct measurement so that he can retrieve the qubit.  They do not address the issue of an innocent message, or give key-consumption rates.  Natori provides a very rudimentary treatment of quantum steganography which is once again just a slight modification of super-dense coding~\cite{natori2006}.  Martin introduced the notion of steganographic communication of quantum information in~\cite{martin2008}.  His protocol is a modification of Bennett and Brassard's quantum-key distribution protocol (QKD).  He hides a steganographic channel in the QKD protocol.

We can imagine a day when quantum networks become ubiquitous.  Quantum steganography could provide a way to embed a stego message into some quantum data of interest which Alice could then send to Bob.  Alice could use a few steganographic qubits located at random locations as a ``watermark.''  Both of them could exploit quantum correlations such as entanglement between various steganographic qubits to determine whether the quantum information was tampered by Eve as it passed through a quantum channel.  An advantage that steganography has over standard encryption schemes is that while users or systems could be communicating over a network through steganography, this communication could go on for a very long time, completely undetected.  In an encryption protocol this is not necessarily the case.  Quantum steganography has far-reaching consequences, and it may give us a measure of security which classical steganography may not be able to parallel.

We begin by showing how one can hide quantum information in the noise of a depolarizing channel, using a shared classical secret key between Alice and Bob.  In our first quantum steganographic protocol the channel is intrinsically noiseless (i.e., all noise is controlled by Alice), and in the second case the channel has its own intrinsic noise (not controlled by Alice and Bob).  We calculate the amount of secret key consumed.  We will later present a quantum steganographic protocol for general quantum channels.  We will also look at the question of whether Alice and Bob can send a finite amount of hidden information, or can actually communicate at a nonzero asymptotic rate (given an arbitrarily large secret key).  This depends on Eve's knowledge of the physical channel, and Alice and Bob's knowledge of Eve's expectations.  Finally we will address the question of security.  This is two-fold: first, can Eve detect that a secret message has been sent, and second, can she read the message?

\section{The Depolarizing Channel}
The quantum analog of the classical \textit{binary symmetric channel} (BSC) is the {\it depolarizing channel} (DC) which is one of the most widely used quantum channel models:
\begin{equation}
\rho \rightarrow \CN\rho = (1-p) \rho + \frac{p}{3} X\rho X + \frac{p}{3} Y\rho Y + \frac{p}{3} Z\rho Z~.
\label{DC}
\end{equation}
That is, each qubit has an equal probability of undergoing an $X$, $Y$, or $Z$ error.  Applying this channel repeatedly to a qubit will map it eventually to the maximally mixed state $I/2$.  We can rewrite this channel in a different but equivalent form:
\begin{equation}
\CN = (1-4p/3) \CI + (4p/3) \CT~.
\label{twirl}
\end{equation}
where  $\CI\rho = \rho$ and $\CT\rho = (1/4)(\rho + X\rho X + Y\rho Y + Z\rho Z)$.  The operation $\CT$ is {\it twirling}:  it takes a qubit in any state $\rho$ to the maximally mixed state $I/2$.  If we rewrite the channel in this way, instead of applying $X$, $Y$, or $Z$ errors with probability $p/3$, we can think of removing the qubit with probability $4p/3$, and replacing it with a maximally mixed state.  This picture makes the steganographic protocol more transparent.  We will first assume that the actual physical channel between Alice and Bob is noiseless.  All the noise that Eve sees is due to deliberate errors that Alice applies to her codewords.
\begin{enumerate}
\item Alice encodes a covertext of $k_c$ qubits into $N$ qubits with an $[[N,k_c]]$ quantum error-correcting code (QECC).
\item From (\ref{twirl}), the DC would maximally mix $Q$ qubits with probability $p_Q$ where
\begin{equation}
p_Q = {N \choose Q} (4p/3)^Q (1-4p/3)^{N-Q}~.
\label{binomial2}
\end{equation}
For large $N$, Alice can send $M = (4/3)pN(1-\delta)$ stego qubits, where $1 \gg \delta\gg$ \\
$\sqrt{(1-4p/3)/(4p/3)N}$.  (The chance of fewer than $M$ errors is negligibly small.)
\item Using the shared random key (or shared ebits), Alice chooses a random subset of $M$ qubits out of the $N$, and swaps her $M$ stego qubits for those qubits of the codeword.  She also replaces a random number $m$ of qubits outside this subset with maximally mixed qubits, so that the total $Q=M+m$ matches the binomial distribution (\ref{binomial2}) to high accuracy.
\item Alice ``twirls'' her $M$ stego qubits using $2M$ bits of secret key or $2M$ shared ebits.  To each qubit she applies one of $I$, $X$, $Y$, or $Z$ chosen at random, so $\rho \rightarrow \CT\rho$.  To Eve, who does not have the key, these qubits appear maximally mixed.  (Twirling can be thought of as the quantum equivalent of a one-time pad.)
\item Alice transmits the codeword to Bob.  From the secret key, he knows the correct subset of $M$ qubits, and the one-time pad to decode them.
\end{enumerate}
This protocol transmits $(4/3)pN(1-\delta)$ secret qubits from Alice to Bob~\ref{fig:depolprotocol}.
\begin{figure}
[ptb]
\begin{center}
\includegraphics[
natheight=8.386600in,
natwidth=8.973300in,
height=3.5in,
width=4.0in]%
{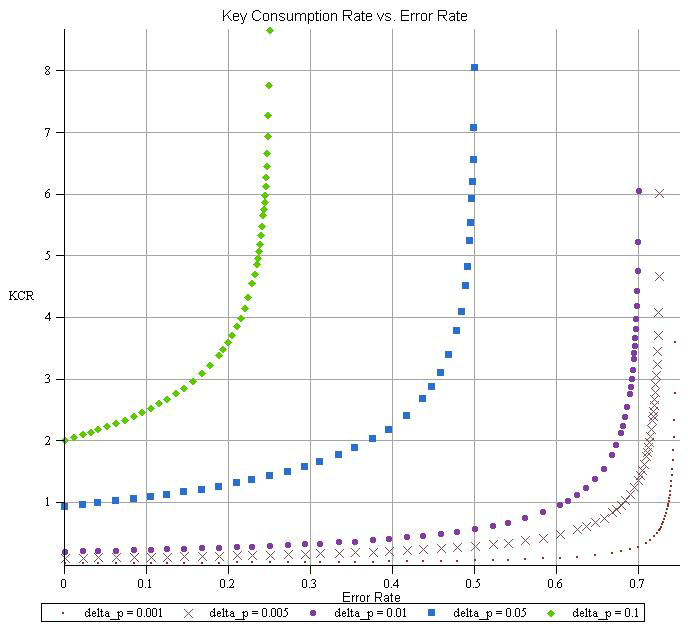}%
\caption{We plot the key consumption rate (KCR) as a function of the error-rate p of the channel.}%
\label{fig:quantumKCR}%
\end{center}
\end{figure}
If the channel contains intrinsic noise, Alice will first have to encode her $k_s$ stego qubits in an $[[M,k_s]]$ QECC, swap those $M$ qubits for a random subset of $M$ qubits in the codeword, and apply the twirling procedure.  This twirling does not interfere with the error-correcting power of the QECC if Bob knows the key.  Assuming the physical channel is also a DC with error rate $p$, and that Alice emulates a DC with error rate $q$, the effective channel will appear to Eve like a DC with error rate $p + q(1-4p/3) \equiv p + \delta p$.  The rate of transmission $k_s/N$ will depend on the rate of the QECC used to protect the stego qubits.  For a BSC this would be $(1-\delta)(1-h(p))\delta p/(1-2p)$.  However, for most quantum channels (including the DC) the achievable rate is not known.

The secret key is used at two points in these protocols.  First, in step 3 Alice chooses a random subset of $M$ qubits out of the $N$-qubit codeword.  There are $C(N,M)$ subsets, so roughly $\log_2 C(N,M)$ bits are needed to choose one.  Next, in step 4, $2M$ bits of key are used for twirling.  This gives us
\begin{equation}
n_k \approx \log_2 \left(\begin{array}{c} N \\ M \end{array}\right) + 2M
\label{qkeyconsume}
\end{equation}
bits of secret key used.  Define the key consumption rate $\CK=n_k/N$ to be the number of bits of key consumed per qubit that Alice sends through the channel.
We use $M \approx 4qN/3$ and $q \approx \delta p/(1-4p/3)$ to express $\CK$ in terms of $p$, $\delta p$, and $N$~\ref{fig:KCR}:
\begin{equation}
\label{KCR}
\CK \approx \log_{2}\left[(4/\beta)^{\beta}(1-\beta N)^{\beta-1}\right]~,\ \ \ \beta \equiv 4\delta{p}/(3-4p)~.
\end{equation}
Alice can consume fewer bits of key if Bob and she have access to a source that produces maximally mixed states.  This would allow them to bypass the twirling procedure.
The protocols given above perform well in emulating a depolarizing channel.  However, there are far more general channels than these, and the protocols may not work well, or at all, in these cases.  If one has a channel that can be written
\begin{equation}
\rho \rightarrow \CN\rho = (1-p_T + p_E) \CI \rho + p_T \CT\rho + p_E \CE\rho
\label{general_channel}
\end{equation}
where $\CE$ is an arbitrary error operation, one can still use the above protocols to hide approximately $p_T N$ stego bits or qubits, while generating $p_E N$ random errors of type $\CE$.  But for some channels, $p_T$ may be very small or zero.  How should we proceed?  Moreover, hiding stego qubits locally as apparently maximally-mixed qubits sacrifices some potential information.  The {\it location} of the error---that is, the choice of the subset holding the errors---could also be used to convey information, potentially increasing the rate and reducing the amount of secret key or shared entanglement required.

A different approach is instead to encode information in the {\it error syndromes}.  For simplicity, we consider the case when $N$ is large.  In this case, it suffices to consider only {\it typical errors}.  We begin with the case where the physical channel is noise-free.

For large $N$, almost all (in probability, $1-\epsilon$) combinations of errors on the individual qubits will correspond to one of the set of \textit{typical errors}.  There are roughly $2^{sN}$ of these, and their probabilities $p_e$ are all bounded within a range $2^{-N(s+\delta)} \le p_e \le 2^{-N(s-\delta)}$.  The number $s$ is the entropy of the channel on one qubit; for the BSC $s=h(p)=-p\log_{2}p-(1-p)\log_{2}(1-p)$, and for the DC $s = -(1-p)\log_{2}(1-p)-p\log_{2}p/3$.  We label the typical error operators $E_0, E_1, \ldots, E_{2^{sN}-1}$, and their corresponding probabilities are $p_j$.  A good choice of QECC for the cover text will be able to correct all these errors.  We make the simplifying assumption that the QECC is {\it nondegenerate}, so each typical error $E_j$ has a distinct error syndrome  $s_j$.

Ahead of time, Alice and Bob partition the typical errors into $C$ roughly equiprobable sets $S_k$, so that
\begin{equation}
\sum_{E_j\in S_k} p_j \approx \frac{1}{C},~\forall k~.
\end{equation}
As far as possible, the errors in a given set should be chosen to have roughly equal probabilities.  The maximum of $C$ is roughly $C\approx 2^{N(s-\delta)}$, and $k=0,\ldots,C-1$.  We can now present a new quantum steganographic protocol, using error syndromes to store information.
\begin{enumerate}
\item Alice prepares $k_c$ qubits of cover text in a state $\ket{\psi_c}$.
\item Alice's secret message is a string of $\log_{2} C \approx N(s-\delta)$ qubits, in a state
\begin{equation}
\ket{\psi_s} = \sum_{k=0}^{C-1} \alpha_k \ket{k}~.
\label{stego_state}
\end{equation}
She ``twirls'' each qubit of this string, using $2N(s-\delta)$ bits of the secret key or shared ebits, to get a maximally mixed state.  To this, she appends $N-k_c-(s-\delta)N$ extra ancilla qubits in the state $\ket0$ to make up a total register of $N-k_c$ qubits.
\item Using the shared secret key, Alice chooses from each set $S_k$ a typical error $E_{j_k}$ with syndrome $s_{j_k}$.  She applies a unitary $U_S$ to the register of $N-k_c$ qubits, that maps \\
$U_S\left(\ket{k}\otimes\ket0^{\otimes N-k_c-(s-\delta)N}\right) = \ket{s_{j_k}}$.  She appends this register to the cover qubits in state $\ket{\psi_c}$, then applies the encoding unitary $U_E$.  Averaging over the secret key, the resulting state will appear to Eve like $\rho \approx \sum_{j=0}^{2^{nS}-1} p_j E_j \ket{\Psi_c}\bra{\Psi_c} E_j^\dagger$, which is effectively indistinguishable from the channel being emulated acting on the encoded cover text.
\item Alice sends this codeword to Bob.  If Eve examines its syndrome, she will find a typical error for the channel being emulated.
\item  Bob applies the decoding unitary $U_D = U_E^\dagger$, and then applies $U_S^\dagger$ (which he knows using the shared secret key).  He discards the cover text and the last $N-k_c-(s-\delta)N$ ancilla qubits, and undoes the twirling operation on the remaining qubits, again using the secret key.  If Eve has not measured the qubits, he will have recovered the state encoded by Alice.
	\end{enumerate}
This protocol may easily be used to send classical information by using a single basis state rather than a superposition like (\ref{stego_state}).  The steganographic transmission rate $\CR$ is roughly $\CR\approx s-\delta\rightarrow s$.  The rate of transmission $s$ is higher than the rate $4p/3$ of our first protocol.  This protocol used $2N(s-\delta)$ bits of secret key (or ebits) for twirling in step 2, and roughly $N\delta$ bits of secret key in choosing representative errors $E_{j_k}$ from each set $S_k$ in step 3.  So the key rate is roughly $\CK\approx 2s-\delta\rightarrow 2s$, better than the first protocol in key usage per stego qubit transmitted.  However, this encoding is much trickier in the case where the channel contains intrinsic noise.

\begin{figure}
[ptb]
\begin{center}
\includegraphics[
natheight=11.0in,
natwidth=12.5in,
height=5.0in,
width=5.0in
]%
{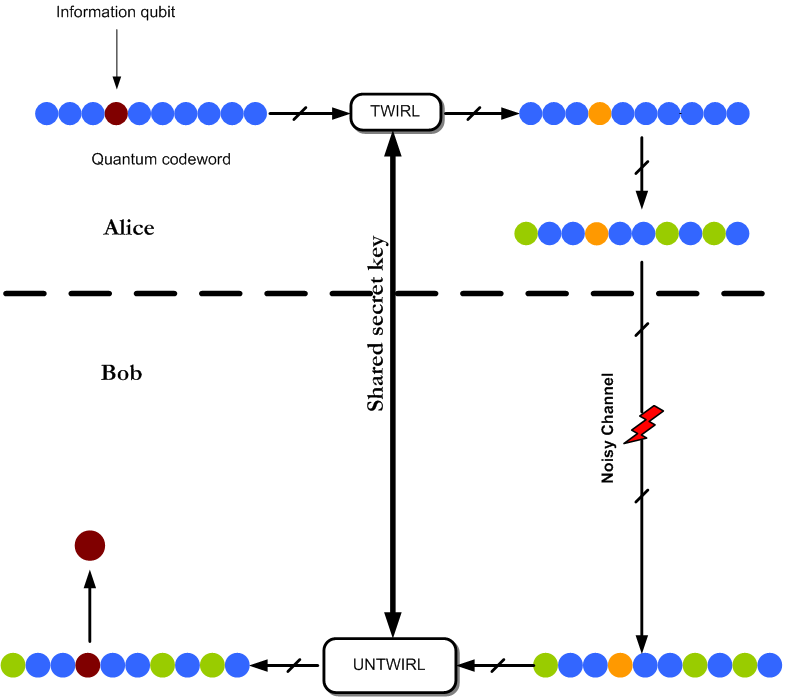}
\caption{Quantum Steganographic Protocol.  Alice hides her information qubit (solid brown circle) by swapping it in with a qubit of her quantum codeword.  She uses her shared secret key with Bob to determine which qubit to swap.  She uses the shared key again to twirl the information qubit.  She further applies random depolarizing errors to the rest of the qubits of the codeword (shown in green).  She sends the codeword through a depolarizing channel to Bob who uses the shared secret to correctly apply the untwirling operation, followed by locating and swapping out Alice's original information qubit.}
\label{fig:depolprotocol}
\end{center}
\end{figure}

In principle this quantum steganographic protocol can be used when the channel contains noise.  The steganographic qubits are first encoded in a QECC to protect them against the noise in the channel.  In practice, for many channels this can be difficult:  the effects of errors on the space of syndromes look quite different from a usual additive error channel.  Also, unlike the depolarizing channel, general channels when composed together may change their type.  However, by drawing on codes with suitable properties, the problem of designing steganographic protocols for general channels may be simplified.  We discuss a simple example in the supporting online material (SOM), but the solution for a general channel is a problem for future work.

What is the standard of security for a stego protocol?  There are two obvious considerations.  First, if Eve becomes suspicious, can she read the message?  By the use of one-time pads or twirling, Alice and Bob can prevent this from happening.

The more important question is, can Alice and Bob avoid arousing Eve's suspicions in the first place?  In order to do this, the messages that Alice sends must emulate as closely as possible the channel that Eve expects.  We can make this condition quantitative.  Let $\CE_C$ be the channel on $N$ qubits that Eve expects, and let $\CE_S$ be the effective channel that Alice and Bob produce with their steganographic protocol.  Then the protocol is secure if $\CE_S$ is $\epsilon$-close to $\CE_C$ in the diamond norm $\left\|\CE_S - \CE_C\right\|_\diamond \le \epsilon$ for some small $\epsilon>0$.  The diamond norm is directly related to the probability for Eve to distinguish $\CE_C$ from $\CE_S$ under ideal circumstances (i.e., when she can control both the inputs and the outputs), and therefore puts an upper bound on her ability to distinguish them in practice.

For a simple example, the difference between two DCs applied to $N$ qubits has norm
\begin{equation}
\label{dc_diamond}
\left\|\CN_r^{\otimes N} - \CN_{p}^{\otimes N}\right\|_{\diamond} = \sum_{j=0}^N \left(\begin{array}{c} N \\ j \end{array}\right) \left| r^j(1-r)^{N-j} - p^j(1-p)^{N-j}\right| ,
\end{equation}
where $p$ is the error-rate of the channel Eve expects and $r=p+\delta p$ is the error-rate of the steganographic channel that emulates Eve's expected channel.  If we make $\delta p < \epsilon\sqrt{p(1-p)/N}$ then we can make this norm as small as we like, while communicating $O(\delta pN) = O(\epsilon\sqrt{N})$ secret qubits.  This indicates that even if Eve has exact knowledge of the channel, Alice and Bob can in principle send an arbitrarily large (but finite) amount of information without arousing Eve's suspicion, by choosing a sufficiently small $\delta p$ and large $N$.  If Eve's knowledge of the channel is imperfect, Alice and Bob can do even better, communicating steganographic information at a nonzero rate.  If Eve is constantly monitoring the channel over a long period of time, and if she has exact knowledge of the channel then she will eventually learn that Alice and Bob are communicating with each other steganographically.  Moreover, with constant measurement Eve can disrupt the superpositions of the steganographic qubits and prevent any information from ever reaching Bob, effectively flooding the quantum channel with noise.

If Alice and Bob have shared ebits, they can perform measurements on each of their halves and distill correlated random bits.  Moreover, with shared ebits Alice can send her quantum information to Bob via quantum teleportation by sending only classical bits through the channel.  These classical bits are the result of her measurement on her half of the ebits and her stego qubits.  To Eve who may be monitoring the channel, these bits will look maximally mixed (random).  For her to change the outcome of what Bob receives on his end, Eve would have to disrupt the bits.  So if Eve is measuring the channel continuously, Alice and Bob can still send quantum information to each other using their shared ebits.

\section{The Diamond Norm}
\label{sec:diamondnorm}
In this section we gather the definition of the diamond norm and some of its relevant properties to derive the norm of the difference between $N$ uses of two binary-symmetric channels (BSC) and two depolarizing channels (DC).  We refer the reader to John Watrous's lecture notes for the definition and properties of the diamond norm~\cite{Watrous}.  As mentioned in the main text the diamond norm give us a measure of how ``close'' or similar two channels can be when they transform an arbitrary density matrix from one Hilbert space to another.  More formally let $\CN$ be some arbitrary super-operator, and let $\CN:L\left(\CV\right)\rightarrow L\left(\CW \right)$, where $L\left(.\right)$ is a space of linear operators on the Hilbert spaces $\CV$ and $\CW$.  Then one can define the diamond-norm of $\CN$ as:
\begin{equation}
\label{eqn:qsecurity-2}
\left\|\CN\right\|_{\diamond} \equiv \left\|I_{L\left(\CV\right)} \otimes \CN\right\|_{tr}~,
\end{equation}
where $\left\|\CN\right\|_{tr}$ is defined as:
\begin{equation}
\label{eqn:qsecurity-3}
\left\|\CN\right\|_{tr} \equiv \max\left\{\left\|\CN\left(O\right)\right\|_{tr} : O \in L\left(\CV \right),\left\|O\right\|_{tr}=1 \right\}~.
\end{equation}
The maximization in (\ref{eqn:qsecurity-3} is over all density matrices.  When the Hilbert space is infinite dimensional we take the supremum of the set defined in Definition(\ref{eqn:qsecurity-3}).
\subsection*{Binary Symmetric Channel}
Let $0 < p < 1/2$ be the rate at which Alice flips the qubits of her codeword.  Let $r \equiv p + \delta{p}$ be the rate at which the BSC flips qubits, where $\delta{p}$ is some additional noise which is not under the control of either Alice or Bob.  We assume that $0 < p < r < 1/2$ because at $p = 1/2$ the channel has zero capacity to send information and $p > 1/2$ means that more qubits are being flipped which is unnatural for this channel.  For a single qubit $(N = 1)$ let $\CN_{p}$ be the BSC that Alice applies to an arbitrary single-qubit density operator $\rho$:
\begin{equation}
\label{eqn:alicechannel}
\CN_{p}\rho = (1-p)\rho + pX \rho X~,
\end{equation}
and let $\CN_{r}$ be the actual BSC
\begin{equation}
\label{eqn:BSC}
\CN_{r}\rho = (1-r)\rho + rX \rho X~.
\end{equation}
We can now express the difference of the two channels as:
\begin{equation}
\label{eqn:channeldiff}
(\CN_{r} - \CN_{p})\rho = (p-r)\rho + (r-p)X \rho X
\end{equation}
We can express the diamond norm of the difference of the channels $\CN_{p}$ and $\CN_{r}$ as:
\begin{eqnarray}
\label{eqn:diamondnormdiff}
 \left\|\CN_{r} - \CN_{p}\right\|_{\diamond} & = & \underset{\rho}{\operatorname{\max}}\left\|(I \otimes (\CN_{r} - \CN_{p}))\rho\right\|_{tr} \\
& = & (r - p) \underset{\rho}{\operatorname{\max}}\left\|(I \otimes I)\rho (I \otimes I) - (I \otimes X)\rho(I \otimes X)\right\|_{tr}~.
\end{eqnarray}
When we substitute $\rho = \psi \otimes \ket{0}\bra{0}$ ($\psi$ is some arbitrary density operator) in the above equation we achieve the maximum.
\begin{eqnarray}
\label{eqn:singlequbitdiff}
\left\|\CN_{r} - \CN_{p}\right\|_{\diamond} & = & (r-p)\left\|\psi \otimes \ket{0}\bra{0} - \psi \otimes \ket{1}\bra{1}\right\|_{tr}  \\
\label{eqn:singqubitdiff1}
 & \leq & \left\|\psi \otimes \ket{0}\bra{0}\right\|_{tr} + \left|-1\left\|\right|\psi \otimes \ket{1}\bra{1}\right\|_{tr} \\
\label{eqn:singqubitdiff2}
 & = & (r-p)\left\|\psi\right\|_{tr}\left\|\ket{0}\bra{0}\right\|_{tr} + \left\|\psi\right\|_{tr}\left\|\ket{1}\bra{1}\right\|_{tr} \\
\label{eqn:singqubitdiff3}
 & = & (r-p)(1 + 1) \\
\label{eqn:singqubitdiff4}
 & = & 2(r-p) \\
\label{eqn:singqubitdiff5}
 & = & 2(p + \delta{p} - p) \\
\label{eqn:singqubitdiff6}
 & = & 2\delta{p}.
\end{eqnarray}
In Equation~(\ref{eqn:singqubitdiff1}) we use the triangle inequality and in Equation~(\ref{eqn:singqubitdiff2}) we use the fact that for any two linear operators $A$ and $B$, the trace norm of their tensor product is equal to the product of their trace norms, i.e., $\left\|A \otimes B\right\|_{tr} = \left\|A\right\|_{tr}\left\|B\right\|_{tr}$.  We would like an expression for the optimal probability to correctly distinguish two channels.
\begin{equation}
\label{eqn:p_opt}
P_{opt} = \frac{1}{2} + \frac{1}{4}\left\|\CN_{r} - \CN_{p}\right\|_{\diamond}~.
\end{equation}
So for a single-qubit use
\begin{equation}
\label{eqn:p_opt1}
P_{opt} = \frac{1}{2}(1+\delta{p})~.
\end{equation}

For the case where we have two qubits, we can write Alice's BSC as:
\begin{equation}
\label{eqn:doubleBSC}
(\CN_{p} \otimes \CN_{p})\rho = (1-p)^2\rho + p(1-p)X_{1}\rho X_{1} + p(1-p)X_{2}\rho X_{2} + p^{2}X_{1}X_{2}\rho X_{1}X_{2}~,
\end{equation}
where $X_{1} \equiv X \otimes I$ and $X_{2} \equiv I \otimes X$, and $X_{1}X_{2} \equiv X \otimes X$.  We can similarly calculate $\CN_{1} \otimes \CN_{1}$.  We can now write the difference between the two channels as:
\begin{align*}
\label{eqn:doublediff}
(\CN_{r} \otimes \CN_{r}-\CN_{p} \otimes \CN_{p})\rho & = (r^2-2r+2p-p^2) \\
& + (r-r^2-p+p^2)(X_{1}\rho X_{1} + X_{2}\rho X_{2}) \\
& + (r^2 - p^2)X_{1}X_{2}\rho X_{1}X_{2}~.
\end{align*}
The diamond norm of the difference between two BSC on two qubits can be expressed as:
\begin{equation}
\label{eqn:twoqubitmax}
\left\|\CN_{r} \otimes \CN_{r} - \CN_{p}\otimes \CN_{p}\right\|_{\diamond}  =  \underset{\rho}{\operatorname{\max}}\left\|(I \otimes (\CN_{r} \otimes \CN_{r} - \CN_{p}\otimes \CN_{p}))\rho\right\|_{tr}~.
\end{equation}
We use a similar construction from the single-qubit case to maximize the right side of Equation~(\ref{eqn:twoqubitmax}).  Letting $\rho = \psi \otimes \ket{00}\bra{00}$ we get:
\begin{equation}
\label{eqn:doublediff1}
\left\|\CN_{r} \otimes \CN_{r} - \CN_{p}\otimes \CN_{p}\right\|_{\diamond}  = \left|(1-r)^2-(1-p)^2\right| + 2\left|r(1-r)-p(1-p)\right| + \left|r^2-p^2\right|~.
\end{equation}
Given our constraints that $0 < p < r < 1/2$, the first term on the right side of Equation~(\ref{eqn:doublediff1}) is negative while the second and third terms are positive.  This give us:
\begin{eqnarray}
\label{eqn:doublediff2}
\left\|\CN_{r} \otimes \CN_{r} - \CN_{p}\otimes \CN_{p}\right\|_{\diamond} & = & 2(r-p)(2-r-p) \\
& = & 2\delta{p}(2-2p-2\delta{p})~.
\end{eqnarray}
So in the double-qubit case $P_{opt}$ is:
\begin{equation}
\label{eqn:p_opt3}
P_{opt} = \frac{1}{2}(1 + \delta{p}(2-2p-2r))~.
\end{equation}
If we observe Equation(~\ref{eqn:doublediff1}) carefully we find that the terms are distributed binomially.  For the case where we have $N$ qubits, we can use $\rho = \psi \otimes \ket{00\cdots 0}\bra{00\cdots 0}$ to maximize the diamond norm for $N$ uses of BSC to get:
\begin{equation}
\label{eqn:diamond_N}
\left\|\CN_r^{\otimes N} - \CN_{p}^{\otimes N}\right\|_{\diamond} = \sum_{j=0}^N \left(\begin{array}{c} N \\ j \end{array}\right) \left| r^j(1-r)^{N-j} - p^j(1-p)^{N-j}\right|~.
\end{equation}

\subsection*{Depolarizing Channel}
The calculation of the diamond norm of the difference between $N$ uses of two depolarizing channels (DC) is similar to the calculation of BSC that we performed in the previous section.  The expression for the channel is
\begin{equation}
\label{eqn:alicedepol}
\CN_{p}\rho = (1-p)\rho + (p/3_(X\rho X + Y\rho Y + Z\rho Z)~.
\end{equation}
Eve sees a channel with a somewhat higher rate $r=p+\delta p$.  As in the BSC case we assume that $0 < p < r < 1/2$.  For $N = 2$ case the difference between the two depolarizing channels is:
\begin{align*}
\label{eqn:depoldiff1}
(\CN_{r}\otimes\CN_{r} - \CN_{p}\otimes\CN_{p})\rho & = ((1-r)^2-(1-p)^2)\rho \\
& + ((1-r)(r/3) - (1-p)(p/3))(X_{1}\rho X_{1} + \cdots + Z_{2}\rho Z_{2}) \\
& + ((r/3)^2 - (p/3)^2)(X_{1}X_{2}\rho X_{1}X_{2} + \cdots + Z_{1}Z_{2}\rho Z_{1}Z_{2})~.
\end{align*}
The density matrix that maximizes the trace norm is $\rho = \psi \otimes \ket{\Phi^{+}}\bra{\Phi^{+}}$, where $\ket{\Phi^{+}} = 1/\sqrt{2}(\ket{00} + \ket{11})$, and $\psi$ is some arbitrary single-qubit density operator.
\begin{align*}
\left\|\CN_{r} \otimes \CN_{r} - \CN_{p}\otimes \CN_{p}\right\|_{\diamond} & = \left|(1-r)^2 - (1-p)^2\right| \\
& + 6\left|(1-r)(r/3) - (1-p)(p/3)\right| \\
& + 9\left|(r/3)^2 - (p/3)^2\right| \\
& = \left|(1-r)^2-(1-p)^2\right| + 2\left|(1-r)r - (1-p)p\right| + \left|r^2 - p^2\right|~.
\end{align*}
After evaluating the absolute value terms, we get:
\begin{eqnarray}
\label{eqn:depol2}
\left\|\CN_{r} \otimes \CN_{r} - \CN_{p}\otimes \CN_{p}\right\|_{\diamond} & = 2(r-p)(2-r-p) \\
& = 2\delta{p}\bigg(2 - 2p - \delta{p}\bigg)~.
\end{eqnarray}
So,
\begin{equation}
\label{eqn:depol3}
P_{opt} = \frac{1}{2} + \frac{1}{2}\delta{p}\bigg(2 - 2p - \delta{p}\bigg)~.
\end{equation}
For the general case for $N$ uses of the depolarizing channel we may write the diamond norm as:
\begin{equation}
\label{eqn:depoldiamond_N}
\left\|\CN_r^{\otimes N} - \CN_{p}^{\otimes N}\right\|_{\diamond} = \sum_{j=0}^N \left(\begin{array}{c} N \\ j \end{array}\right) \left| r^j(1-r)^{N-j} - p^j(1-p)^{N-j}\right|~,
\end{equation}
which is exactly the same expression as for the BSC.

\subsection*{Achievable Rate for Protocol 2}
We will work out the simplest example---the BSC in the case where the physical channel is noise-free.  The errors in the codewords that Alice sends to Bob are binomially distributed.  Let $pN$ be the mean of this distribution and let the variance be $pN\delta$, where $0 < \delta \ll 1$.  Here $N$ is the length of each of codeword.  Let
\begin{equation}
\label{eqn:rate1}
p_{k} = {N \choose{k}}p^{k}(1-p)^{N-k}
\end{equation}
be the errors that Alice applies to her codewords.  For each $k$ from $Np(1-\delta)$ to $Np(1+\delta)$ choose $C_{k}$ strings of weight $k$.  Let
\begin{equation}
\label{eqn:rate2}
C = \sum_{k = Np(1-\delta)}^{Np(1+\delta)} C_{k}~.
\end{equation}
Let these sets of strings be called $S_{k}$, and
\begin{equation}
\label{eqn:rate3}
S = \cup_{k} S_{k}
\end{equation}
So the total number of strings in the set $S$ is $C$.  Define the probability $q \equiv 1/C$.  Then we want to satisfy $qC_{k} = C_{k}/C = p_{k}$.  Clearly we must have $C_{k} \leq {N \choose{k}}$, for all $k$.  This implies that:
\begin{align*}
\label{eqn:rate4}
C_{k}p^{k}(1-p)^{N-k} & \leq {N \choose{k}}p^{k}(1-p)^{N-k}  \\
\Rightarrow C_{k}p^{k}(1-p)^{N-k} & \leq C_{k}q \\
\Rightarrow p^{k}(1-p)^{N-k} & \leq q
\end{align*}
We want $C$ to be as large as possible, which means we want $q$ to be as small as possible.  This constraint then gives us
\begin{align*}
q & = p^{Np(1-\delta)}(1-p)^{N(1-p+p\delta)} \\
\Rightarrow C & = 1/q \\
\Rightarrow C & = p^{-Np(1-\delta)}(1-p)^{-N\left(1-p+p\delta\right)}
\end{align*}
The number of bits that Alice can send is, therefore
\begin{align*}
M & = \log_{2}C \\
& = N(-p\log_{2}p-(1-p)\log_{2}(1-p)+\delta(p\log_{2}p-p\log_{2}(1-p))) \\
& = N(h(p)-p\delta\log_{2}((1-p)/p))
\end{align*}
So with this encoding Alice can send almost $Nh(p)$ bits.

\subsection*{Diamond norm for protocol 2}
Again we consider the simplest case of the BSC.  Let $N$ be sufficiently large so that the total probability of the typical errors is $> 1-\epsilon$, and these typical errors have weight $k$ in the range $Np(1-\delta) \le k \le Np(1+\delta)$.  We divide up all errors of weight $k$ into $C_k$ partitions containing
\[
n_k\approx \frac{{N\choose{k}}}{C_k}\approx \left(\frac{1-p}{p}\right)^{k-Np(1-\delta)}
\]
errors each.  Within each set the errors are all equally likely to be chosen.  However, because the number of errors is unlikely to divide exactly evenly into $C_k$ sets, the probabilities $q_k$ of an error of weight $k$ will be slightly different from the probability $p_k = p^k(1-p)^{N-k}$ of the binomial distribution.  We can put a (not-very-tight) bound on this difference:
\begin{equation}
|q_k - p_k| < \frac{p_k}{\left(\frac{1-p}{p}\right)^{k-Np(1-\delta)}-1}
< \frac{1-p}{1-2p} p^{2k}(1-p)^{N-2k} .
\end{equation}
Plugging this into the expression for the diamond norm, we get
\begin{align*}
||\CN_p^{\otimes N} - \CN_{\rm enc}||_\diamond < & \epsilon + \sum_{k=Np(1-\delta)+1}^{Np(1+\delta)} {N\choose{k}} |p_k-q_k| \\ \tag{S37}
< & \epsilon + \left(\frac{1-p}{1-2p}\right) \left(\frac{p}{1-p}\right)^{Np(1-\delta)}
\left(\frac{1-2p+2p^2}{1-p}\right)^N ~,
\end{align*}
which is exponentially small in $N$.

\subsection*{Error-correction for protocol 2 with a noisy channel}
Since errors can act in a complicated manner on the space of syndromes, it is not entirely clear what the optimal encoding is even for a simple channel.  Here we present one encoding for the BSC that gives an achievable rate in the limit of large $N$, but it is quite likely that higher rates are possible.

In the noiseless case, it is possible to use the $C(N,M)$ strings of weight $M$ as a code---each string represents one possible weight-$M$ error.  If we then apply a BSC with probability $p$, on average $Np$ bits would be flipped.  If $Np \ll M$ then one can keep only a subset of the weight-$M$ strings, separated by a distance $>2Np$.

This encoding quickly becomes inefficient as $p$ gets larger.  Using the shared secret key, Alice can instead chose only a subset of the $N$ bits  to hold the codewords.  If this subset includes $N'$ bits, then the errors on the remaining $N-N'$ bits are irrelevant and do not need to be corrected.  The limit of this would be similar to encoding 1 in the paper, where $N' \approx 2M$.

Let $N'=qN$ for some $0<q\le1$.  The number of strings of weight $M$ is $C(qN,M)$, and there will be an average number of bit flips $pqN$ on the relevant portion of the codeword.  Keep a subset of these codewords separated by distance $2pqN$.  Decoding is done by finding the closest codeword to the output string.

As $N,M\rightarrow\infty$ then the number of codewords will go like
\[
C(N,M,p,q) \sim \frac{{qN\choose{M}}}{{qN\choose{pqN}}}~.
\]
The number of bits will be $\log_2 C(N,p,q)$.

Since $q$ is a parameter we can choose freely, we choose it to maximize the rate $\CR(N,M,p,q) \equiv (1/N) \log_2 C(N,M,p,q)$.  Using the Stirling approximation, differentiating with respect to $q$, and setting the result equal to 0, we can solve for $q$:
\[
q = \frac{M}{N}\left(\frac{2^{h(p)}}{2^{h(p)}-1}\right) ~.
\]
We can then plug this back into the formula for $\CR$.  If the physical channel has error rate $p$ and Alice is attempting to emulate a channel with error rate $p+\delta p$, then $M=N\delta p/(1-2p)$.  This gives us the following expression for the rate:
\[
\CR(p,\delta p) = - \frac{\delta p}{1-2p} \log_2\left(2^{h(p)}-1\right) .
\]
We can compare this to the rate from encoding 1, which for the BSC is $2\delta p(1-h(p))/(1-2p)$.  It is not hard to see that $\CR(p,\delta p)$ above approaches this rate as $p\rightarrow1/2$ (and both rates go to zero), but as $p\rightarrow0$ this encoding does considerably better than encoding 1.  It is quite likely, however, that there may be even more efficient encodings.

\section{Concluding Remarks}
We presented two different protocols for hiding quantum information in the codewords of a quantum error-correcting code.  Using the first protocol, Alice can hide quantum information in the actual qubits of the QEC, by sharing secret information about the location of the chosen qubits with Bob.  In the second protocol she hides quantum information by encoding it into the syndromes of a QEC.  In both protocols, she encodes the information in such a way as to match the probability distribution of the depolarizing channel.  We generalized how Alice and hide quantum information by using general quantum channels and a general QEC by utilizing typical sequences, it is not immediately clear what sort of encoding will yield an optimal rate for the transmission of quantum information between Alice and Bob. We also presented a rigorous security criterion and showed that it is always possible for Alice to send a finite amount of quantum information to Bob even when the noise is not completely under their control.  In the next chapter we present an example of how Alice hides four stego-qubits in to the perfect code and how using a particular encoding she can transmit an optimal number of stego-qubits to Bob without apprising Eve.  We essentially cast the problem of encoding as a linear optimization one.

\chapter{Hiding Quantum Information in the Perfect Code}
\label{chap:hidingqinfo}
\begin{saying}
By this it appears how necessary it is for any man that aspires to true knowledge to examine the definitions of former authors; and either to correct them, where they are negligently set down, or to make them himself. For the errors of definitions multiply themselves, according as the reckoning proceeds, and lead men into absurdities, which at last they see, but cannot avoid, without reckoning anew from the beginning; in which lies the foundation of their errors.  For between true science and erroneous doctrines, ignorance is in the middle. Natural sense and imagination are not subject to absurdity. Nature itself cannot err: and as men abound in copiousness of language; so they become more wise, or more mad, than ordinary. Nor is it possible without letters for any man to become either excellently wise or (unless his memory be hurt by disease, or ill constitution of organs) excellently foolish. For words are wise men's counters; they do but reckon by them: but they are the money of fools, that value them by the authority of an Aristotle, a Cicero, or a Thomas, or any other doctor whatsoever, if but a man. \\
---\textit{Thomas Hobbes}
\end{saying}

\lettrine{I}n the previous Chapter~\ref{chap:quantumsteg} we detailed protocols to send quantum steganographic information from Alice to Bob.  In this chapter we give specific examples of the second protocol that we outlined in the previous chapter.  In our first example Alice and Bob use the $[[5,1,3]]$ code, which is also called the ``perfect'' code.  It is called perfect because it was the first example of the smallest quantum error-correcting code that could correct an arbitrary single-qubit error.  It is a nondegenerate code because each error gets mapped to a unique syndrome.  Alice and Bob use this code to hide up to four qubits of information which they send over a channel that to Eve looks like a depolarizing channel.  Our second and third examples are based on the six-qubit code and the Steane code, respectively.  We end this chapter by showing how we can go to bigger block sizes by combining several five-qubit blocks.  By hiding a qubit of information in each block, we hope that the combined errors from all the smaller block on larger combined block is a typical error.

\section{Protocol}
The perfect code is a distance three code and so it can correct up to an arbitrary single-qubit error.  It encodes one logical qubit into five physical qubits.  As we mentioned before it is a nondegenerate code and so each single-qubit error is mapped to a unique syndrome.  The perfect code has $n - k = 5 - 1 = 4$ stabilizer generators which we list in Table~\ref{tbl:perfect_code}, along with the logical operators.  Figure~\ref{fig:five_qubit_code} shows the encoding unitary circuit for the perfect code.  We obtained this circuit using the algorithm described in Section~\ref{sec:eaqeccircuit}.  For the sake of convenience and completeness we reproduce the algorithm below to generate the encoding unitary circuit for the perfect code.
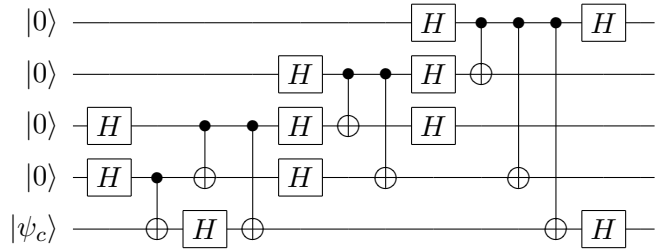
\begin{figure}
[ptb]
\begin{center}
\[
\Qcircuit@C= 0.5em @R = 0.5em @!R{
\lstick{\ket{0}} & \qw & \qw & \qw & \qw & \qw & \qw & \qw & \gate{H} & \ctrl{1} & \ctrl{3} & \ctrl{4} & \gate{H} & \qw & \qw\\
\lstick{\ket{0}} & \qw & \qw & \qw & \qw & \gate{H} & \ctrl{1} & \ctrl{2} & \gate{H} & \targ & \qw & \qw & \qw & \qw & \qw\\
\lstick{\ket{0}} & \gate{H} & \qw & \ctrl{1} & \ctrl{2} & \gate{H} & \targ & \qw & \gate{H} & \qw & \qw & \qw & \qw & \qw & \qw\\
\lstick{\ket{0}} & \gate{H} & \ctrl{1} & \targ & \qw & \gate{H} & \qw & \targ & \qw & \qw & \targ & \qw & \qw & \qw & \qw\\
\lstick{\ket{\psi_{c}}} & \qw & \targ & \gate{H} & \targ & \qw & \qw & \qw & \qw & \qw & \qw & \targ & \gate{H} & \qw & \qw\\
}
\]
\end{center}
\caption
{Encoding circuit for the perfect code.  The \textit{H} gate is a Hadamard gate, and the two-qubit gates are CNOT gates.  The first four qubits are ancilla qubits and the fifth qubit $\ket{\psi_{c}}$ is the cover-qubit.}
\label{fig:five_qubit_code}
\end{figure}
We begin by first converting the stabilizer generators in Table~\ref{tbl:perfect_code} into a
binary form which we refer to as a $Z|X$ matrix.  We obtain the the left $Z$ submatrix
by inserting a ``1'' wherever we see a $Z$ in the stabilizer generators.  We obtain the
$X$ submatrix by inserting a ``1'' wherever we see a corresponding $X$
in the stabilizer generator.  If there is a $Y$ in the generator, we insert a ``1''
in the corresponding row and column of both the $Z$ and $X$ submatrices.
The idea is to convert matrix~\ref{perfectcodemat1} to matrix~\ref{perfectcodemat8} through a series of row and column operations.
The binary form of the matrix in~\ref{perfectcodemat1} corresponds to the stabilizer generators in Table~\ref{tbl:perfect_code}.
We can use CNOT, Hadamard, Phase, and SWAP gates.
\begin{enumerate}
\item When we apply a CNOT gate from qubit $i$ to qubit $j$, it adds column $i$ to
column $j$ in the $X$ submatrix, and in the $Z$ submatrix it adds column $j$ to column $i$.
\item A Hadamard on qubit $i$ swaps
column $i$ in the $Z$ submatrix with column $i$ in the $X$ submatrix.
\item A Phase gate on qubit
$i$ adds column $i$ in the $X$ submatrix to column $i$ in the $Z$ submatrix.
\item When we apply a
SWAP gate from qubit $i$ to qubit $j$, we exchange column $i$ with column $j$ in $Z$ submatrix
and column $i$ and column $j$ in the $X$ submatrix.
\end{enumerate}
Row operations do not change the error-correcting properties of the code.  They do not
cost us in terms of gates.

\begin{table}[tbp] \centering
\begin{tabular}
[c]{c|ccccc}\hline\hline
$g_{1}$ & $X$ & $Z$ & $Z$ & $X$ & $I$\\
$g_{2}$ & $I$ & $X$ & $Z$ & $Z$ & $X$\\
$g_{3}$ & $X$ & $I$ & $X$ & $Z$ & $Z$\\
$g_{4}$ & $Z$ & $X$ & $I$ & $X$ & $Z$\\ \hline
$\overline{X}$ & $X$ & $X$ & $X$ & $X$ & $X$\\
$\overline{Z}$ & $Z$ & $Z$ & $Z$ & $Z$ & $Z$
\\\hline\hline
\end{tabular}
\caption{Stabilizer generators $g_1$, \ldots, $g_4$, and logical operators $\bar{X}$ and $\bar{Z}$ for the perfect code.  The convention in the above generators is that $Y=ZX$.}
\label{tbl:perfect_code}
\end{table}

\begin{equation}
\left[  \left.
\begin{array}
[c]{ccccc}
0 & 1 & 1 & 0 & 0 \\
0 & 0 & 1 & 1 & 0 \\
0 & 0 & 0 & 1 & 1 \\
1 & 0 & 0 & 0 & 1
\end{array}
\right\vert
\begin{array}
[c]{ccccc}
1 & 0 & 0 & 1 & 0 \\
0 & 1 & 0 & 0 & 1 \\
1 & 0 & 1 & 0 & 0 \\
0 & 1 & 0 & 1 & 0
\end{array}
\right]
\label{perfectcodemat1}
\end{equation}
Swap rows one and four.
\begin{equation}
\left[  \left.
\begin{array}
[c]{ccccc}
1 & 0 & 0 & 0 & 1 \\
0 & 0 & 1 & 1 & 0 \\
0 & 0 & 0 & 1 & 1 \\
0 & 1 & 1 & 0 & 0
\end{array}
\right\vert
\begin{array}
[c]{ccccc}
0 & 1 & 0 & 1 & 0 \\
0 & 1 & 0 & 0 & 1 \\
1 & 0 & 1 & 0 & 0 \\
1 & 0 & 0 & 1 & 0
\end{array}
\right]
\label{perfectcodemat2}
\end{equation}
Apply Hadamard to qubit one and five.
\begin{equation}
\left[  \left.
\begin{array}
[c]{ccccc}
0 & 0 & 0 & 0 & 0 \\
0 & 0 & 1 & 1 & 1 \\
1 & 0 & 0 & 1 & 0 \\
1 & 1 & 1 & 0 & 0
\end{array}
\right\vert
\begin{array}
[c]{ccccc}
1 & 1 & 0 & 1 & 1 \\
0 & 1 & 0 & 0 & 0 \\
0 & 0 & 1 & 0 & 1 \\
0 & 0 & 0 & 1 & 0
\end{array}
\right]
\label{perfectcodemat3}
\end{equation}
Apply CNOT from qubit one to qubit two, four, and five.
\begin{equation}
\left[  \left.
\begin{array}
[c]{ccccc}
0 & 0 & 0 & 0 & 0 \\
0 & 0 & 1 & 1 & 1 \\
0 & 0 & 0 & 1 & 0 \\
0 & 1 & 1 & 0 & 0
\end{array}
\right\vert
\begin{array}
[c]{ccccc}
1 & 0 & 0 & 0 & 0 \\
0 & 1 & 0 & 0 & 0 \\
0 & 0 & 1 & 0 & 1 \\
0 & 0 & 0 & 1 & 0
\end{array}
\right]
\label{perfectcodemat4}
\end{equation}
Apply Hadamard on qubit one and then swap row two with row four.
\begin{equation}
\left[  \left.
\begin{array}
[c]{ccccc}
1 & 0 & 0 & 0 & 0 \\
0 & 1 & 1 & 0 & 0 \\
0 & 0 & 0 & 1 & 0 \\
0 & 0 & 1 & 1 & 1
\end{array}
\right\vert
\begin{array}
[c]{ccccc}
0 & 0 & 0 & 0 & 0 \\
0 & 0 & 0 & 1 & 0 \\
0 & 0 & 1 & 0 & 1 \\
0 & 1 & 0 & 0 & 0
\end{array}
\right]
\label{perfectcodemat5}
\end{equation}
Apply Hadamard on qubit two and three, followed by CNOT from qubit two to qubit three and four.
\begin{equation}
\left[  \left.
\begin{array}
[c]{ccccc}
1 & 0 & 0 & 0 & 0 \\
0 & 0 & 0 & 0 & 0 \\
0 & 0 & 1 & 1 & 0 \\
0 & 0 & 0 & 1 & 1
\end{array}
\right\vert
\begin{array}
[c]{ccccc}
0 & 0 & 0 & 0 & 0 \\
0 & 1 & 0 & 0 & 0 \\
0 & 0 & 0 & 0 & 1 \\
0 & 0 & 1 & 0 & 0
\end{array}
\right]
\label{perfectcodemat6}
\end{equation}
Now apply Hadamard gate on qubit two, three, and four, followed by CNOT from qubit three to qubit four and five.
\begin{equation}
\left[  \left.
\begin{array}
[c]{ccccc}
1 & 0 & 0 & 0 & 0 \\
0 & 1 & 0 & 0 & 0 \\
0 & 0 & 0 & 0 & 0 \\
0 & 0 & 0 & 0 & 1
\end{array}
\right\vert
\begin{array}
[c]{ccccc}
0 & 0 & 0 & 0 & 0 \\
0 & 0 & 0 & 0 & 0 \\
0 & 0 & 1 & 0 & 0 \\
0 & 0 & 0 & 1 & 0
\end{array}
\right]
\label{perfectcodemat7}
\end{equation}
Apply a Hadamard on qubit five, followed by a CNOT from qubit four and five, and then ending the operations with a Hadamard on qubit three and four.
\begin{equation}
\left[  \left.
\begin{array}
[c]{ccccc}
1 & 0 & 0 & 0 & 0 \\
0 & 1 & 0 & 0 & 0 \\
0 & 0 & 1 & 0 & 0 \\
0 & 0 & 0 & 1 & 0
\end{array}
\right\vert
\begin{array}
[c]{ccccc}
0 & 0 & 0 & 0 & 0 \\
0 & 0 & 0 & 0 & 0 \\
0 & 0 & 0 & 0 & 0 \\
0 & 0 & 0 & 0 & 0
\end{array}
\right]
\label{perfectcodemat8}
\end{equation}
We have finally obtained a binary matrix that corresponds to the canonical stabilizer generators in Table~\ref{tbl:perfect_code}.  Multiplying the above operations in reverse takes us from the unencoded canonical stabilizers to the encoded ones.
\begin{figure}
[ptb]
\begin{center}
\includegraphics[
natheight=10.4in,
natwidth=16.0in,
height=3.0in,
width=5.0in
]
{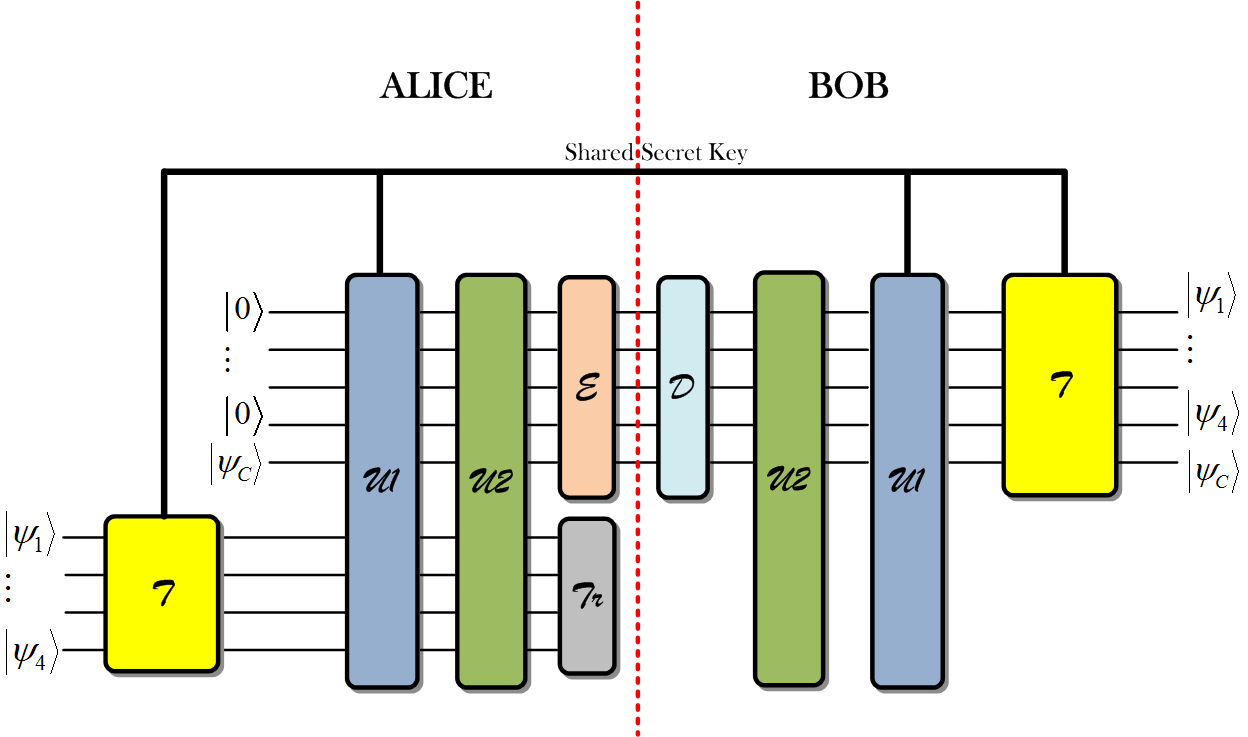}%
\caption{Steganographic protocol to hide four qubits of information using the perfect code.  Based on a classical shared secret key with Bob, Alice hides four qubits, $\ket{\psi_{1}},\ldots, \ket{\psi_{4}}$ by applying the twirling procedure to each of her four qubits and transforms each qubit into a maximally mixed state.  We call this system the stego subsystem $S$.  The twirling procedure is applied by the yellow box.  We label the first four ancilla qubits $\ket{0}$ along with the cover-qubit $\ket{\psi_{C}}$ as Alice's subsystem $A$.  The first unitary $\CU_{1}$ depicted by the blue box is applied to the subsystems $\CA$ and $\CS$, followed by unitary $\CU_{2}$.  Finally Alice applies the encoding unitary $\CE$ to subsystem $\CA$ and sends it over a channel to Bob.  She traces out subsystem $\CS$.  Bob upon receiving the codeword decodes it by applying the decoding circuit $\CD$.  He uses the shared secret key to undo the unitary transformations $\CU_{1}$ and $\CU_{2}$ that were applied by Alice.  He finally uses the shared secret key to untwirl the qubits to obtain the stego qubits.}
\label{fig:stegoprotocol_perfect_code}%
\end{center}
\end{figure}

Next we need the syndromes corresponding to single-qubit errors.  The perfect code has a total of fifteen single-qubit errors.  Each of the five qubits can undergo a single bit-flip error $X_{1}, X_{2}, \ldots,$ or $X_{5}$, a single phase-flip error $Z_{1}, Z_{2}, \ldots,$ or $Z_{5}$ or a combination of both characterized by $Y_{1}, Y_{2}, \ldots,$ or $Y_{5}$.  We also add the no error operator $IIIII$ to this list.  We express $X \otimes I \otimes I \otimes I \otimes I$ succinctly as either $XIIII$ or as $X_{1}$, and similarly with the other error operators.  We generate the syndrome of each single-qubit error operator by taking the symplectic product of the parity check matrix~\ref{perfectcodemat1} with the binary vector of the error operator corresponding to a single-qubit error.  We write the binary vector of an error operator in the same way that we constructed the rows of the parity check matrix~\ref{perfectcodemat1}.  As an example the binary vector corresponding to $XIIII$ is $\left[00000|10000\right]$.  Before we show how to take the symplectic product between a binary matrix and a vector, we must clarify the procedure between two binary vector of length two.

Suppose $u,v\in\left(  \mathbb{Z}_{2}\right)  ^{2}$. Let us employ the shorthand $u=\left[  z|x\right]  $ and $v=\left[  z^{\prime}|x^{\prime}\right]  $\ where $z$, $x$, $z^{\prime}$, $x^{\prime}\in\mathbb{Z}_{2}$. For example, suppose $u=\left[  0|1\right]  $, and suppose $v = \left[1|0\right]$.  Let $\odot$\ denote the \textit{symplectic product} between two elements $u,v\in\left(\mathbb{Z}_{2}\right)  ^{2}$:%
\begin{equation}
\label{eqn:symproddef}
u\odot v\equiv zx^{\prime}-xz^{\prime}.
\end{equation}
Then based on our example above, the symplectic product between $u$ and $v$ is:
\begin{eqnarray}
\left[0|1\right] \odot \left[1|0\right] & = &0.0 - 1.1 \nonumber \\
& = & 0 - 1 \nonumber \\
& = & 1
\end{eqnarray}
Let the binary matrix $A$ and the binary vector $v$ be of dimensions $M \times 2N$ and $2N \times 1$, respectively.  So $v$ is a column vector.  Furthermore, let $A_{ij}$ be an element of $A$ indexed by row $i$ and column $j$.  We divide the matrix $A$ into the $Z$ part and the $X$ part which we label as $A^{Z}$ and $A^{X}$ respectively.  The $Z$ part of $A$ consists of columns from $1$ to $N$, and its $X$ part consists of columns from $N+1$ to $2N$.  So $A_{ij}^{Z}$ is an element of $A$ belonging to the $Z$ part at row $i$ and column $j$.  Similarly let $v^{Z}$ and $v^{X}$ be the $Z$ and $X$ part of the binary vector $v$.  Rows of $v$ labeled from 1 to $N$ belong to the $Z$ part, whereas rows labeled from $N+1$ to $2N$ belong to the $X$ part.  $v_{i}^{Z}$ refers an element of $v$ indexed by row $i$ belonging to the $Z$ part.  Let $w$ be the symplectic product of $A$ and $v$.  We calculate each element $w_{i}$ of $w$ as:
\begin{equation}
\label{eqn:symprodmatvec}
w_{i} = \sum_{i = 1}^{M}\sum_{j = 1}^{2N} A_{ij}^{Z}v_{j}^{X}+A_{ij}^{X}v_{j}^{Z}
\end{equation}
In the above equation it should be clear that in the first part of the sum the elements of $A^{Z}$ are indexed by $j$ between 1 and $N$, whereas the elements of $v^{X}$ indexed by $i$ go from $N+1$ to $2N$.  Based on the above equation we can calculate the syndromes of single-qubit errors for the perfect code which we list in Table~\ref{tbl:syndrome_perfect_code}.
\begin{table}[tbp] \centering
\begin{tabular}
[c]{|c|c||c|c|}\hline\hline
Error   & Syndrome & Error & Syndrome\\ \hline \hline
$IIIII$ & $0000$ & $IIIIZ$ & $0001$ \\ \hline
$IIZII$ & $0010$ & $IIIXI$ & $0011$ \\ \hline
$IXIII$ & $0100$ & $IIXII$ & $0101$\\ \hline
$ZIIII$ & $0110$ & $IIYII$ & $0111$ \\ \hline
$XIIII$ & $1000$ & $IZIII$ & $1001$ \\ \hline
$IIIIX$ & $1010$ & $IIIIY$ & $1011$ \\ \hline
$IIIZI$ & $1100$ & $IYIII$ & $1101$ \\ \hline
$YIIII$ & $1110$ & $IIIYI$ & $1111$
\\\hline\hline
\end{tabular}
\caption{Syndrome table of the perfect code.  The error operators are ordered according to syndrome values.}
\label{tbl:syndrome_perfect_code}
\end{table}
We now detail the steps Alice takes to hide four stego qubits.  We label Alice's subsystem as $\CA$, and the stego subsystem as $\CS$.  The initial combined $\CA$ and $\CS$ subsystems after the $\CS$ subsystem has passed through the twirling procedure is:
\begin{equation}
\label{eqn:initialstate}
\rho^{\CS \CA} = \bigg(\frac{I}{2}\bigg)^{\CS} \otimes \bigg(\frac{I}{2}\bigg)^{\CS} \otimes \bigg(\frac{I}{2}\bigg)^{\CS} \otimes \bigg(\frac{I}{2}\bigg)^{\CS} \otimes \ket{0000}\ket{\psi_{C}}^{\CA}\bra{0000}\bra{\psi_{C}}^{\CA}.
\end{equation}
The maximally mixed state is $I/2$.  If we measure this state, we will get the result '0' with probability $1/2$ or we will get the result '1' with probability $1/2$.  So the results of the measurement are completely random.  We can rewrite the $\CS$ subsystem of the state in Equation(\ref{eqn:initialstate}) in the computational basis as:
\begin{equation}
\label{eqn:Icomp}
I^{\otimes 4} = \sum_{i = 0}^{15} \ket{k}\bra{k}.
\end{equation}
In Equation(\ref{eqn:Icomp}), we index the ket states from 0 to 15, but we must think of their binary representation.  So for example $\ket{0}$ is $\ket{0000}$.  When we combine Equation(\ref{eqn:Icomp}) with the initial state described by Equation(\ref{eqn:initialstate}), we get:
\begin{equation}
\label{eqn:initialstate1}
\rho^{\CS \CA} = \frac{1}{16}\sum_{k = 0}^{15}\ket{k}\bra{k}^{\CS} \otimes \ket{0000}\ket{\psi_{C}}^{\CA}\bra{0000}\bra{\psi_{C}}^{\CA}.
\end{equation}
We further simplify the state $\rho^{\CS \CA}$ by writing it as:
\begin{equation}
\label{eqn:initialstate2}
\rho^{\CS \CA} = \frac{1}{16}\sum_{k = 0}^{15}\ket{k}\bra{k}^{\CS} \otimes \ket{\bar{0}}\bra{\bar{0}}^{\CA}\otimes\ket{\psi_{C}}\bra{\psi_{C}}^{\CA},
\end{equation}
where $\ket{\bar{0}} \equiv \ket{0000}$.  We define the first unitary $U_{1}$ shown in the blue box in Figure~\ref{fig:stegoprotocol_perfect_code} as:
\begin{equation}
\label{eqn:unitary1}
U_{1} \equiv \sum_{i = 0}^{15} \ket{i}\bra{i}^{\CS} \otimes \CE_{i}^{\CA} \otimes \CO_{i}^{\CA}.
\end{equation}
The $E_{i}$ are the error operators listed in Table~\ref{tbl:syndrome_perfect_code}, ordered from top to bottom.  So $E_{1}$ corresponds to the error operator $IIIII$.  We define the second unitary $U_{2}$ depicted by the green box in Figure~\ref{fig:stegoprotocol_perfect_code} as:
\begin{equation}
\label{eqn:unitary2}
U_{2} \equiv \sum_{j = 0}^{15} (X^{j})^{\CS} \otimes \ket{j}\bra{j}^{\CA} \otimes I_{C}^{\CA}.
\end{equation}
In Equation(\ref{eqn:unitary2}) $X^{j}$ is a shorthand notation for $X^{j_{1}} \otimes X^{j_{2}} \otimes X^{j_{3}} \otimes X^{j_{4}}$, where each $j_{1},\ldots, j_{4}$ is either 0 or 1.  For example, $X^{0} = X^{0} \otimes X^{0} \otimes X^{0} \otimes X^{0} = I \otimes I \otimes I \otimes I$.  The operator $I_{C}$ acts on the cover-qubit (information qubit) $\ket{\psi_{C}}$.  After Alice applies the unitary $U_{1}$ and $U_{2}$, followed by the encoding unitary $\CE$, and the tracing out procedure, she must end up with a channel model that looks plausible to Eve.  In this case she is trying to match the depolarizing channel:
\begin{equation}
\label{eqn:depol_prefect_code}
\rho \rightarrow \CN(\rho) \equiv (1-p)\rho + (1-p)^{4}\frac{p}{3}\sum_{i = 0}^{4} X_{i}\rho X_{i} + Y_{i}\rho Y_{i} + Z_{i}\rho Z_{i}.
\end{equation}
We have taken each error operator $E_{i}$ from Table~\ref{tbl:syndrome_perfect_code} and listed it as a tensor product $\CE_{i}\otimes\CO_{i}$ in the first two columns of Table~\ref{tbl:encoded_errors_perfect_code}.  We will use this table as we derive the channel model.  We write the state that emerges after Alice applies $U_{1}$ to $\rho^{\CS \CA}$ as:
\begin{equation}
\label{eqn:U1}
\rho_{1}^{\CS \CA} \equiv U_{1} \rho^{\CS \CA} U_{1}^{\dagger}.
\end{equation}
Alice then applies $U_{2}$, and the subsequent state is:
\begin{equation}
\label{eqn:U2}
\rho_{2}^{\CS \CA} \equiv U_{2} \rho_{1}^{\CS \CA} U_{2}^{\dagger}.
\end{equation}
The encoding unitary that Alice applies is $U_{e} = I^{\CS} \otimes U_{e}^{\CA}$:
\begin{equation}
\label{eqn:Ue}
\rho_{e}^{\CS \CA} \equiv U_{e} \rho_{2}^{\CS \CA} U_{e}^{\dagger},
\end{equation}
followed by tracing out the $\CS$ subsystem:
\begin{equation}
\label{eqn:traceout}
Tr_{S}\left[\rho_{e}^{\CS \CA}\right] = \rho_{e}^{\CA} \rightarrow \CN(\rho_{e}) = (1-p)\rho_{e} + (1-p)^{4}\frac{p}{3}\sum_{i = 0}^{4} X_{i}\rho_{e} X_{i} + Y_{i}\rho_{e} Y_{i} + Z_{i}\rho_{e} Z_{i}.
\end{equation}
We claim that the above three unitary transformations and the partial-trace of the $\CS$ subsystem produces the correct depolarizing channel.  We now detail the protocol.  Alice and Bob a-priori know which quantum error-correcting code they will be using to send steganographic messages to each other.  Alice and Bob also share a classical secret key which is indicated by the bold black lines in Figure~\ref{fig:stegoprotocol_perfect_code}.  They require this key for two purposes.  First, it allows Alice to twirl an arbitrary stego-qubit into a maximally mixed state.  Given an arbitrary density matrix $\rho$, the twirling operator $\CT$ takes $\rho \rightarrow 1/4(\rho + X \rho X + Y \rho Y Z \rho Z) \equiv I/2$.  Alice and Bob use two bits to determine which of the four operators $I, X, Y$, or $Z$ to apply.  It is important to realize that the twirled state looks maximally mixed to Eve who does not have access to the secret key.
\begin{figure}
[ptb]
\begin{center}
\[
\Qcircuit@C= 1.5em @R = 4.0em @!R{
\lstick{\ket{0}} & \targ & \ctrl{1}  & \qw && \lstick{\ket{\psi}}\\
\lstick{\ket{\psi}} & \ctrl{-1} & \targ & \qw && \lstick{\ket{0}}\\
}
\]
\end{center}
\caption
{Swap circuit}
\label{fig:swap}
\end{figure}
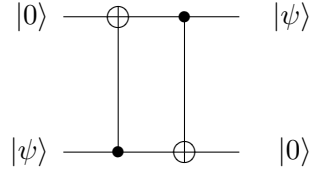
Also note that the state $\rho$ that Alice produces is an average state, the result of applying the operators $I, X, Y$, or $Z$ repeatedly.   In this protocol since Alice is hiding four stego-qubits, she and Bob consume eight bits in the twirling procedure.  Second, Alice and Bob use the key to select those error operators that have distinct syndromes.  These are the operators that go into the construction of the the unitary $U_{1}$ in Equation~(\ref{eqn:unitary1}).  The error operators $\CE_{i}^{\CA} \otimes \CO_{i}^{\CA}$ are listed in Table~\ref{tbl:encoded_errors_perfect_code}.  The intuition behind using the unitaries $U_{1}$ and $U_{2}$ is that by using them Alice can swap in her stego-qubits with the ancillas of the codeword.  They act as generalized swap operators.  We depict the usual swap circuit in Figure~\ref{fig:swap}.  When we swap in an arbitrary qubit $\ket{\psi}$ with an ancilla state $\ket{0}$ we only need two CNOT gates, instead of three.
\begin{table}[tbp] \centering
\begin{tabular}
[c]{|c|c||c||c|c|}\hline\hline
$\CE_i$   & $\CO_i$ & Encoded Error & Syndrome $i$ & Syndrome Label\\ \hline \hline
$IIII$ & $I$ & $IIIII$ & 0000 & $s_0$ \\ \hline
$IIIX$ & $Z$ & $IIIIZ$ & 0001 & $s_1$\\ \hline
$IIXI$ & $I$ & $IIZII$ & 0010 & $s_2$\\ \hline
$IIXX$ & $I$ & $IIIXI$ & 0011 & $s_3$\\ \hline
$IXII$ & $I$ & $IXIII$ & 0100 & $s_4$\\ \hline
$IXZY$ & $Y$ & $IIXII$ & 0101 & $s_5$\\ \hline
$ZXXI$ & $Z$ & $ZIIII$ & 0110 & $s_6$\\ \hline
$-IXYY$ & $Y$ & $IIYII$ & 0111 & $s_7$\\ \hline
$XIII$ & $I$ & $XIIII$ & 1000 & $s_8$\\ \hline
$XZIX$ & $I$ & $IZIII$ & 1001 & $s_9$\\ \hline
$XIXI$ & $X$ & $IIIIX$ & 1010 & $s_{10}$\\ \hline
$XIXX$ & $Y$ & $IIIIY$ & 1011 & $s_{11}$\\ \hline
$XXIZ$ & $X$ & $IIIZI$ & 1100 & $s_{12}$\\ \hline
$XYIX$ & $I$ & $IYIII$ & 1101 & $s_{13}$\\ \hline
$YXXI$ & $Z$ & $YIIII$ & 1110 & $s_{14}$\\ \hline
$XXXY$ & $X$ & $IIIYI$ & 1111 & $s_{15}$
\\\hline\hline
\end{tabular}
\caption{Encoded error operators of the perfect code.}
\label{tbl:encoded_errors_perfect_code}
\end{table}
Alice applies $U_{1}$ to the initial state $\rho^{\CS \CA}$ to get the state $\rho_{1}^{\CS \CA}$:
\begin{eqnarray}
\label{eqn:applyU1simplify}
\rho_{1}^{\CS \CA} & = & \frac{1}{16}\sum_{i,j,k = 0}^{15} \ket{i}\bracket{i}{k}\bracket{k}{j}\bra{j} \otimes \CE_{i}\ket{\bar{0}}\bra{\bar{0}}\CE_{j}^{\dagger} \otimes \CO_{i}\ket{\psi_{C}}\bra{\psi_{C}}\CO_{j}^{\dagger}, \nonumber \\
& = & \frac{1}{16}\sum_{i,j,k = 0}^{15}\delta_{ik}\delta_{kj}\ket{i}\bra{j} \otimes \CE_{i}\ket{\bar{0}}\bra{\bar{0}}\CE_{j}^{\dagger} \otimes \CO_{i}\ket{\psi_{C}}\bra{\psi_{C}}\CO_{j}^{\dagger}, \nonumber \\
& = & \frac{1}{16}\sum_{i = 0}^{15} \ket{i}\bra{i} \otimes \CE_{i}\ket{\bar{0}}\bra{\bar{0}}\CE_{i}^{\dagger} \otimes \CO_{i}\ket{\psi_{C}}\bra{\psi_{C}}\CO_{i}^{\dagger}.
\end{eqnarray}
Next Alice applies the unitary $U_{2}$ to the state $\rho_{1}^{\CS \CA}$:
\begin{eqnarray}
\label{eqn:U2apply1}
\rho_{2}^{\CS \CA} & = & \frac{1}{16}\sum_{i,k,l = 0}^{15} X^{k}\ket{i}\bra{i}X^{l} \otimes \bra{k}\CE_{i}\ket{\bar{0}}\bra{\bar{0}}\CE_{i}^{\dagger}\ket{l}\ket{k}\bra{l}\otimes\CO_{i}\ket{\psi_{C}}\bra{\psi_{C}}\CO_{i}^{\dagger}, \\
\label{eqn:U2apply3}
& = & \frac{1}{16}\sum_{i,k,l = 0}^{15} X^{k} \ket{i}\bra{i}X^{l} \otimes \bracket{k}{i}\bracket{i}{l}\ket{k}\bra{l} \otimes \CO_{i}\ket{\psi_{C}}\bra{\psi_{C}}\CO_{i}^{\dagger}, \\
\label{eqn:U2apply4}
& = & \frac{1}{16}\sum_{i,k,l = 0}^{15}X^{k}\ket{i}\bra{i}X^{l}\otimes \delta_{ki}\delta_{il}\ket{k}\bra{l}\otimes \CO_{i}\ket{\psi_{C}}\bra{\psi_{C}}\CO_{i}^{\dagger}, \\
\label{eqn:U2apply5}
& = & \frac{1}{16} \sum_{i} X^{i} \ket{i}\bra{i}X^{i} \otimes \ket{i}\bra{i}\otimes \CO_{i}\ket{\psi_{C}}\bra{\psi_{C}}\CO_{i}^{\dagger},  \\
\label{eqn:U2apply6}
& = & \frac{1}{16}\sum_{i = 0}^{15} \ket{\bar{0}}\bra{\bar{0}} \otimes \ket{i}\bra{i} \otimes \CO_{i}\ket{\psi_{C}}\bra{\psi_{C}}\CO_{i}^{\dagger},\\
& = & \frac{1}{16}\sum_{i = 0}^{15} \ket{\bar{0}}\bra{\bar{0}} \otimes \left(\CE_{i}\otimes\CO_{i}\right)\ket{\bar{0}}\bra{\bar{0}} \otimes \ket{\psi_{C}}\bra{\psi_{C}}\left(\CE_{i}\otimes\CO_{i}\right)^{\dagger}.
\end{eqnarray}
As mentioned before these error operators appear in Table~\ref{tbl:encoded_errors_perfect_code}.  In Equation~(\ref{eqn:U2apply3}) we see that each $\CE_{i}$ transforms $\ket{0}$ to the corresponding syndrome $\ket{i}$.  Also recall that $\ket{\bar{0}} \equiv \ket{0000}$.  Alice finally applies the encoding unitary $U_{e}$ to $\rho_{2}^{\CS \CA}$:
\begin{eqnarray}
\rho_{e}^{\CS \CA} & = & \frac{1}{16}\sum_{i = 0}^{15} \ket{\bar{0}}\bra{\bar{0}}^{\CS}\otimes U_{e}^{\CA}\left(\CE_{i} \otimes \CO_{i}\right)\ket{\bar{0}}\bra{\bar{0}}\otimes \ket{\psi_{C}}\bra{\psi_{C}}\left(\CE_{i}\otimes\CO_{i}\right)^{\dagger}(U_{e}^{\CA})^{\dagger}, \nonumber \\
& = & \frac{1}{16}\sum_{i = 0}^{15} \ket{\bar{0}}\bra{\bar{0}}^{\CS}\otimes U_{e}^{\CA}\left(\CE_{i} \otimes \CO_{i}\right)(U_{e}^{\CA})^{\dagger}U_{e}^{\CA}\ket{\bar{0}}\bra{\bar{0}}\otimes \ket{\psi_{C}}\bra{\psi_{C}}(U_{e}^{\CA})^{\dagger}U_{e}^{\CA}\left(\CE_{i}\otimes\CO_{i}\right)^{\dagger}(U_{e}^{\CA})^{\dagger}, \nonumber \\
& = & \frac{1}{16}\sum_{i = 0}^{15} \ket{\bar{0}}\bra{\bar{0}}^{\CS}\otimes U_{e}^{\CA}\left(\CE_{i} \otimes \CO_{i}\right)(U_{e}^{\CA})^{\dagger}\rho_{e}^{\CA}U_{e}\left(\CE_{i}\otimes\CO_{i}\right)^{\dagger}(U_{e}^{\CA})^{\dagger}, \nonumber \\
& = & \frac{1}{16}\sum_{i = 0}^{15} \ket{\bar{0}}\bra{\bar{0}}^{\CS}\otimes E_{i}^{e}\rho_{e}(E_{i}^{e})^{\dagger}.
\end{eqnarray}
Alice finally traces out the stego subsystem:
\begin{eqnarray}
Tr_{S}\left[\rho_{e}^{\CS \CA}\right] & = & \frac{1}{16}\sum_{i = 0}^{15} \bracket{\bar{0}}{\bar{0}}E_{i}^{e}\rho_{e}(E_{i}^{e})^{\dagger}, \nonumber \\
& = & \frac{1}{16}\sum_{i = 0}^{15} E_{i}^{e}\rho_{e}(E_{i}^{e})^{\dagger}.
\end{eqnarray}
So finally:
\begin{equation}
\label{eqn:encodedrho}
\rho_{e}^{\CA} = \frac{1}{16}\sum_{i = 0}^{15} E_{i}^{e}\rho_{e}(E_{i}^{e})^{\dagger},
\end{equation}
where $\rho_{e}$ is a codeword of the perfect code.  The first encoded error operator (column three) in Table~\ref{tbl:encoded_errors_perfect_code} is $IIIII$.   To plausibly use the five-qubit code, the error rate $p$ must be sufficiently low that it is rare for more than one single-qubit error to occur.  So if Eve intercepts this state and checks the error syndromes, the result should look reasonable for a codeword that has passed through the depolarizing channel.
\section{Average rate of steganographic transmission}
\label{sec:optimize}
While the encoding given in the previous section is reasonable for encoding stego-qubits in a single codeword, it would not do if Alice wants to send many messages to Bob.  It would quickly look suspicious if almost every codeword had exactly one single-qubit error.  For longer messages, Alice must most often apply no errors, from time to time apply a single error, and perhaps occasionally apply two errors (an uncorrectable error), in such a way as to match the statistics of a depolarizing channel with error probability $p$.

In this section we present details of how Alice optimizes the number of stego-qubits that she sends to Bob by encoding her stego-qubits in the syndromes of the perfect code.  She either sends no stego-qubits to Bob (and applies no errors to the codeword) or she sends four stego-qubits to him (applying either one or two errors).  We would like to maximize the number of stego-qubits that Alice can send to Bob under the constraint that Alice's encoding scheme should match the probability distribution of the depolarizing channel.  The channel model that we are interested in is:
\begin{eqnarray}
\label{eqn:depolchannel}
\CN(\rho) & \equiv & (1-p)^{5}\rho + \frac{p(1-p)^4}{3}\sum_{i = 0}^{4}\bigg(X_{i}\rho X_{i} + Y_{i}\rho Y_{i} + Z_{i}\rho Z_{i}\bigg)
+ \frac{p^2(1-p)^3}{9} \sum_{i<j} \bigg(X_{i}X_{j}\rho X_{i}X_{j} \nonumber \\
& + & X_{i}Y_{j}\rho X_{i}Y_{j} +  X_{i}Z_{j}\rho X_{i}Z_{j} + Y_{i}X_{j}\rho Y_{i}X_{j} + Y_{i}Y_{j}\rho Y_{i}Y_{j} + Y_{i}Z_{j}\rho Y_{i}Z_{j} + Z_{i}X_{j}\rho \nonumber \\
& + & Z_{i}X_{j} + Z_{i}Y_{j}\rho Z_{i}Y_{j} + Z_{i}Z_{j}\rho Z_{i}Z_{j}\bigg).
\end{eqnarray}
There are of course three, four, and five qubit errors, but it is very unlikely that such errors will occur, and so we neglect them for this paper.  The techniques we give here could easily be used to find stego-encodings for those cases as well; but it might appear suspicious if Alice and Bob used an encoding that frequently suffered uncorrectable errors.

We use three distinct classes of encodings.  The first class corresponds to no errors, and transmits no hidden information.  This always has the syndrome $s_0 \equiv 0000$ in Table~\ref{tbl:encoded_errors_perfect_code}.

Alice's second encoding class corresponds to all single-qubit errors, plus no error, with equal probability.  These are the syndromes $s_{0} \equiv 0000$ to $s_{15} \equiv 1111$ in the same table.  This is the encoding described in the previous section.  So this encoding class includes a single encoding, and transmits four stego-qubits.

The third encoding class corresponds to the set of all two-qubit errors.  There are ninety such errors.  These ninety two-qubit errors naturally divide into six sets of fifteen errors each, where each set has one error corresponding to each of the fifteen nonzero syndromes.  By also including the ``no error'' operator, each set corresponds to an encoding that can transmit four stego-qubits.  We list these six sets in Table~\ref{tbl:two-qubit_errors_perfectcode}.  When Alice sends four stego-qubits to Bob, she must use all sixteen distinct syndromes.  A single row of this table, spanning all sixteen syndromes, corresponds to a single encoding.  Each encoding proceeds exactly as described above for the single-error case, except that the operators $\CE_i$ and $\CO_i$ from Table~\ref{tbl:two-qubit_errors_perfectcode} are used in Eq.~(\ref{eqn:unitary1}) instead of those from Table~\ref{tbl:encoded_errors_perfect_code}.

We now need to solve for how often the three classes of encodings should be used to match the channel statistics.  Let $Q_0$ be the fraction of times Alice uses the (trivial) first encoding class; $Q_1$ the fraction of times that Alice uses the second (single-error) encoding class; and $Q_2$ the fraction of times that Alice uses one of the six two-error encodings (which should be used equally often).  The channel distribution constraints are as follows:
\begin{eqnarray}
\label{eqn:channelconstraint1}
p_{0} & = & (1-p)^5 = Q_0 + \frac{1}{16}(Q_1 + Q_2) , \\
p_{1} & = & 5p(1-p)^4 = \frac{15}{16}Q_1, \\
p_{2} & = & 10p^2(1-p)^3 = \frac{15}{16}Q_2 .
\end{eqnarray}
In the channel constraint equations above, $p_{0}, p_{1}$, and $p_{2}$ represent the total probability of the channel applying no errors, one error, or two errors on the codewords.  The right-hand-side of the equations represent how Alice matches the channel's probability distribution.  For example, she always applies no errors if she uses the first (trivial) encoding, and with probability $1/16$ if she uses one of the other encodings.  We solve for $Q_{0,1,2}$ to get:
\begin{eqnarray}
Q_0 &=& p_0 - (p_1+p_2)/15 , \nonumber\\
Q_1 &=& (16/15)p_1 , \\
Q_2 &=& (16/15)p_2 . \nonumber
\end{eqnarray}
Note that these numbers do not add up to 1, because we have neglected errors of weight three or more.  For small $p$ they will come close.

The average number of stego-qubits that Alice can send to Bob under the above constraints is $N_{avg} = 4(Q_1+Q_2) = (64/15)(p_1+p_2)$.  We plot this function in Figure~\ref{fig:optimize}, along with the Shannon entropy of the channel (which is the maximum possible rate at which stego information could be sent).  Note that this curve only makes sense for fairly small values of $p$; for higher values, it no longer makes sense to neglect higher-weight errors.

\begin{figure}[htp]
  \begin{center}
    	 \includegraphics[width = 3.5in]{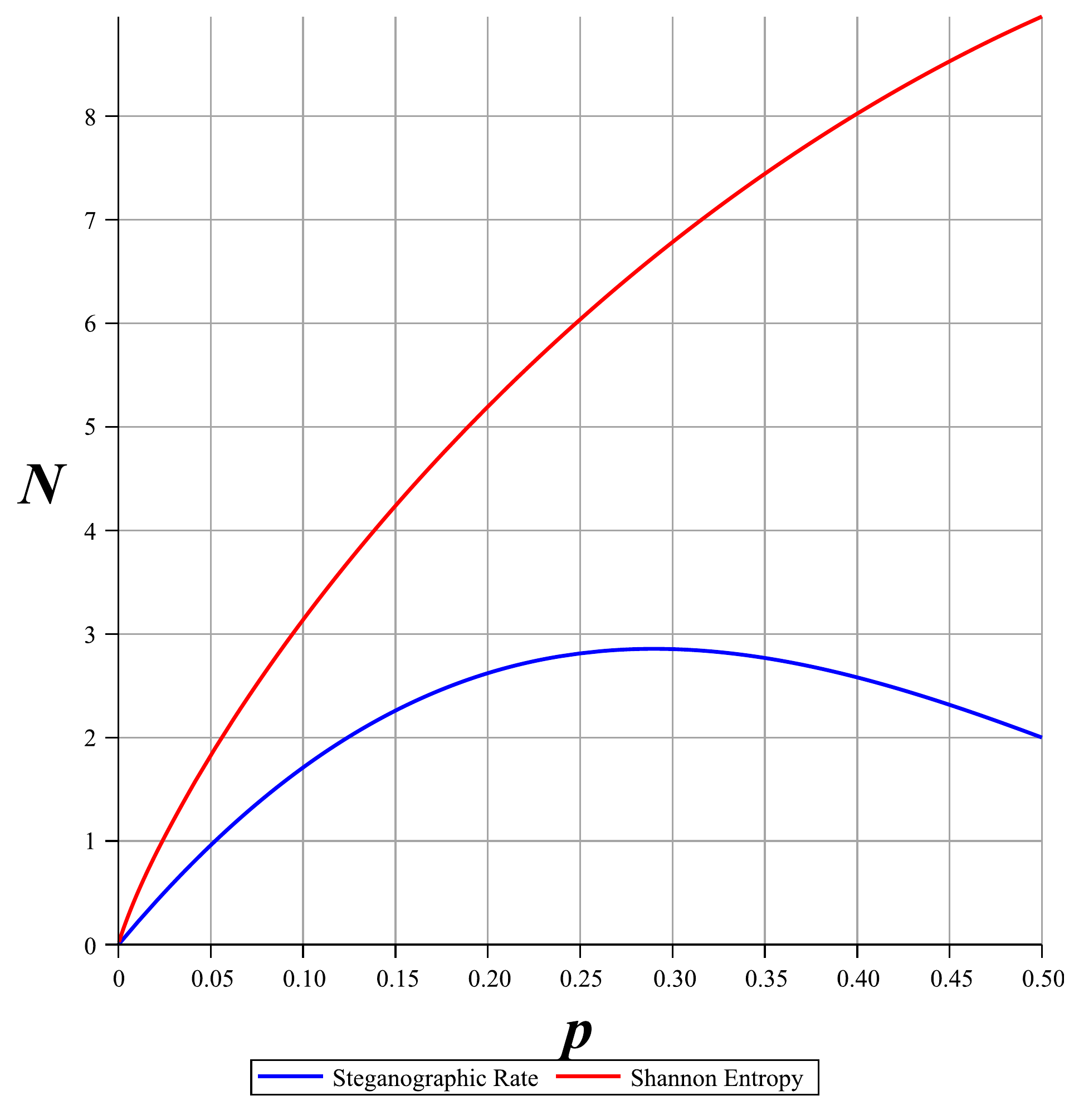}
  \end{center}
  \caption{The red curve shows the Shannon entropy of the binary-symmetric channel.  The blue curve underneath the entropy curve is the rate of steganographic information maximized for various values of the error-rate $p$ of the channel.  In the limit of large $N$, encoding in a single large block (rather than multiple sub-blocks) can achieve a rate of steganographic transmission approaching the Shannon entropy \protect{\cite{shawandbrun}}.  For  small $p$ one doesn't do too much worse using multiple finite-sized code blocks.}
  \label{fig:optimize}
\end{figure}

\section{Key Consumption}
\label{sec:keyconsume}
We define the key consumption rate as the number of classical bits of key consumed by Alice and Bob per qubit block.  Alice can send an $N$-qubit block to Bob by combining $N/5$ five-qubit blocks together.  How much shared secret key do Alice and Bob require to send steganographic information to each other?  In Section~\ref{sec:optimize} we defined $Q_{0}$ as the probability of Alice using the first encoding to send no stego-qubits to Bob.  The probability of using the second encoding is $Q_{1}$, when Alice applies single errors to her codewords.  Alice and Bob can also use two-error encodings to send four qubits to each other, of which there are a total of six as shown in Table~\ref{tbl:two-qubit_errors_perfectcode}.  We assume that each of these six encodings are equiprobable, each with probability $Q_2/6$.  Therefore, there are a total of eight different encodings that Alice and Bob must choose from in order to send steganographic information to each other.

For a single five-qubit block, Alice and Bob must share three classical bits to specify which of the eight encodings has been used.  They also require an extra eight bits for Alice to construct her four twirled qubits and for Bob to untwirl these qubits once he receives the steganographic information.  So Alice and Bob would consume twelve classical secret key bits for each five-qubit block, if they encode one block at a time, and so for an $N$-qubit block they would consume at most $12N$ classical bits of key---not a very economical protocol. In fact, it would be less than that, since Alice needs no key bits to do twirling for the first, trivial, encoding.  Taking this into account, the total key consumption rate would be $(3+8(Q_1+Q_2))/5$ secret key bits per qubit.

However, Alice and Bob can do much better than this by encoding into multiple block at once.  In the limit of large $N$, the number of bits that they require to specify the key is given by the Shannon entropy of the probability distribution of the three different encodings.  This gives us a key usage rate of
\begin{equation}
\label{eqn:keyconsume}
K =  \frac{1}{5} \left(- Q_{0}\log_{2}{Q_{0}} - Q_{1}\log_{2}{Q_{1}} - Q_2\log_{2}{Q_2/6}\right)
\end{equation}
key bits per qubit.  For small $p$ this dramatically outperforms the rate for encoding one block at a time.

\begin{table}\centering
\begin{tabular}[c]{|c|c|c||c|c|c|}
\hline \hline
Syndrome $i$ & $\CE_{i} \otimes \CO_{i}$ & Encoded Error & Syndrome $i$ & $\CE_{i} \otimes \CO_{i}$ & Encoded Error \\ \hline \hline
\multirow{6}{*}{0001} & $IZIXI$ & $XZIII$ & \multirow{6}{*}{0010} & $ZIXIZ$ & $XZIII$ \\
& $\textbf{IIZYY}$ & $\textbf{IXXII}$ && $\textbf{IZXZX}$ & $\textbf{IYIYI}$ \\
& $-ZIZYX$ & $ZIYII$ && $ZIXZY$ & $YIIZI$ \\
& $IIIXI$ & $IIZXI$ && $IIXIX$ & $XIIIX$ \\
& $ZIIYY$ & $YIIYI$ && $IZXIY$ & $IZIIY$	\\
& $-IZIYX$ & $IYIZI$ && $IIXIZ$ & $IIIXZ$ \\ \hline \hline
\multirow{6}{*}{0011} & $ZZXXZ$ & $YYIII$ & \multirow{6}{*}{0100} & $ZXIIZ$ & $ZIZII$ \\
& $\textbf{-ZIYYX}$ & $\textbf{ZIXII}$ && $\textbf{IXZZY}$ & $\textbf{IIYXI}$ \\
& $-IIYYY$ & $IXYII$ && $IXIZX$ & $XIIZI$ \\
& $IZXXX$ & $IZIIX$ && $ZXIIY$ & $YIIIX$ \\
& $IIXXY$ & $XIIIY$ && $IXIZZ$ & $IIIYY$	\\
& $IIXXZ$ & $IIZIZ$ && $IXZZX$ & $IIXIZ$ \\ \hline \hline
\multirow{6}{*}{0101} & $IYIXI$ & $XYIII$ & \multirow{6}{*}{0110} & $IXXII$ & $IXZII$ \\
& $\textbf{ZXIXZ}$ & $\textbf{ZIIXI}$ && $\textbf{IXYZY}$ & $\textbf{IIXXI}$ \\
& $IYIYX$ & $IZIZI$ && $-IYXZX$ & $IZIYI$ \\
& $IXIYI$ & $IIIYX$ &&  $IXXZI$ & $IIIZX$ \\
& $-ZXIXX$ & $YIIIY$ &&  $IYXIY$ & $IYIIY$	\\
& $IXIXZ$ & $IXIIZ$ && $-IXYZX$ & $IIYIZ$ \\ \hline \hline
\multirow{6}{*}{0111} & $-ZYXXZ$ & $YZIII$ & \multirow{6}{*}{1000} & $XZZZY$ & $IYXII$ \\
& $\textbf{IXXXI}$ & $\textbf{IXIXI}$ && $\textbf{XIZIZ}$ & $\textbf{IIYYI}$ \\
& $IXXYX$ & $XIIYI$ &&  $XIIZX$ & $IXIZI$ \\
& $IYXXX$ & $IYIIX$ &&  $XIIIX$ & $IIZIX$ \\
& $-IXXYZ$ & $IIIZY$ &&  $XIIIY$ & $IIIXY$	\\
& $ZXXXI$ & $ZIIIZ$ &&  $XZIIZ$ & $IZIIZ$
\end{tabular}
\caption{Table of double-qubit errors.  The error operators in bold represent a single encoding.  The table represents a total of six different encodings where each encoding utilizes sixteen distinct error operators and their corresponding distinct syndromes (continued).}
\label{tbl:two-qubit_errors_perfectcode_1}
\end{table}
\begin{table}\centering
\begin{tabular}[c]{|c|c|c||c|c|c|}
Syndrome $i$ & $\CE_{i} \otimes \CO_{i}$ & Encoded Error & Syndrome $i$& $\CE_{i} \otimes \CO_{i}$ & Encoded Error \\ \hline \hline
\multirow{6}{*}{1001} & $-YIZYX$ & $YIYII$ & \multirow{6}{*}{1010} & $YIXIZ$ & $YXIII$ \\
& $\textbf{-YIIYY}$ & $\textbf{ZIIYI}$ && $\textbf{-XZYZY}$ & $\textbf{IYYII}$ \\
& $XIZXZ$ & $IIXZI$ &&  $XIXII$ & $XIZII$ \\
& $XIIXX$ & $IIIXX$ &&  $XZXII$ & $IZIXI$ \\
& $XIIXY$ & $IIZIY$ && $XIYIZ$ & $IIXYI$	\\
& $XIIXZ$ & $XIIIZ$ && $-YIXZY$ & $ZIIZI$ \\ \hline \hline
\multirow{6}{*}{1011} & $-YZXXZ$ & $ZYIII$ & \multirow{6}{*}{1100} & $XXIII$ & $XXIII$ \\
& $\textbf{-YIYYX}$ & $\textbf{YIXII}$ && $\textbf{-XYZZY}$ & $\textbf{IZXII}$ \\
& $XZXXI$ & $IZZII$ && $YXIIZ$ & $YIZII$ \\
& $XIXXI$ & $XIIXI$ && $-YXIIY$ & $ZIIIX$ \\
& $XIXYX$ & $IXIYI$ && $XXZZI$ & $IIYIY$	\\
& $-XIYXZ$ & $IIYZI$ && $XYIIZ$ & $IYIIZ$ \\ \hline \hline
\multirow{6}{*}{1101} & $XXZYY$ & $XIXII$ & \multirow{6}{*}{1110} & $XYYZY$ & $IZYII$ \\
& $\textbf{YXIXZ}$ & $\textbf{YIIXI}$ && $\textbf{XYXII}$ & $\textbf{IYIXI}$ \\
& $XXIYX$ & $IIZYI$ &&  $XXXZX$ & $IIZZI$ \\
& $XXZYZ$ & $IIYIX$ && $XXXIX$ & $IXIIX$ \\
& $YXIXX$ & $ZIIIY$ && $XXYZI$ & $IIXIY$	\\
& $XXIYY$ & $IIIZZ$ && $-XXXZY$ & $IIIYZ$ \\ \hline \hline
\multirow{6}{*}{1111} & $YYXXZ$ & $ZZIII$ & \multirow{6}{*}{0000} & $IIIII$ & $IIIII$ \\
& $\textbf{-XXYYY}$	 & $\textbf{XIYII}$ && $\textbf{IIIII}$ & $\textbf{IIIII}$  \\	
& $XYXXI$	 & $IYZII$ && $IIIII$ & $IIIII$ \\	
& $XXYYZ$	 & $IIXIX$ && $IIIII$ & $IIIII$ 	\\
& $XXXXY$	 & $IXIIY$ && $IIIII$ & $IIIII$ 	\\
& $YXXXI$	 & $YIIIZ$ && $IIIII$ & $IIIII$ 	\\ \hline \hline
\end{tabular}
\caption{Table of double-qubit errors.  The error operators in bold represent a single encoding.  The table represents a total of six different encodings where each encoding utilizes sixteen distinct error operators and their corresponding distinct syndromes.}
\label{tbl:two-qubit_errors_perfectcode}
\end{table}

\chapter{Conclusion and Future Directions}
\label{chap:conclusion}
\begin{saying}
When I examine myself and my
methods of thought, I come to the
conclusion that the gift of fantasy
has meant more to me than any
talent for abstract, positive
thinking.\\
---\textit{Albert Einstein}
\end{saying}
In this thesis we have made two contributions to quantum information science.  We augmented the theory of quantum error-correcting codes by bridging the gap between the perfect code and the Steane code.  We provided a new proof technique on how to show that the Steane code is the smallest CSS code that can correct an arbitrary single-qubit error.  We also showed how to convert the six-qubit code into a subsystem code.  We hope that in the future this code will serve as a good test-bed to protect quantum computers from the debilitating effects of decoherence.

Our second major contribution to QIS is the development of the theory of quantum steganography.  We introduced a broad theory of quantum steganography but in the future we would like to explore the kinds of quantum error-correcting codes and encodings that will give good rates for steganographic communication between Alice and Bob under different noise models such as, depolarizing noise, dephasing noise, and general noise operators. In the current model we assume that the eavesdropper (Eve) is limited to only observing the channel and performing non-demolition measurements of the quantum codewords.

While the above is a useful starting point, one can think about other equally important and valid scenarios. It may be reasonable to consider a couple of models of Eve: the maximally powerful, limited by physics and information theory, and the more realistic, limited by computational power and finite information. We can set up the problem as containing only three parties, Alice, Bob, and Eve. But in a realistic situation, Eve may simultaneously be monitoring communications between a large number of parties --- for instance, an entire network. In this situation, Eve may have no reason to be suspicious of Alice and Bob more than anyone else. The desire to keep this network running efficiently may lead to some limitations on what Eve will and/or can do. Eve must allow communication, or she could just shut the channel down. If she intercepts and demands authentication for every communication, this makes quantum communication impossible. Therefore, she can only intercept and demand authentication for some fraction of the messages, chosen at random. Without blocking the message, though, she could still check the error syndromes. She could do this every time, or some fraction of the time, which could be the same fraction with which she demands authentication, or a different fraction. She could just check the syndromes, or she could actually perform error-correction. There may be other set-ups as well in which there might not be an Eve per se. Suppose that there are two ``innocent'' communicators, April and Bill, and two ``secret'' communicators, Alice and Bob. Alice and Bob want to piggyback secret communication on top of April and Bill's messages. So after April sends her codeword, Alice intercepts it, adds in secret information in a way that will look like noise and sends it on. Bob could get it either before Bill does (in which case he might extract the information and try to erase all evidence that it was there), or afterwards (in which case he might get it after Bill has performed error-correction and discarded the ancillas). Another scenario is when Eve is not necessarily an adversary, but a ``helpful'' network. For instance, there may be quantum repeater stations that correct and forward the quantum information. If Alice and Bob know how the quantum repeaters work, might they be able to sneak through some extra information that was not removed in the correction? What is the capacity of quantum steganographic protocols between three parties, Alice, Bob, and Charlie? We do not know the answer to these questions and we believe that they are worth exploring.

We would also like to explore where quantum steganography fits in with quantum Shannon theory. Could a combination of some primitive quantum information-theoretic protocols in quantum Shannon theory when combined with steganography give rise to novel and interesting new insights? Can one treat quantum steganography as a new quantum information science primitive?  Quantum steganography is rich in open problems, and there is clearly a need to develop this theory further because once we do have commercially available quantum networks, there will be a need for private communication.  We hope that in time this theory will become a subfield of quantum Shannon theory.  We also hope that in the near future experimentalists will be able to test our protocols in the laboratory.




\end{document}